\documentclass[camera,letterpaper,nomarginnotes,nonarrowgutter]{jpaper}

\usepackage{booktabs}
\usepackage{threeparttable}
\usepackage[numbers,sort&compress,square]{natbib}
\usepackage{graphicx}
\usepackage{xspace}
\usepackage{listings}
\usepackage{amsmath}
\usepackage{amssymb}
\usepackage{microtype}
\usepackage{multirow}
\usepackage{array}
\usepackage{pifont}
\usepackage{mathtools}
\usepackage{amsthm}
\usepackage{tikz}
\usetikzlibrary{arrows.meta,positioning,automata}
\usepackage{fancyhdr}
\usepackage{hyperref}
\usepackage{enumitem}
\setlist[itemize]{leftmargin=*,nosep}
\setlist[enumerate]{leftmargin=*,nosep}

\hypersetup{colorlinks=true,citecolor=blue,urlcolor=blue,linkcolor=blue,
            bookmarks=true,breaklinks=true}

\fancypagestyle{firstpage}{\fancyhf{}\fancyfoot[C]{\thepage}}

\theoremstyle{plain}
\newtheorem{theorem}{Theorem}[section]

\theoremstyle{definition}
\newtheorem{definition}{Definition}[section]

\usepackage{booktabs}
  \usepackage{amssymb}
  \usepackage{xcolor}
  \usepackage{pifont}              % for \ding

\definecolor{editcolor}{HTML}{B00020}
\newif\ifshowedits \showeditsfalse
\newcommand{\edit}[1]{\ifshowedits\textcolor{editcolor}{#1}\else#1\fi}

\definecolor{mkcolor}{HTML}{B00020}
\newif\ifshowmk \showmkfalse
\newcommand{\mk}[1]{\ifshowmk\textcolor{blue}{#1}\else#1\fi}

\definecolor{arevcolor}{HTML}{6A0DAD}
\newif\ifshowarev \showarevfalse
\newcommand{\arev}[1]{\ifshowarev\textcolor{black}{#1}\else#1\fi}
\newenvironment{arevblock}{\ifshowarev\color{black}\fi}{}

\definecolor{brevcolor}{HTML}{008B8B}
\newif\ifshowbrev \showbrevfalse           % inline \brev{} renders black
\newif\ifshowbrevblock \showbrevblockfalse % the brevblock environment renders black too
\newcommand{\brev}[1]{\ifshowbrev\textcolor{brevcolor}{#1}\else#1\fi}
\newenvironment{brevblock}{\ifshowbrevblock\color{brevcolor}\fi}{}

\definecolor{crevcolor}{HTML}{C2185B}
\newif\ifshowcrev \showcrevfalse
\newcommand{\crev}[1]{\ifshowcrev\textcolor{crevcolor}{#1}\else#1\fi}

\definecolor{backcolor}{HTML}{E65100}
\newif\ifshowback \showbackfalse
\newcommand{\back}[1]{\ifshowback\textcolor{backcolor}{#1}\else#1\fi}

\definecolor{realcolor}{HTML}{1565C0}
\newif\ifshowreal \showrealfalse            % on-silicon / real-system additions, shown in blue

\definecolor{artifactcolor}{HTML}{2E7D32}
\newif\ifshowartifact \showartifactfalse     % artifact/reproducibility pointers, shown in green

\definecolor{fixcolor}{HTML}{AD1457}
\newif\ifshowfix \showfixfalse               % edits addressing margin comments, shown in magenta
\newcommand{\fix}[1]{\ifshowfix\textcolor{fixcolor}{#1}\else#1\fi}

\DeclareRobustCommand*\circled[1]{\tikz[baseline=(char.base)]{\node[shape=circle,fill,inner sep=0.4pt] (char) {\fontsize{8.5pt}{8.5pt}\selectfont\textcolor{white}{#1}};}}

\definecolor{nrevcolor}{HTML}{00838F}
\newif\ifshownrev \shownrevfalse                % restructuring: failure surfaces moved into the motivation, rendered black
\newcommand{\nrev}[1]{\ifshownrev\textcolor{nrevcolor}{#1}\else#1\fi}
\newenvironment{nrevblock}{\ifshownrev\color{nrevcolor}\fi}{}

\definecolor{prevcolor}{HTML}{6A1B9A}
\newif\ifshowprev \showprevfalse                 % proofreading pass, shown in purple
\newcommand{\prev}[1]{\ifshowprev\textcolor{prevcolor}{#1}\else#1\fi}
\newenvironment{prevblock}{\ifshowprev\color{prevcolor}\fi}{}

\definecolor{rrevcolor}{HTML}{795548}
\newif\ifshowrrev \showrrevfalse           % 2026-09-05 references pass: bib fixes, added citations, and the wording they forced; shown in brown
\newcommand{\rrev}[1]{\ifshowrrev\textcolor{rrevcolor}{#1}\else#1\fi}

\definecolor{pendingcolor}{HTML}{EF6C00}
\newif\ifshowpending \showpendingfalse       % NOT-yet-done results (e.g. board runs to complete); shown bracketed in orange, dropped entirely when hidden

\definecolor{numupdcolor}{HTML}{5D4037}
\newif\ifshownumupd \shownumupdfalse         % statistics recomputed from the consolidated single-branch artifact; shown in brown
\newcommand{\numupd}[1]{\ifshownumupd\textcolor{numupdcolor}{#1}\else#1\fi}

\definecolor{srevcolor}{HTML}{827717}
\newif\ifshowsrev \showsrevfalse            % 2026-09-08 Sec. 5/Sec. 6 pass after the workflow restructure; shown in olive
\newcommand{\srev}[1]{\ifshowsrev\textcolor{srevcolor}{#1}\else#1\fi}

\definecolor{vrevcolor}{HTML}{283593}
\newif\ifshowvrev \showvrevfalse            % 2026-09-08 correctness pass on Sec. 5 and Sec. 6; shown in indigo
\newcommand{\vrev}[1]{\ifshowvrev\textcolor{vrevcolor}{#1}\else#1\fi}

\definecolor{drevcolor}{HTML}{004D40}
\newif\ifshowdrev \showdrevfalse            % 2026-09-09 line-count double-count fix; shown in dark teal
\newcommand{\drev}[1]{\ifshowdrev\textcolor{drevcolor}{#1}\else#1\fi}

\newcommand{\goodbg}{green!14}
\newcommand{\warnbg}{yellow!26}
\newcommand{\badbg}{red!11}
\newcommand{\tinybox}{\rule{0.38em}{0.38em}}
\newcommand{\gmark}{\cellcolor{\goodbg}\textcolor{green!42!black}{\ding{51}}}
\newcommand{\ymark}{\cellcolor{\warnbg}\textcolor{olive!75!black}{\tinybox}}
\newcommand{\rmark}{\cellcolor{\badbg}\textcolor{red!62!black}{\ding{55}}}

\newcommand{\titlenamed}{CertiFlash}
\newcommand{\titlename}{CertiFlash\xspace}
\xspaceaddexceptions{'}
\newcommand{\arbiter}{\titlename}
\newcommand{\arbiterd}{\titlenamed}

\newcommand{\coq}{Rocq\xspace}

\title{\arbiterd: A Formal Verification Framework \\ for Flash Translation Layers in Computational Solid State Drives}

\author{
    Harshita Gupta\textsuperscript{$\dagger$} \hspace{0.5em}
    Mayank Kabra\textsuperscript{$\dagger$} \hspace{0.5em}
    Rakesh Nadig\textsuperscript{$\dagger$} \hspace{0.5em}
    Nika Mansouri Ghiasi\textsuperscript{$\dagger$} \hspace{0.5em}
    Sahand Divsalar\textsuperscript{$\dagger$} \vspace{0.2em} \\
    Fatma Nisa Bostancı\textsuperscript{$\dagger$} \hspace{0.5em}
    Ataberk Olgun\textsuperscript{$\dagger$} \hspace{0.5em}
    Konstantinos Kanellopoulos\textsuperscript{$\dagger$} \hspace{0.5em}
    Jisung Park\textsuperscript{$\ddagger$} \vspace{0.2em} \\
    Haiyu Mao\textsuperscript{$\mathsection$} \hspace{0.5em}
    Abdullah Giray Yağlıkçı\textsuperscript{$\P$} \hspace{0.5em}
    Mohammad Sadrosadati\textsuperscript{$\dagger$} \hspace{0.5em}
    Onur Mutlu\textsuperscript{$\dagger$} \vspace{0.5em} \\
    \textsuperscript{$\dagger$}\emph{ETH Z\"urich} \hspace{0.5em}
    \textsuperscript{$\ddagger$}\emph{POSTECH} \hspace{0.5em}
    \textsuperscript{$\mathsection$}\emph{King's College London} \hspace{0.5em}
    \textsuperscript{$\P$}\emph{CISPA}
}

\begin{document}
\maketitle
\thispagestyle{firstpage}

\begin{abstract}
\setlength{\parskip}{0pt}
\crev{\back{In traditional \mk{compute}-centric systems,} \mk{many modern} data-intensive \mk{applications} suffer from significant data movement overhead because they move large amounts of
% \fix{low-reuse}
data from storage to the \mk{compute unit}. Storage-centric computing reduces this overhead by moving computation near or inside solid-state drives (SSDs).} \crev{\back{To enable it, SSD policies (e.g., address translation, garbage collection, and data placement) \mk{require modification\rrev{.}}} 
% to build application-specific SSDs.} 
\rrev{These} policies are part of the Flash Translation Layer (FTL), the SSD's firmware. 
\fix{\mk{However,} modifying them is complex and error-prone, \rrev{and} a single bug can violate the device's correctness, isolation, and ownership guarantees\rrev{.}}} Even a functionally correct FTL can leak data between tenants sharing the device, drop integrity tags, or assign a flash block to the wrong tenant. \fix{\mk{This is because} FTL logic has direct access to \mk{many} 
% shared, 
security-critical device components, \mk{\rrev{e.g., the host-interface queue and the metadata in internal DRAM}}.} \mk{We show that a \drev{faulty} FTL can corrupt the device state at five surfaces inside the SSD\rrev{, and we demonstrate those failures on a DaisyPlus OpenSSD platform}.} \mk{\rrev{Prior work verifies} individual FTL designs, 
% for \mk{\emph{only}} functional correctness. Prior proofs 
\rrev{but faces} two limitations. (1) \rrev{These works establish} only functional correctness, so a modified FTL can violate isolation, integrity, and ownership and still pass verification, providing no security guarantee.
% for the modified design. 
(2) \rrev{They are tied} to the state and operations of a single FTL design, so every FTL modification requires regenerating every proof from scratch, significantly increasing the verification effort.}

\arev{We propose} \arbiter, a formal verification framework \crev{for FTLs, mechanized} in the \coq proof assistant\back{, that gives designers a machine-checked proof \rrev{of} both security and correctness}. \arbiter models an FTL as a deterministic state machine \fix{with a single global invariant over mapping, isolation, integrity, ownership, and allocation\drev{, grouped by primary purpose}}. \back{We prove \rrev{once, over a general FTL model,} that (i) every FTL operation 
% (including garbage collection and wear leveling) 
preserves the global invariant
% , including the garbage-collection and wear-leveling operations that run without a host request, 
and (ii) the model refines an idealized block device. %
% For a new design, \arbiter inherits both proofs by discharging five hypotheses about its own operations through a \coq interface.
For a new design, \vrev{a designer discharges five hypotheses about its own operations through a \coq interface instead of redoing either proof}.}
% We validate \arbiter in three ways. 
% A designer describes their own FTL through a small \coq interface and inherits both proofs.
% The designer inherits the framework's proofs and writes only what the modification adds, at a cost set by how much of the global invariant it disturbs. 
% \fix{\rrev{Across four case studies, a designer adds 14 to 3{,}232 lines against a 14{,}910-line framework, \mk{significantly reducing the verification effort.}}}
\fix{\rrev{Across four case studies, a designer adds \numupd{27} to 3{,}231 lines against a \drev{16{,}489}-line framework, \mk{significantly reducing the verification effort.}}} The source code of \arbiter is freely available at: \url{https://github.com/CMU-SAFARI/CertiFlash}.
% The development contains no admitted lemmas and no axioms, so the trusted base is the \coq kernel and standard library alone.
\end{abstract}

\section{Introduction}
\label{sec:intro}

\crev{Modern data-intensive applications suffer from significant data movement overhead because they access large amounts of low-reuse data~\cite{mutlu.imw13, kanev.isca15, wang.micro2016, mckee2004reflections, mutlu.superfri15, Park_MICRO2022, hajinazar2021simdram, gu2016biscuit, mutlu2022modern, boroumand2018google, boroumand2021google, ahn2015scalable, ahn2015pim, mutlu2019processing, mutlu2019enabling, oliveira2021damov, ghiasi2022alp, seshadri-micro-2017, Gao_MICRO2021, barbalace_blockndp_2020, augusta2015jafar, boroumand2019conda, fernandez2020natsa, singh2019napel, gao2016hrl, lee2020smartssd, singh2021fpga, medal2019, liang-fpl-2019, Mansouri-Ghiasi_ASPLOS2022, ghiasi2026sage, ghiasi2026grains, Ghiasi2024MegIS, oliveira2024mimdram, nider2020processing, hsieh.isca16, park2024attacc,chen2022offload,li2018cisc, li2023optimizing,maity2025unguided, olivier2019hexo,wei2022pimprof, weiner2022tmo, yang2023lambda, lincoln-hpca, jang2025inf, pan2024instattention, jaliminchecs, kang2024isp}. This data moves all the way from the solid-state drive (SSD) through main memory to the processing units for its first use, and is rarely used again.} Many prior works
(e.g.,~\cite{seshadri-osdi-2014,acharya-asplos-1998,Kabra2025CipherMatch,wang-eurosys-2019,kim-fast-2021,kang-msst-2013,torabzadehkashi-pdp-2019,keeton-sigmod-1998,koo-micro-2017,tiwari-fast-2013,tiwari-hotpower-2012,boboila-msst-2012,bae-cikm-2013,torabzadehkashi-ipdpsw-2018,pei-tos-2019,do-sigmod-2013,kim-infosci-2016,riedel-computer-2001,riedel-vldb-1998,liang-atc-2019,cho-wondp-2013,jun2015bluedbm,lee2020smartssd,ajdari-hpca-2019,liang-fpl-2019,jun-hpec-2016,kang-tc-2021,kim-sigops-2020,Lee_ISCA2022,li2023ecssd,ruan2019insider,wang2016ssd1,jeong-tpds-2019,mao2012cache,gouk2024dockerssd,Ghiasi2024MegIS,Mansouri-Ghiasi_ASPLOS2022,yavits2021giraf,Kim_HPCA2023,lim-icce-2021,narasimhamurthy2019sage,jun-isca-2018,fakhry2023review,gu2016biscuit,yang2023lambda,jo2016yoursql,Chen2025REIS,Li_ATC2021,wang2016ssd,mahapatra2025rag,pan2024instattention,wang2024beacongnn,Yu2024CambriconLLM,Park_MICRO2022,Gao_MICRO2021,Chen_MICRO2024,kim2025crossbit,wong2024tcam,wong2025anvil,chen2024search,choi2020flash,chun2022pif,lee2025aif,wang2018three,ghiasi2026sage,ghiasi2026grains,ghiasi2026enabling}) demonstrate that moving computation near or inside \rrev{the} SSD~\cite{samsung-980pro, adatasu630,intelqlc,intels4510, intelp4610, nadig2023venice, micheloni2010inside, micheloni2013inside, inteloptane, samsungmlc, samsung2017znand} (\S\ref{sec:prelim:arch}), i.e., \emph{storage-centric computing}, accelerates data-intensive applications such as genome analytics~\crev{(e.g.,}~\cite{ghiasi2026grains, ghiasi2026enabling, Soysal2025MARS, Mansouri-Ghiasi_ASPLOS2022, Ghiasi2024MegIS, ghiasi2026sage}\crev{)}, machine learning~\crev{(e.g.,}~\cite{Yu2024CambriconLLM,Kim_HPCA2023, Chen2025REIS}\crev{)}, graph processing~\crev{(e.g.,}~\cite{Matam_ISCA2019,Li_ATC2021,Liang_DAC2022,Lee_ISCA2022}\crev{)}, and database analytics~\crev{(e.g.,}~\cite{Duffy_PVLDB2023, Park_ASPLOS2025, Bisson_IPCCC2018}\crev{)} by reducing data movement between storage, main memory, and processing units.

\fix{Storage-centric computing~\cite{nadig2026conduit,mutlu2022modern,mutlu2024memory,mutlu2019processing,mutlu2025memory,mutlu2019enabling} requires customizing SSD policies (\S\ref{sec:prelim:ftl}), e.g., logical-to-physical address translation or data placement, which are present in the Flash Translation Layer (FTL). The vendor modifies the FTL to build an application-specific SSD (\S\ref{sec:isp_analysis}), and the modified FTL runs on the SSD's controller. 
% \rrev{Modifying the FTL is not a localized change. Prior storage-centric designs change different subsets of eight FTL components, and no two change the same set (\S\ref{sec:isp_analysis}).} 
Custom FTL logic has direct access to shared, security-critical device components, i.e., the host-interface queue, internal DRAM, the embedded compute units, the flash controller, and the NAND flash chips, which we treat as five failure surfaces~(\S\ref{sec:security-model:surfaces}). These surfaces are commonly shared in a scenario where multiple tenants~\rrev{\cite{tavakkol2018flin, huang2017flashblox, kwon2020dc, min2021isolating, min2023multi, kim2015ops, jaliminche2023enabling, gonzalez2017multi, sun2025fleetio}} share a common SSD\rrev{. NVMe namespaces and per-namespace keys define the isolation boundary among multiple tenants~\cite{nvmenamespaces, nvme_base_spec}, but the FTL enforces them and must implement them correctly.} 
\rrev{An FTL is complex, so modifying it correctly is error-prone} for any FTL
% implementer.
\edit{designer}\fix{, i.e., the party that supplies or extends the FTL}.
% \footnote{Throughout the paper we use \emph{tenant} for any end-user application or cloud-provider control plane that supplies a custom FTL, and \emph{vendor} for the SSD manufacturer; we say \emph{FTL implementer} to refer to either.} 
Production\mk{-scale} FTLs maintain many runtime data structures, e.g., a logical-to-physical (L2P) mapping table that maps host-supplied addresses \crev{to physical flash pages}, free-block pools, write pointers, and (in modern designs) demand-cached translation pages~\cite{gupta2009dftl} and per-tenant ownership metadata.~\edit{These structures are \mk{continuously} updated, \mk{both} by \crev{(i)} host-visible read and write requests, \mk{and} \crev{(ii)} internal maintenance operations such as garbage collection and wear leveling~\crev{\cite{agrawal2008design, xie2014adaptive, jung2012taking, shahidi2016exploring, wu2016gcar, choi2018parallelizing, chang2007wear, yang2014garbage}}, which silently relocate live data and update persistent metadata without explicit host invocation.} The risk is not limited to an FTL being functionally incorrect, but in a multi-tenant scenario, even a functionally correct FTL can expose one tenant's data to another, drop an integrity tag the host relies on, or assign a flash block to the wrong tenant.} 
~\fix{We observe that a single bad state, e.g.,} a stale L2P entry, a duplicate in the free-block list, \fix{or} an omitted integrity tag, can corrupt address translation, leak data across tenants, or invalidate data-integrity guarantees~\edit{the host}
% that every layer above the SSD 
relies on.
% FTL transitions are also complex. Beyond host-visible read and write requests, the FTL performs internal maintenance operations, such as garbage collection and wear leveling, that relocate live data and update persistent metadata without explicit host invocation. 

\rrev{\textbf{Real-System Attack Demonstration.} We demonstrate these risks on a real SSD, the DaisyPlus OpenSSD~\cite{daisyplus_openssd}, by modifying its FTL~\drev{(\S\ref{sec:attacks:board})}. The attacks fall into two settings. First, in a \emph{single-tenant} setting, a modified FTL can, e.g., corrupt the tenant's own data, breaking data integrity. Second, in a \emph{multi-tenant} setting, it can, e.g., expose one tenant's data to another or assign a flash block to the wrong tenant, breaking tenant isolation and block ownership. Each attack modifies only the FTL and issues ordinary host I/O, with no privileged access. The corruption can persist silently in FTL metadata for many operations before any host-visible failure, so neither the host nor the victim observes the violation. Unverified FTL modifications therefore pose a fundamental security risk.}
% bugs are generally hard to expose because they often require specific sequences of writes, erasures, remappings, and background-maintenance operations before they affect externally visible SSD behavior.

\textbf{Key Problem.}
% Conventional SSD validation methods, such as unit testing, trace replay~\cite{agrawal2008design, tavakkol2018mqsim}, and fault injection~\cite{zheng2013power} on emulated or prototyped devices~\back{\cite{li2018femu, kim2023nvmevirt, kwak2018cosmos}}, cover only part of the reachable \mk{FTL states, i.e., the possible values of the FTL's internal data structures, such as the logical-to-physical mapping table, the free-block and allocation lists, the per-block valid-page and erase counts.} Hence, these methods can miss~\edit{faults}
% Conventional SSD validation methods, such as unit testing, trace replay~\cite{agrawal2008design, tavakkol2018mqsim}, and fault injection~\cite{zheng2013power} on emulated or prototyped devices~\back{\cite{li2018femu, kim2023nvmevirt, kwak2018cosmos}}, cover only part of the reachable \mk{FTL states, i.e., the possible values of the FTL's internal data structures, such as the logical-to-physical mapping table, the free-block and allocation lists, \rrev{and} the per-block valid-page and erase counts.} Hence, these methods can miss~\edit{faults}
Conventional SSD validation methods, such as unit testing, trace replay~\cite{agrawal2008design, tavakkol2018mqsim}, and fault injection~\cite{zheng2013power} on emulated or prototyped devices~\back{\cite{li2018femu, kim2023nvmevirt, kwak2018cosmos}}, cover only part of the reachable \mk{FTL states, i.e., \rrev{the combinations of values its internal data structures can hold, namely} the logical-to-physical mapping table, the free-block and allocation lists, \rrev{and} the per-block valid-page and erase counts.} Hence, these methods can miss~\edit{faults}
% bugs
that require specific sequences of host requests and maintenance transitions~\cite{chang2020determinizing, zheng2013power}. 
\fix{Each customization changes which states the FTL can reach, so its validation cannot be reused and must be redone for the new design.}
\fix{Formal verification closes testing's coverage gap, because a proof holds over every reachable state. Prior work verifies FTLs formally in two ways. (1) SCFTL~\cite{chang2020determinizing} and Flashix~\cite{bodenmuller2021flashix} \rrev{each verify one specific FTL design.} (2) Qiao et al.~\cite{qiao2019formal} instead build a framework that applies to any FTL \rrev{and demonstrate it on BAST}~\cite{kim2002space}\rrev{.}} 
Unfortunately, these approaches have two limitations. (i) They establish \emph{only} functional correctness. \rrev{None of them states an isolation, integrity, or ownership property}, so a functionally correct FTL passes while leaking data across tenants. (ii) Their proofs do not carry over to a new design\rrev{. The invariant is supplied per design, so a new FTL writes its own and}
% and prior work does not specify what the hypotheses must preserve
% so each new design supplies its own properties and re-proves every hypothesis. An effort that verifies every customization independently re-establishes every property of the unchanged parts, so its cost scales with the size of the FTL rather than with the size of the change.
%  so each new design supplies its own properties and \rrev{proves them itself}. An effort that verifies every customization independently re-establishes every property of the unchanged parts, so its cost scales with the size of the FTL rather than with the size of the change.
\rrev{proves it}. An effort that verifies every customization independently re-establishes every property of the unchanged parts, so its cost scales with the size of the FTL rather than with the size of the change.
% Their proofs are written from scratch for that design and do not transfer to a different FTL.

\fix{\mk{\textbf{Our goal} is to design a \rrev{reusable} formal verification framework that enables FTL designers to develop a machine-checked proof of the correctness and security of their customized FTL.} To this end, we propose \arbiter, a formal verification framework implemented in the \coq proof assistant~\cite{coq}. 
% \arbiter proves two properties of a designer's FTL. (1) The FTL \emph{preserves} a \emph{global invariant}, i.e., \mk{a conjunction of multiple correctness and security properties} over the FTL's state that every operation must keep true. (2) The FTL \emph{refines} an idealized block device\footnote{An FTL refines an idealized block device when every read the host issues returns the data that a simple, correct block device would return, i.e., the FTL's host-visible behavior is indistinguishable from that of the idealized device} (\S\ref{sec:model:invariants}).
}
% mechanized 
% \footnote{\crev{An operation \emph{preserves} a property when the property holds after it whenever it held before.}}
% \footnote{\arev{These are properties of the FTL's internal state and its host-visible behavior: the global invariant holds at every reachable state, and every valid execution matches that of an idealized block device (\S\ref{sec:model:simulation}).}}.

\textbf{\fix{Key Mechanism.}} \fix{\arbiter defines a general FTL model that captures what every FTL has in common, i.e., the address mapping, page states, ownership labels, free blocks it keeps, and the six operations it performs on them, i.e., read, write, invalidate, integrity-tag update, garbage collection, and wear leveling. \arbiter proves the two properties, i.e., (1) invariant preservation and (2) refinement of the idealized block device, once over this model.} 
\fix{\arbiter provides \fix{two ways to establish these two properties for a given FTL.}}
(i)~\emph{Direct verification}, in which the~\edit{FTL designer} writes their FTL as a \coq model and proves end-to-end that the operations preserve all the required properties. \mk{\rrev{\drev{\arbiter carries out this mode once for its own model. A design repeats it only if it changes the global invariant.}}} (ii)~\emph{Modular verification}, in which the designer describes their FTL through a \coq interface 
\vrev{and discharges five hypotheses instead of redoing the existing proofs}.   
\fix{The interface asks for three inputs \srev{(\S\ref{sec:verify:modular})}: (a) the FTL's state and operations; (b) a refinement map that projects the FTL's state onto \arbiter's model; and (c) proofs of five hypotheses about the FTL's own operations, namely that they preserve the global invariant and that \drev{a write neither loses nor corrupts data}.} \mk{\rrev{\arbiter's proofs then carry over unchanged, so the designer proves only what the modification changes.}}  
% The two modes are not mutually exclusive, and a designer can mix them across operations within a single FTL. All theorems are machine-checked. 
% The trusted computing base is the \coq kernel and the standard library. The development contains no admitted lemmas, no classical reasoning, and \arev{no axioms (\S\ref{sec:security-model:scope})}.

\textbf{Key Results.} We evaluate \arbiter on four case studies spanning the ways a new design can differ from a verified one (\S\ref{sec:evaluation}): an orthogonal extension, DFTL~\cite{gupta2009dftl}, the in-storage-processing system MegIS~\cite{Ghiasi2024MegIS}, and deduplication~\cite{chen2011caftl}. \fix{Across the four case studies, a designer adds \numupd{27} to 3{,}231 lines \rrev{against a \drev{16{,}489}-line framework}.}

\noindent \textbf{Contributions.} \arev{We make the following contributions:}

\begin{itemize}
\item \mk{\rrev{We identify five failure surfaces} that custom FTL logic can violate, and we \rrev{demonstrate} attacks on each surface\rrev{, using} a real SSD, the DaisyPlus OpenSSD platform~\cite{daisyplus_openssd} \drev{(\S\ref{sec:attacks:board})}.} 
% Every attack leaves the device functionally correct, so conventional SSD validation does not detect it.

\item \mk{We propose \arbiter, the first framework that \rrev{gives a modified FTL a machine-checked proof of tenant isolation, data integrity, and block ownership}, not only functional correctness. \arbiter states mapping correctness, tenant isolation, data integrity, block ownership, and block allocation as a single global invariant (Definition~\ref{def:ftl-invariant}) and proves that 1) every FTL operation, including garbage collection and wear leveling, preserves it (Theorem~\ref{thm:preservation}), and 2) reads return the data stored at the corresponding address in the
idealized block device (Theorem~\ref{thm:refinement}).}

\item \mk{We make verification reusable across FTL designs through two modes (§\ref{sec:formal-model}): direct verification, in which a designer proves their implementation correct end-to-end, and modular verification, in which a designer describes their FTL through a small \coq interface %
and \vrev{discharges five hypotheses instead of redoing \arbiter's proofs}. To our knowledge, this is the first work to make FTL security proofs reusable across designs.}

\item \mk{We evaluate \arbiter on four case studies and show that \rrev{verification effort is significantly reduced by \arbiter's modular verification}.}

\end{itemize}

% All theorems are machine-checked. The trusted computing base is the \coq kernel and the standard library; the development contains no axioms, no admitted lemmas, and no classical reasoning, so every proof is constructive.

\section{Background}
\label{sec:prelim}

\edit{We describe (1) the solid-state drive (SSD) architecture, the Flash Translation Layer (FTL), and the NAND-flash hardware constraints (\S\ref{sec:prelim:ssd-ftl}); 
\back{(2) \fix{storage-centric \rrev{computing, which drives FTL customization}}~(\S\ref{sec:prelim:deploy})\rrev{; and (3) how prior storage-centric designs modify the FTL~(\S\ref{sec:isp_analysis}).}}} 

\subsection{SSD Organization and FTL}
\label{sec:prelim:ssd-ftl}

\subsubsection{SSD Architecture}
\label{sec:prelim:arch}

A modern SSD comprises five key components arranged along the path from host I/O to physical flash storage~\fix{\cite{nadig2026harmonia, tavakkol2018mqsim, tavakkol2018flin, agrawal2008design}}~\rrev{(as shown in Figure~\ref{fig:ssd})}.~\back{(i)}~The \emph{host-interface layer} (HIL)~\circled{1} accepts incoming I/O requests over an industry-standard protocol (e.g., NVMe over PCIe~\back{\cite{nvme_base_spec, eshghi2012ssd, snai_sdc2021}}), enqueues them \fix{in its request queue} for the \back{SSD} controller, and delivers completion responses to the host. \back{(ii)}~The SSD controller~\circled{2} hosts \emph{\fix{embedded compute} units~\circled{3}}, \fix{i.e.,} CPU cores (e.g.,~\cite{cortex_r4}) \back{or} accelerator engines (e.g.,~\cite{cortex_r4, tavakkol2018mqsim, agrawal2008design, Ghiasi2024MegIS, park2006high}), that execute the SSD's firmware, which includes the FTL (\S\ref{sec:prelim:ftl})~\circled{4}. \back{(iii)}~The controller has direct access to an \emph{internal DRAM}~\circled{5} that caches FTL metadata and recently-used data \back{read from flash}~\cite{agrawal2008design, bang2011memory}. \back{(iv)}~The \emph{flash controller}~\circled{6} translates the FTL's logical operations into low-level NAND commands and issues them to the NAND flash 
chips~\cite{agrawal2008design, cai-date-2012, cai-iccd-2012, cai-inteltechj-2013, cai-sigmetrics-2014, cai-dsn-2015, cai-hpca-2017, luo-sigmetrics-2018, cai-insidessd-2018, cai2013program,cai2015data, cai2017error}, \nrev{and implements the error-mitigation mechanisms, e.g., ECC, refresh, and read-retry, that keep the flash reliable}~\cite{cai-dsn-2015, cai-iccd-2012, cai-date-2012, cai-sigmetrics-2014, luo2015warm, luo2016enabling, luo2018heatwatch, park-asplos-2021}. \back{(v)}~The \emph{NAND flash chips}~\circled{7} provide the physical storage and execute the read, program, and erase commands the flash controller issues~\cite{onfi, micheloni2010inside, grupp2009characterizing, dirik2009performance, cai-insidessd-2018} and transfer data through multiple parallel channels. \back{NAND flash imposes three constraints on these commands. First, reads and programs act on a page, \fix{commonly} \prev{16}\,KB~\fix{\rrev{\cite{pekny2022isscc, cho20221, kang201913, maejima2018512gb}}}. Second, erases act on a whole block, a fixed group of consecutive pages~\back{\cite{micheloni2010inside, cai2017error, grupp2009characterizing, cai-date-2012, cai-insidessd-2018}}. Third, a programmed page cannot be programmed again until its block is erased. These constraints force the FTL to write \emph{out-of-place}, i.e., a write to a \fix{logical address} whose previous mapping is still live takes a fresh page from a free-block pool and invalidates the previous one. \fix{The FTL manages this out-of-place write and the metadata it requires}~(\S\ref{sec:prelim:ftl}).} Each NAND page carries a small \emph{out-of-band} (OOB) area next to its data~\cite{micheloni2010inside}, where the FTL stores the page's metadata, which consists of (i) page state (erased, live, or stale), (ii) a \emph{reverse mapping}, i.e., the logical address the page was written for, and (iii) an \emph{integrity tag} that lets a later read check the page's contents\rrev{, which we model after the end-to-end protection information a drive already carries per logical block~\cite{nvme_nvm_command_set}}.
\rrev{An FTL that keeps no in-DRAM reverse map reads this field during garbage collection, to learn the logical address of a live page before relocating it}~\cite{gupta2009dftl, agrawal2008design}.

\begin{figure}[h]
    \centering
    \includegraphics[width=1\linewidth]{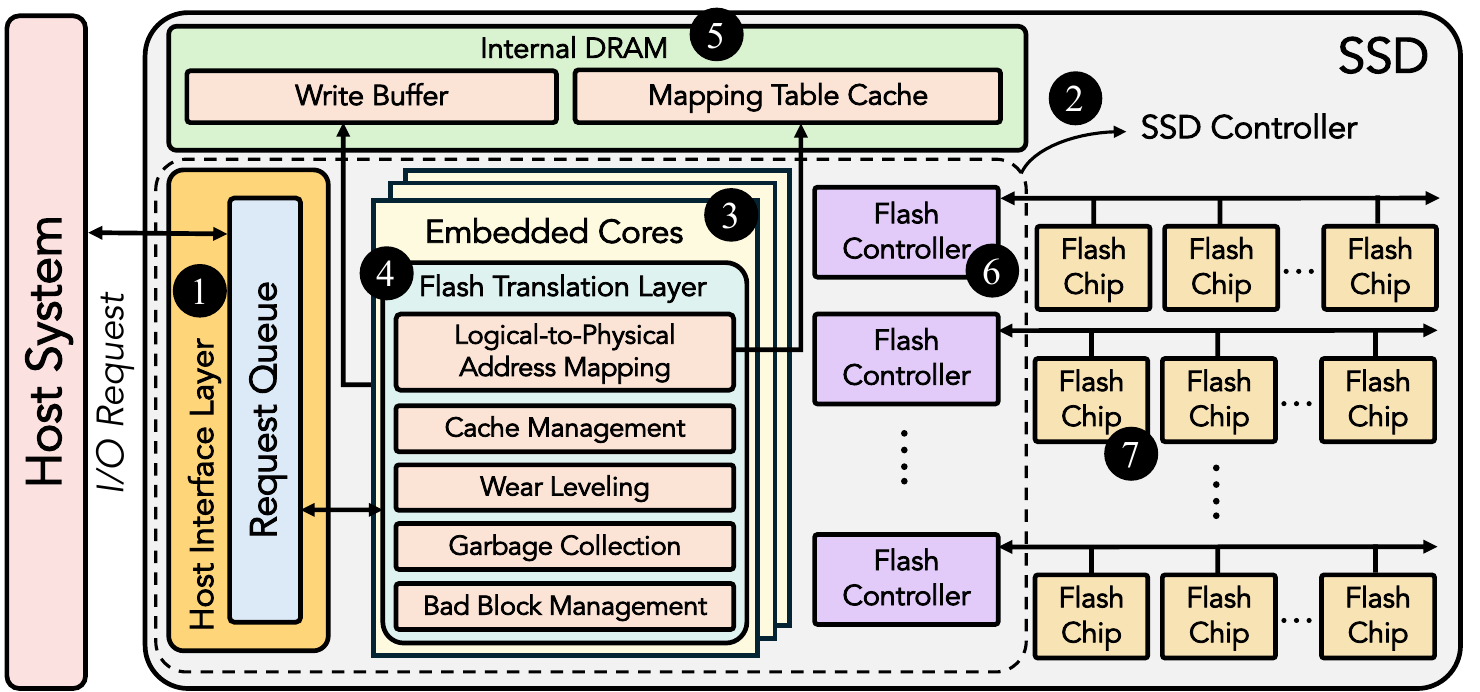}
    \caption{Overview of a modern solid-state drive (SSD).}
    \label{fig:ssd}
\end{figure}

\subsubsection{Flash Translation Layer (FTL)}
\label{sec:prelim:ftl}

The FTL~\cite{gal2005algorithms, park2009design, chung2009survey, goodson2010design} is the SSD control layer that \back{enables the host to use NAND flash as a block device, hiding the three constraints above while sustaining the device's performance, reliability, and lifetime}.
It comprises~\edit{\rrev{eight}} key components.

\noindent \textbf{(a) Logical-to-Physical (L2P) Mapping.}
The FTL maps each logical block address (LBA) from the host to a physical page number (PPN) in NAND flash~\cite{gal2005mapping}. Different FTL designs differ in granularity~\cite{kim2002space, ma2014survey, gupta2009dftl, sun2023leaftl, picoli2019ox}: \emph{page-mapped} FTLs use a single L2P table; \emph{hybrid-mapped} FTLs (e.g., BAST~\cite{kim2002space}, FAST~\cite{lee2007log}) combine a coarse \emph{Block Mapping Table} (BMT) with a fine-grained \emph{Log Page Table} (LPT)~\rrev{\cite{gupta2009dftl, ma2014survey}}; and \emph{demand-cached} FTLs (e.g., DFTL~\cite{gupta2009dftl}) keep only a working set of L2P entries in DRAM, fetching the rest from on-flash translation pages on demand. \back{ A page-mapped table is the fastest to consult but too large to hold in DRAM at scale; hybrid mapping shrinks it at the cost of merge operations, and demand caching keeps DRAM small but requires a flash read on every miss.}

\noindent \textbf{(b) Cache Management.}
The FTL caches L2P mappings and recently-used data in the SSD's internal DRAM~\cite{agrawal2008design, bang2011memory, ma_lazyftl_2011, zhou2015efficient, lim2010faster, shin2009ftl, wang2024learnedftl}, and buffers writes to reduce program/erase cycles and improve flash endurance.

\noindent \edit{\textbf{(c) Garbage Collection (GC).}
The FTL reclaims space by selecting blocks with few valid pages, migrating their live data to fresh blocks, and erasing the resulting old blocks~\cite{chung2009survey, agrawal2008design, xie2014adaptive, cai-insidessd-2018, jung2012taking, shahidi2016exploring, wu2016gcar, choi2018parallelizing}.} \fix{Erased blocks return to the free-block pool from which writes and relocations take fresh pages. GC runs without a host request and rewrites the L2P mapping and per-page metadata of every page it moves.}

\noindent \edit{\textbf{(d) \prev{Wear Leveling} (WL).}
The FTL balances erase wear across blocks by tracking per-block erase counts and redirecting future writes to less-used regions~\cite{gal2005algorithms, chung2009survey, chang2007wear, yang2014garbage}.} \fix{Like GC, it relocates live data without a host request.}

\noindent \textbf{(e) Crash Recovery.}
\fix{The FTL restores a consistent state after power loss through journaling and checkpointing~\cite{gupta2009dftl, chang2020determinizing}.}
% \noindent \fix{\textbf{(b) Per-Page Metadata.}

\noindent \mk{\textbf{(f) Data Layout in Flash.} The data layout defines how the FTL places host data across the NAND flash chips~\cite{onfi, micheloni2010inside, grupp2009characterizing, dirik2009performance, cai-insidessd-2018}, i.e., the channel, die, plane, block, and page each logical page is written to. It determines the parallelism the SSD can extract, the number of valid pages garbage collection must relocate, and which tenants' data share a block.}

\noindent \mk{\textbf{(g) DRAM Data Layout.} The SSD's internal DRAM holds the FTL's metadata and a write buffer for host data~\cite{agrawal2008design, bang2011memory}. The DRAM data layout defines where each structure resides, i.e., the logical-to-physical (L2P) mapping table, the free-block and allocation lists, the per-block valid-page and erase counts, and the integrity tags. All tenants and all FTL modules share this single flat address space.}

\noindent \mk{\textbf{(h) Flash Controller.} The flash controller issues read, program, erase commands to the NAND flash chips and manages their timing, error-correcting code (ECC), and status handling. 
% It executes the physical addresses the FTL supplies without validating them against any ownership or isolation policy.
}

\subsection{\nrev{Storage-Centric Computing}}
\label{sec:prelim:deploy}

% \subsubsection{Storage-Centric Computing}
% \label{sec:isp}

% Storage-centric architectures achieve this via three approaches: (1) \emph{Near-Storage Processing} (NSP)~\cite{Matam_ISCA2019, jang2024smart, ren2025device, acharya-asplos-1998, riedel-vldb-1998, keeton-sigmod-1998, riedel-computer-2001, boboila-msst-2012, tiwari-fast-2013, kang-msst-2013, do-sigmod-2013, bae-cikm-2013, wang2016ssd1, jun-isca-2018, torabzadehkashi-pdp-2019, pei-tos-2019, liang-atc-2019}, offloades computation to a unit near the storage interface (i.e., PCIe), bypassing host DRAM; (2) \emph{In-Storage Processing} (ISP)~\cite{Soysal2025MARS, Ghiasi2024MegIS, Mansouri-Ghiasi_ASPLOS2022, Li_ATC2021, Lee_ISCA2022, Liang_DAC2022, Duffy_PVLDB2023, Park_ASPLOS2025, Bisson_IPCCC2018, Chen2025REIS, Yu2024CambriconLLM, Kim_HPCA2023}, executing computations on embedded SSD controller cores or near-flash accelerators, leveraging internal bandwidth and DRAM buffers; and (3) \emph{In-Flash Processing} (IFP)~\cite{Kabra2025CipherMatch, Chen_MICRO2024, Park_MICRO2022, Gao_MICRO2021, chen2024search, hu2022ice}, leveraging the operational principles of NAND flash memory to perform bitwise operations within flash arrays.
\edit{Storage-centric architectures~\cite{nider2020processing,mutlu2022modern,mutlu2024memory,mutlu2019processing,mutlu2025memory,mutlu2019enabling,nadig2026conduit} \fix{reduce data movement between the SSD and processing units} via two approaches: (1) \emph{In-Storage Processing} (ISP)~\back{(e.g.,}~\cite{park2016storage, Kabra2025CipherMatch, mailthody-micro-2019, kim-fast-2021, kang-msst-2013, torabzadehkashi-pdp-2019, seshadri-osdi-2014, wang-eurosys-2019, acharya-asplos-1998, keeton-sigmod-1998, koo-micro-2017, tiwari-fast-2013, tiwari-hotpower-2012, boboila-msst-2012, bae-cikm-2013, torabzadehkashi-ipdpsw-2018, pei-tos-2019, do-sigmod-2013, kim-infosci-2016, riedel-computer-2001, riedel-vldb-1998, liang-atc-2019, cho-wondp-2013, jun2015bluedbm, lee2020smartssd, ajdari-hpca-2019, liang-fpl-2019, jun-hpec-2016, kang-tc-2021, kim-sigops-2020, Lee_ISCA2022, li2023ecssd, ruan2019insider, wang2016ssd1, jeong-tpds-2019, mao2012cache, gouk2024dockerssd, Ghiasi2024MegIS, Mansouri-Ghiasi_ASPLOS2022, kang2021iceclave, yavits2021giraf, Kim_HPCA2023, lim-icce-2021, narasimhamurthy2019sage, jun-isca-2018, fakhry2023review, gu2016biscuit, yang2023lambda, jo2016yoursql, Chen2025REIS, Li_ATC2021, wang2016ssd, mahapatra2025rag, wang2024beacongnn, Yu2024CambriconLLM, pan2024instattention}\back{)}, 
executing computation on the SSD's embedded controller cores~\cite{do-sigmod-2013, gu2016biscuit, koo-micro-2017} or on accelerator units placed in or near the SSD (e.g., smart NICs, on-board FPGAs~\cite{lee2020smartssd}, or near-flash accelerators), leveraging the \back{SSD's} internal bandwidth~\back{\cite{chen2011essential, hu2011performance, hu2012exploring, nadig2023venice, kim2022networked}} and DRAM \back{as buffers to store intermediate data}; and (2) 
\emph{In-Flash Processing} (IFP)~\back{(e.g.,}~\cite{Park_MICRO2022, Gao_MICRO2021, Chen_MICRO2024, kim2025crossbit, wong2024tcam, wong2025anvil, chen2024search, choi2020flash, chun2022pif, lee2025aif, wang2018three, Kabra2025CipherMatch, kang-tc-2021}\back{), 
leveraging the analog operation of NAND flash arrays to perform computations inside the array}.} These architectures accelerate diverse applications, e.g., genomics~\cite{ghiasi2026grains, ghiasi2026sage, Soysal2025MARS, Mansouri-Ghiasi_ASPLOS2022, Ghiasi2024MegIS} and machine learning~\cite{Chen2025REIS, Yu2024CambriconLLM, Kim_HPCA2023}.

% \back{Each depends on how the FTL addresses and places data.}

% \subsection{Motivation: FTL Customization}
\subsubsection{\back{FTL Modifications in Prior Designs}}
\label{sec:isp_analysis}

\crev{\fix{\rrev{Storage-centric computing requires modifications to the FTL, and prior designs have modified it extensively.} We analyze 17 prior storage-centric designs\rrev{, and} Table~\ref{tab:ssd_modifications_matrix} records \rrev{the} eight FTL dimensions~(\S\ref{sec:prelim:ftl}) that \rrev{they} modify.}}
\begin{table}[h]
\centering
\caption{FTL modifications in prior storage-centric designs.}
\label{tab:ssd_modifications_matrix}
\setlength{\tabcolsep}{3.2pt}
\renewcommand{\arraystretch}{1.12}
\resizebox{\columnwidth}{!}{%
\begin{tabular}{@{}llcccccccc@{}}
\toprule
& & \multicolumn{8}{c}{\textbf{Modifications to storage system}} \\
\cmidrule(l){3-10}
\textbf{Applications} & \textbf{Prior Works}
& \textbf{(A)} & \textbf{(B)} & \textbf{(C)} & \textbf{(D)} & \textbf{(E)} & \textbf{(F)} & \textbf{(G)} & \textbf{(H)} \\
\midrule
\multirow{3}{*}{Genome analytics}
& MARS~\cite{Soysal2025MARS}                 & \gmark & \rmark & \gmark & \gmark & \rmark & \gmark & \gmark & \gmark \\
& MegIS~\cite{Ghiasi2024MegIS}               & \gmark & \rmark & \gmark & \gmark & \rmark & \gmark & \gmark & \gmark \\
& GenStore~\cite{Mansouri-Ghiasi_ASPLOS2022} & \gmark & \rmark & \gmark & \gmark & \rmark & \gmark & \gmark & \gmark \\
\midrule
\multirow{4}{*}{Graph-based applications}
& GLIST~\cite{Li_ATC2021}       & \gmark & \rmark & \gmark & \rmark & \rmark & \gmark & \gmark & \gmark \\
& GraphSSD~\cite{Matam_ISCA2019}& \gmark & \rmark & \gmark & \gmark & \rmark & \gmark & \gmark & \gmark \\
& SmartSAGE~\cite{Lee_ISCA2022} & \rmark & \rmark & \rmark & \rmark & \rmark & \gmark & \gmark & \gmark \\
& VStore~\cite{Liang_DAC2022}   & \gmark & \rmark & \gmark & \gmark & \rmark & \gmark & \gmark & \gmark \\
\midrule
\multirow{3}{*}{Database analytics}
& Dotori~\cite{Duffy_PVLDB2023}   & \ymark & \ymark & \ymark & \ymark & \ymark & \ymark & \ymark & \ymark \\
& AnyKey~\cite{Park_ASPLOS2025}   & \gmark & \gmark & \gmark & \gmark & \gmark & \gmark & \gmark & \gmark \\
& Crail-KV~\cite{Bisson_IPCCC2018}& \ymark & \ymark & \ymark & \ymark & \ymark & \ymark & \ymark & \ymark \\
\midrule
\multirow{3}{*}{Machine learning}
& REIS~\cite{Chen2025REIS}              & \gmark & \rmark & \gmark & \ymark & \rmark & \gmark & \gmark & \gmark \\
& CambriconLLM~\cite{Yu2024CambriconLLM}& \rmark & \rmark & \rmark & \rmark & \rmark & \gmark & \gmark & \gmark \\
& OptimStore~\cite{Kim_HPCA2023}        & \gmark & \gmark & \gmark & \gmark & \rmark & \gmark & \gmark & \gmark \\
\midrule
\multirow{4}{*}{Bulk-bitwise operations}
& CIPHERMATCH~\cite{Kabra2025CipherMatch}& \gmark & \rmark & \ymark & \ymark & \rmark & \gmark & \gmark & \gmark \\
& Ares-Flash~\cite{Chen_MICRO2024}       & \gmark & \gmark & \gmark & \gmark & \gmark & \gmark & \gmark & \gmark \\
& Flash-Cosmos~\cite{Park_MICRO2022}     & \gmark & \gmark & \gmark & \gmark & \gmark & \gmark & \gmark & \gmark \\
& ParaBit~\cite{Gao_MICRO2021}           & \gmark & \gmark & \gmark & \gmark & \gmark & \gmark & \gmark & \gmark \\
\bottomrule
\end{tabular}}

{\scriptsize \gmark~= explicitly modified;\ \ymark~= potentially required;\ \rmark~= not modified.\\
(A)~logical-to-physical mappings, (B)~cache management, (C)~garbage collection, (D)~wear leveling,
(E)~crash recovery, (F)~data layout in flash, (G)~DRAM data layout, (H)~flash controller.}
\end{table}

\begin{itemize}
\item \crev{\textbf{Genome analytics.} \fix{These workloads stream large read sets sequentially and in parallel. Prior in-storage designs~\cite{ghiasi2026grains, ghiasi2026enabling, ghiasi2026sage, Soysal2025MARS, Ghiasi2024MegIS, Mansouri-Ghiasi_ASPLOS2022} map data with coarse-grained sequential L2P entries (A), schedule garbage collection and wear leveling so that maintenance does not interrupt streaming (C, D),
% schedule garbage collection and wear leveling \rrev{outside the in-storage phase} so that maintenance does not interrupt streaming (C, D),
place data across flash planes and DRAM so that multi-plane reads can serve it (F, G), and extend the controller with commands that invoke the in-storage accelerator (H).}}
\item \crev{\textbf{Graph and database analytics.} These workloads issue random I/O with little spatial locality. Prior designs~\cite{Li_ATC2021, Matam_ISCA2019, Lee_ISCA2022, Liang_DAC2022, Duffy_PVLDB2023, Park_ASPLOS2025, Bisson_IPCCC2018} add graph-native or key-value mappings (A), relax wear leveling (D), make garbage collection exploit locality (C), co-locate related data (F), and extend controller scheduling and the I/O protocol (H).}
\item \crev{\textbf{Machine learning.} These workloads demand high bandwidth and parallel access. Prior designs~\cite{Chen2025REIS, Yu2024CambriconLLM, Kim_HPCA2023} coarsen L2P mappings to clusters or model weights (A), reconfigure the flash layout for parallel access (F), and extend controller command handling (H).}
\item \crev{\textbf{Bulk-bitwise operations.} In-flash designs~\cite{Kabra2025CipherMatch, Chen_MICRO2024, Park_MICRO2022, Gao_MICRO2021} compute inside the NAND array. They require computation-aware mappings (A), operand-aligned layouts that mark compute regions (F), and custom flash commands (H), and \rrev{several} disable ECC or force SLC mode for precise array-level results.}
\end{itemize}

\noindent\fix{FTL customization is not limited to storage-centric computing. \rrev{Open-channel SSDs~\cite{lightnvm, picoli2020open} move the FTL to the host, and ZNS SSDs~\cite{znssds} expose zones so that the host controls data placement and reclaim while the device keeps mapping within a zone, wear leveling, and ECC. Key-value SSDs~\cite{jin2017kaml, Park_ASPLOS2025} accept variable-length keys instead of fixed-size logical block addresses, so their FTL indexes keys to physical pages instead of translating logical addresses.}}

\section{Threat Model}
\label{sec:security-model}

\arev{\rrev{We describe (1) the two settings we target, single-tenant and multi-tenant SSDs}~(\S\ref{sec:security-model:setting})\rrev{; and (2) the scope of our model, i.e., which device functions we verify and which we hold fixed}~(\S\ref{sec:security-model:scope})\rrev{. In both settings, we treat the FTL as untrusted.} 
% \arbiter's guarantees under this model are properties of the FTL's state and its host-visible behavior.
}

%
% \subsection{Our Settings}
\subsection{\rrev{Target Settings}}
\label{sec:security-model:setting}

\subsubsection{Single-Tenant SSDs} A single tenant owns the entire SSD, and the FTL manages all physical blocks on its behalf. This is the setting of \rrev{most storage-centric designs} (e.g.,~\cite{jang2024smart, Matam_ISCA2019, Soysal2025MARS, Ghiasi2024MegIS, Mansouri-Ghiasi_ASPLOS2022, Li_ATC2021, Lee_ISCA2022, Liang_DAC2022, Duffy_PVLDB2023, Park_ASPLOS2025, Bisson_IPCCC2018, Chen2025REIS, Yu2024CambriconLLM, Kim_HPCA2023, kim-fast-2021}), where a vendor customizes the FTL for one application and the device serves one workload at a time. The FTL is the only component that maintains mapping correctness, data integrity, and block allocation for that tenant, so an FTL fault can corrupt address translation, silently lose data, or invalidate the integrity guarantees the host relies on.

\subsubsection{Multi-Tenant SSDs} In commercial cloud systems, \rrev{it is standard practice to share physical resources across tenants to improve utilization and cost efficiency (e.g., CPU cores among database tenants~\cite{das2013cpu}).} \back{The same principle applies to SSDs (e.g.,~\cite{kim2015ops, jaliminche2023enabling, wang2018oc, luo2013s, gonzalez2017multi, liu2021self, sun2025fleetio}). Consolidation of multiple tenants onto a single high-capacity SSD avoids two costs. First, a drive dedicated to one tenant stays underutilized. \rrev{Second, a database that several tenants share must be replicated on every drive.}} \back{Recent ISP systems (e.g.,~\cite{kang2021iceclave, gouk2024dockerssd}) target the same multi-tenant setting. \fix{Tenants then share the SSD's flash, internal DRAM, and FTL metadata. \rrev{NVMe Namespaces, sets, and per-namespace media keys express the isolation boundary~\cite{nvmenamespaces, nvme_base_spec, tcg_key_per_io}, but every one of them is enforced through the FTL.}
% , so the FTL is the last on-device enforcement point for tenant isolation.
Any fault in the FTL implementation can therefore corrupt the mechanism every tenant's isolation depends on.}}

A multi-tenant FTL configuration consists of $n$ tenants sharing one SSD, each tenant $\tau_i$~\back{\cite{huang2017flashblox, kwon2020dc, kim2018utilitarian}} \rrev{owns a contiguous, non-empty range $R_i \subseteq \{0, \ldots, \mathtt{addr\_space} - 1\}$ of logical addresses. A physical block becomes that tenant's when the FTL programs one of the tenant's addresses into it, and no block ever holds two tenants' data.}~\back{The SSD vendor supplies the FTL, which runs on the controller's embedded cores; the tenants supply workloads, not code. The FTL has direct access to DRAM-resident metadata (mapping tables, per-page metadata, key tables) and can issue read, program, and erase commands to the flash controller.} \fix{We treat the FTL as untrusted. Any modification, whether faulty or malicious~\rrev{\cite{zaddach2013backdoor, meijer2019self, wertenbroek2024pandora}}, may update the metadata and access the flash in any order, and tenants may issue any host I/O. \rrev{Such an FTL aims to break a tenant's isolation, corrupt or lose data, or return data the host never wrote. %
% Three things sit outside what we verify, (i) whether our model matches a real FTL, (ii) device management, and (iii) the NAND flash medium}~(\S\ref{sec:security-model:scope})\rrev{.}
%
% We do not verify our model's fidelity to a real FTL, device management, or the NAND flash itself, which we assume returns the data that was programmed}~(\S\ref{sec:security-model:scope})\rrev{.}
% We assume that a real FTL behaves as our model does, that device management is correct, and that the flash returns the data that was programmed
}
% ~(\S\ref{sec:security-model:scope})\rrev{.}
} 
% Each tenant supplies its own FTL implementation, which runs as a firmware component on the SSD controller's embedded cores alongside the vendor's resident firmware. Its operations can read and update internal DRAM (mapping tables, per-page metadata, and any key tables held in DRAM) and can issue commands to the flash controller.
%

\subsection{\arev{Scope of the Model}}
\label{sec:security-model:scope}

\mk{
\rrev{\arbiter verifies the FTL state and the operations that act on it. The FTL state is the value of every FTL-managed data structure, i.e., the address mapping, the flash pages and their per-page metadata, the ownership labels and the region table, the allocator state, and the encryption keys}~(\S\ref{sec:model:state})\rrev{. The model has six operations, i.e., host reads, writes, invalidations, writing per-page integrity metadata, garbage collection, and wear leveling.} These six operations are the model's entire transition set, and every SSD policy a design changes, e.g., how a logical address maps to a
physical page, which physical page a write uses, which block is reclaimed, and how wear is spread, is a choice made inside them.}

\mk{
\rrev{\arbiter does not verify three things. First, whether our model matches a real FTL.} 
% We model the FTL state and its six operations, and we take the flash primitives to be the NAND command set~\cite{onfi, agrawal2008design}\rrev{.}
\fix{Second, device management}, e.g., creating, deleting,
and resizing tenant namespaces~\cite{nvme_base_spec, nvmenamespaces},
sanitizing the device~\cite{kim2020evanesco}, and setting NAND timing.
\arbiter holds these functions fixed \fix{and takes from them the region table
and the per-address ownership labels that say which addresses belong to which
tenant}. %
% \fix{Third, the NAND flash medium: we assume ECC corrects read-disturb and retention errors, so a read returns the data that was written}
%
% \fix{Third, the NAND flash medium\rrev{. We assume} ECC corrects read-disturb and retention errors, so a read returns the data that was written}
\fix{Third, the NAND flash \rrev{itself. We assume} ECC corrects read-disturb and retention errors, so a read returns the data that was written}
(e.g.,~\back{\cite{cai2017error, cai2015data, luo2015warm,
park-asplos-2021}}).} \mk{We do not model concurrency between FTL invocations, bad block retirement~\back{\cite{cai2017error}}, or timing side channels\rrev{~\cite{juffinger2025secret, juffinger2025hmb}}. At most
one FTL invocation runs at a time.}

\section{Motivation and \rrev{Our} Goal}
\label{sec:motivation}

%
% \mk{We provide our motivational analysis in two areas. First, we identify five failure surfaces (10 failures in total) that an unverified FTL can violate and demonstrate those failures on a real SSD (§\ref{sec:motivation:failures}). Second, we analyze prior FTL verification approaches and show they do not detect these failures
\mk{\rrev{We describe (1) the five failure surfaces at which an unverified FTL can break tenant isolation, data integrity, or block ownership, and the ten failures we demonstrate on a real SSD}~(§\ref{sec:motivation:failures})\rrev{; and (2) why prior FTL verification frameworks do not rule these failures out} \rrev{(§\ref{sec:motivation:priorwork})}.}

\subsection{\nrev{Failures of an Unverified FTL}}
\label{sec:motivation:failures}
\label{sec:security-model:surfaces}

\subsubsection{\nrev{Failure Surfaces}}
\begin{prevblock}
\nrev{A modified FTL runs on the controller with direct access to the five components of \S\ref{sec:prelim:arch}. Each is a \emph{failure surface}, a component at which a faulty or malicious FTL action can break tenant isolation, data integrity, or block ownership.} \rrev{We name a surface by the part the FTL reaches through, e.g., the host-interface layer's request queue, not by where the corrupted state ends up.} \prev{The five surfaces expose ten failures, described below and proven reachable in \rrev{\S}\ref{sec:attacks}.} 
\begin{itemize}
\item \textbf{HIL request queue (FS\#1).} Host commands reach the FTL through this queue. While serving a command, a modified FTL can erase (i)~a block a tenant's address still maps to, so the freed block is reallocated to a second tenant while the first still maps into it, or (ii)~a block index outside the device, so a block that does not exist enters the free pool.
% ~\rrev{A block handed on without being cleared carries the previous tenant's data with it~\cite{wei2011reliably}.}
\item \textbf{Internal DRAM (FS\#2).} The internal DRAM holds the FTL's address mapping and per-page metadata. A modified FTL can (i)~map two addresses to one physical page, so a read of one returns another tenant's data, (ii)~map two offsets of one address to one page, so two of that address's own pages return the same data, or (iii)~map an address to a page recorded under another namespace, so a read crosses the namespace boundary.
\item \textbf{Embedded compute units (FS\#3).} The controller's cores and accelerators run the FTL firmware and in-storage programs. A modified accelerator routine can map one tenant's address to a page (i)~another tenant owns, so a read returns that tenant's data, or (ii)~recorded under another namespace, so a read crosses the namespace boundary.
\item \textbf{Flash controller (FS\#4).} The flash controller turns the FTL's operations into NAND commands. A modified FTL can program a page but omit its integrity tag, so no later read can verify the page's data.
\item \textbf{NAND flash chips (FS\#5).} The NAND chips execute the read, program, and erase commands, and raw access reaches them directly\rrev{~\cite{cai-hpca-2017}}. A modified FTL can (i)~return to the free pool a block that is already there, so the pool holds it twice and hands it to two tenants, or (ii)~program a page and map it to an address in another namespace, so a read crosses the namespace boundary.
\end{itemize}
\end{prevblock}

\subsubsection{\nrev{Demonstration on a Real SSD}}
\label{sec:attacks:board}
\begin{nrevblock}
\prev{We demonstrate all ten failures on a real SSD. In each, we modify the FTL firmware at the surface and run the FTL from an NVMe host. Six failures produce a host-visible effect: an NVMe read of the attacker's own address returns another tenant's data, which we confirm by comparing SHA-256~\prev{\cite{fips180-4}} digests. The other four corrupt FTL metadata that no NVMe read exposes, and we confirm each by the firmware's boot-time self-test of the violated invariant clause. We present FS\#2, the L2P aliasing, in detail. The others follow the same procedure.}
% \S\ref{sec:attacks} proves that each of the five is reachable and shows how it is blocked.

\noindent\textbf{Setup.}
We use a DaisyPlus OpenSSD~\cite{daisyplus_openssd, kwak2018cosmos} (Xilinx Zynq UltraScale+ ZU17EG, Micron NAND) running \rrev{the OpenSSD reference FTL (gr3ftl), which manages the flash on the device and implements the NVMe controller in firmware}. The board attaches to a Linux host over PCIe and enumerates as an ordinary NVMe SSD with 4\,KiB logical blocks. The FTL maps in 16\,KiB slices of four blocks. Tenants are logical-address ranges. The victim and the attacker are distinct logical slice addresses (LSAs), and we issue all I/O from the host with the standard \texttt{nvme} tool.

\noindent\textbf{L2P aliasing.}
The FTL's mapping step is the concrete counterpart of the model's L2P map: \rrev{an LSA} resolves to a virtual slice address (VSA) backed by a physical page. The aliasing \prev{failure} is a remap of the attacker's entry onto the victim's,
\begin{lstlisting}
victimVsa = L2P[VICTIM_LSA].vsa;    /* 1024 */
L2P[ATTACKER_LSA].vsa = victimVsa;  /* 2048 -> 1024's slice */
\end{lstlisting}
so two logical slices resolve to one physical slice. A victim writes a secret to its slice. The attacker then issues an ordinary NVMe read of \emph{its own} address. With the remap in place, the read returns the victim's secret. The SHA-256 of the attacker read equals the victim's and differs from an uninvolved control block,
\begin{lstlisting}
victim   read sha256 = 9436f68d..dcba
control  read sha256 = 07bdd78f..5abb
attacker read sha256 = 9436f68d..dcba  (= victim)
\end{lstlisting}
a cross-tenant confidentiality breach carried entirely over the standard block interface, with no privileged access, firmware exploit, or side channel. Two lines in the FTL suffice.
\end{nrevblock}

\noindent \nrev{\emph{The Verification Problem.}}
\nrev{%
Each modification in Table~\ref{tab:ssd_modifications_matrix} changes the states the FTL can reach at these surfaces. A custom FTL must therefore be proven, before deployment, to reach no failing state at any surface. \rrev{Each new design changes the FTL again, so this proof is needed again. The verification problem is to produce it for every custom FTL at a cost proportional to what the design changes, not to the whole FTL.}}

\subsection{Limitations of Prior FTL Verification}
\label{sec:motivation:priorwork}
 
\rrev{Prior work on the verification of FTLs and flash file systems almost entirely verifies against a
\emph{functional} specification, i.e., that the device behaves as an idealized
disk or file system.} We identify three key limitations of prior designs.

\noindent\textbf{Limitation 1: No security property is stated.}
\rrev{Qiao et al.~\cite{qiao2019formal} define FTL correctness as refinement to a
disk-like specification, SCFTL~\cite{chang2020determinizing} verifies snapshot
consistency, and Flashix~\cite{bodenmuller2021flashix} and
Cogent~\cite{amani2016cogent} verify flash \emph{file systems} against
file-system specifications.} None of them states tenant
isolation, integrity-tag coverage, or block ownership. A modified FTL can
therefore violate all three and still satisfy the specification these
frameworks prove. Our real-system attacks make this concrete: \rrev{five of the
ten attacks in \S\ref{sec:motivation:failures} change only ownership labels and
integrity tags, which no functional specification constrains, so every
framework above admits them}. Tripathy et al.~\cite{tripathy2023formal} do check
security properties, but they model-check one ransomware-resistant FTL rather
than state a contract that another FTL can satisfy.
 
\noindent\textbf{Limitation 2: Proofs are tied to one implementation.}
\rrev{SCFTL~\cite{chang2020determinizing} and Flashix~\cite{bodenmuller2021flashix} each verify a single fixed implementation}, and
Flashix's layered structure composes within a fixed FTL rather than across designs.
\rrev{Cogent~\cite{amani2016cogent} is a language and compiler, and its refinement proof relates C code
to its own specification, so it transfers to no other FTL either.}
Qiao et al.\ build a framework \rrev{in which an FTL refines a disk, and demonstrate it on the
hybrid-mapped BAST FTL~\cite{kim2002space}. 
% Two of their five hypotheses, read-after-write and non-interference between pages, are close to \arbiter's.
In their framework, the invariant is a parameter, so a newer FTL design must be verified from scratch, which significantly increases design effort. Their model also has $only$ three operations,
i.e., initialize, read, and write, so it covers neither garbage collection nor
wear leveling}. Hence, the verification effort is required again for every
design change.
 
\noindent\textbf{Limitation 3: The verified designs predate storage-centric
computing.} %
\rrev{The FTLs these frameworks verify, e.g., the hybrid-mapped BAST~\cite{kim2002space}, requires only address-translation and maintenance state.
%
% A storage-centric FTL additionally keeps per-tenant ownership metadata and the coarse mappings its accelerators read (\S\ref{sec:isp_analysis}), and it is exactly this state that a custom FTL modifies. No prior framework constrains it.}
A storage-centric FTL that can be used in multi-tenant scenarios also requires per-tenant ownership metadata, the range mappings its accelerators read (\S\ref{sec:isp_analysis}), and a custom FTL modifies exactly this state. No prior framework constrains it %
\rrev{and only \drev{verifies} a fixed set of FTL operations. A storage-centric design may also propose a new operation built on the NAND commands and metadata updates beneath it~(\S\ref{sec:prelim:arch}). Verifying such a design therefore requires proofs at that lower level, which no prior framework provides.}}
 
Hence, \rrev{each existing FTL verification framework establishes} functional correctness for one fixed, conventional FTL, \rrev{and the one framework that checks security properties does so for a
single design}. A storage-centric FTL needs isolation, integrity, and ownership
guarantees that survive modification. A verification framework for modified
FTLs must fix these guarantees once, as a contract, so that a new design
inherits them instead of re-proving them.

\mk{\textbf{Our goal} is to design a \rrev{reusable} formal verification framework that enables FTL designers to develop a machine-checked proof of the correctness and security of their customized FTL.}
\section{\arbiter: Formal Verification Framework}
\label{sec:formal-model}

We propose \arbiter, a formal verification framework implemented in the \coq proof assistant~\cite{coq}. \arbiter proves two properties of a designer's FTL. (1) The FTL \emph{preserves} a \emph{global invariant}, i.e., \mk{a conjunction of multiple correctness and security properties} over the FTL's state that every operation must keep true. (2) The FTL \emph{refines} an idealized block device, \numupd{i.e., every read the host issues returns the data a simple, correct block device would return, so its host-visible
behavior is indistinguishable from that of the idealized device.}

% In Section~\ref{sec:attacks} we define one violation per surface formally and prove that the violation is reachable from any well-formed state.

\arbiter reasons about an FTL at two levels. First, at the \emph{operation} level, the six operations of the FTL model (\S\ref{sec:model:state}) are the unit of proof. \arbiter proves once, for its own model, that each operation preserves the global invariant and that the model refines an idealized block device (\S\ref{sec:model:invariants}, \S\ref{sec:model:oplevel}). %
Second, at the \srev{\emph{command}} level, each operation expands into the \srev{commands} the device executes \srev{(\S\ref{sec:model:commands})}, which carries the guarantee down to what an FTL issues directly, not only to an FTL that acts through the six operations. %
\srev{A designer reaches these guarantees by direct or by modular verification, which \S\ref{sec:verify} describes.} All proofs are mechanized in \coq~\cite{coq, huet1997coq, chlipala2013certified}. We define the FTL failure as follows. 

\begin{definition}[FTL failure]
\label{def:failure}
An FTL exhibits a \emph{failure} if,~\edit{starting from a well-formed state, \fix{i.e., one that satisfies the global invariant,} it can reach}
% some sequence of its operations, starting from a well-formed initial state, reaches
\back{a state in which the global invariant (stated in \S\ref{sec:model:invariants}, Definition~\ref{def:ftl-invariant}) is false.} The failure surfaces of \S\ref{sec:motivation:failures} are the components at which an FTL introduces such a state.
\end{definition} 

% \back{The definition does not refer to intent. A failure is a reachable state, whether an omitted metadata update or a deliberate sequence of operations produces it.} \back{\S\ref{sec:security-model:surfaces} names the five components at which such a failure can be introduced.}

 % \S\ref{sec:attacks} \prev{formalizes the ten failures as state transformations, proves each reachable, and shows how \arbiter rules out all of them.}

\subsection{The FTL Model}
\label{sec:model:state}

%
% \fix{We give an \emph{abstract specification} $\mathcal{S}$ of an idealized block device's host-visible behavior (\S\ref{sec:model:abstract}).
% \fix{We give an \emph{abstract specification} of an idealized block device's host-visible behavior (\S\ref{sec:model:abstract}). We model the FTL as a deterministic state machine $(\mathcal{Q}, \Sigma, \delta, q_0)$. A state in $\mathcal{Q}$ is an \texttt{FTLState} record capturing the FTL's full internal state, its address mapping, per-page metadata, ownership labels, allocator state, and encryption keys (\S\ref{sec:model:concrete}, Table~\ref{tab:ftl-state}). %
% $\Sigma$ is the operations the FTL serves (\S\ref{sec:model:ops}), $\delta$ their transition function, and $q_0$ a well-formed initial state.} \srev{The device executes \emph{commands} rather than operations, so we also model the ten commands and the expansion of each operation into a sequence of them (\S\ref{sec:model:commands}).}

%
We first specify the host-visible behavior of an idealized block device
(\vrev{\S\ref{sec:model:abstract}}). We then model the FTL as a deterministic state machine
$(Q, \Sigma, \delta, q_0)$. Each state in $Q$ is an
$\mathtt{FTLState}$ record that captures the FTL's address mapping, per-page
metadata, ownership labels, allocator state, and encryption keys
(\vrev{\S\ref{sec:model:concrete}, Table~\ref{tab:ftl-state}}). $\Sigma$ contains the operations that the FTL serves
(\vrev{\S\ref{sec:model:ops}}), $\delta$ defines their transitions, and $q_0$ is a well-formed
initial state. Each operation expands into a sequence of lower-level commands, which we explain in \vrev{\S\ref{sec:model:commands}}.

\subsubsection{Abstract Specification}
\label{sec:model:abstract}
\rrev{$\mathcal{S}$ is the set of abstract device states. A state $h \in \mathcal{S}$ is a partial map $(\mathtt{Addr (a)} \times \mathtt{Page (p)}) \rightharpoonup \mathtt{Data (d)}$ that assigns data to a (logical address, page offset) pair. $\mathtt{write}(a, p, d)$ updates $h(a,p)$ with data $d$, $\mathtt{read}(a, p)$ returns it, and $\mathtt{invalidate}(a, p)$ makes $h(a,p)$ undefined.} The specification captures only host-visible functional behavior. It abstracts away \mk{the tasks managed by FTL,} such as tenant ownership, physical layout, out-of-place writes, and free-block management. 
% Tenant isolation is enforced separately, by the concrete invariant's ownership clauses (\S\ref{sec:model:invariants}).

\subsubsection{Concrete State}
\label{sec:model:concrete}
\vrev{$q_0$ is a freshly formatted device} after vendor setup has installed the address labels (\S\ref{sec:security-model:scope}), with no page mapped. Throughout the paper, $T?$ denotes an optional $T$, and all index types are natural numbers.
% \drev{A logical address spans $\mathtt{pages\_per\_block}$ pages, so the same constant bounds a page offset within an address and within a block. The two are counts, not locations: the physical pages of one address still need not share a block.} 
\arev{The 16 fields of \texttt{FTLState} (Table~\ref{tab:ftl-state}) fall into five groups.}\footnote{These 16 fields are a base. Richer designs route additional fields through the $\mathtt{CUSTOM\_FTL}$ interface \srev{(\S\ref{sec:verify:modular})}.}

\begin{table}[ht]
\centering
\caption{The 16 fields of \texttt{FTLState}, grouped by role.}
\label{tab:ftl-state}
\renewcommand{\arraystretch}{0.85}
\resizebox{\columnwidth}{!}{%
\begin{tabular}{@{}lll@{}}
\toprule
Group & Field & Type \\
\midrule
Mapping & \texttt{l2p\_map} & $\mathtt{Addr} \!\times\! \mathtt{Page} \to \mathtt{PhysAddr?}$ \\
\midrule
\multirow{3}{*}{Flash pages}
 & \texttt{page\_state} & $\mathtt{Block} \!\times\! \mathtt{Page} \to \mathtt{PageState}$ \\
 & \texttt{page\_role} & $\mathtt{Block} \!\times\! \mathtt{Page} \to \mathtt{Role?}$ \\
 & \texttt{page\_meta} & $\mathtt{Block} \!\times\! \mathtt{Page} \to \mathtt{PageMeta}$ \\
\midrule
\multirow{3}{*}{Ownership}
 & \texttt{addr\_tenant}, \texttt{addr\_namespace} & $\mathtt{Addr} \to \mathtt{TenantId?}$ / $\mathtt{NamespaceId?}$ \\
 & \texttt{block\_tenant}, \texttt{block\_namespace} & $\mathtt{Block} \to \mathtt{TenantId?}$ / $\mathtt{NamespaceId?}$ \\
 & \texttt{region\_table} & $\mathbb{N} \to \mathtt{RegionMeta?}$ \\
\midrule
\multirow{6}{*}{Allocation}
 & \texttt{free\_block\_list} & $\mathtt{list\;Block}$ \\
 & \texttt{free\_block} & $\mathtt{Block} \to \mathtt{bool}$ \\
 & \texttt{wear\_count} & $\mathtt{Block} \to \mathbb{N}$ \\
 & \texttt{open\_block} & $\mathtt{TenantId} \!\times\! \numupd{\mathtt{NamespaceId}} \to \mathtt{Block?}$ \\
 & \texttt{write\_ptr} & $\mathtt{TenantId} \!\times\! \numupd{\mathtt{NamespaceId}} \to \mathtt{Page}$ \\
 & \texttt{block\_open} & $\mathtt{Block} \to \mathtt{bool}$ \\
\midrule
Crypto & \texttt{key\_table} & $\mathtt{KeyId} \to \mathtt{KeyMeta?}$ \\
\bottomrule
\end{tabular}}
\end{table}

\begin{arevblock}
\begin{itemize}[nosep,leftmargin=*]
\item \textbf{Mapping.} \brev{$\mathtt{l2p\_map}$ is the L2P mapping table. It sends each \emph{logical page}, a host address with a page offset, to its own physical page, so the pages of one address may sit in different blocks.}

\item \textbf{Flash pages.} $\mathtt{page\_state}$ marks each cell erased ($\mathtt{PS\_Empty}$), stale ($\mathtt{PS\_Invalid}$), or live ($\mathtt{PS\_Valid}(d)$). $\mathtt{page\_role}$ separates user data ($\mathtt{RData}$) from FTL metadata ($\mathtt{RMeta}$). \fix{$\mathtt{page\_meta}$ models the page's out-of-band area (\S\ref{sec:prelim:ftl}). It holds the owner tenant, the owner namespace, an optional integrity tag, and the reverse mapping, i.e., the logical page stamped into the page when it was programmed.}

\item \textbf{Ownership.} $\mathtt{addr\_tenant}$ and $\mathtt{addr\_namespace}$ label each logical address with its owning tenant and namespace, and $\mathtt{region\_table}$ records each region's tenant, namespace, and address range. Vendor namespace setup installs all three (\S\ref{sec:security-model:scope}), and no operation in the model changes them. $\mathtt{block\_tenant}$ and $\mathtt{block\_namespace}$ carry the same two labels at block granularity. The FTL copies them from the address when it allocates a block to that address, and clears them when it erases the block. The model carries both granularities because the host names addresses while the device's access-control hardware sees blocks.

\item \textbf{Allocation.} $\mathtt{free\_block\_list}$ is a stack of erased blocks. An erase pushes a block onto it, and a write or relocation takes its destination block from the head. $\mathtt{free\_block}$ is a per-block flag for the same pool, which lets a check test one block without scanning the list. Every operation that allocates or frees a block updates both. $\mathtt{wear\_count}$ records how many times each block has been erased, and every erase increments it. $\mathtt{open\_block}$ names the block each tenant and namespace is currently filling, and $\mathtt{write\_ptr}$ points to the next free page in it. $\mathtt{block\_open}$ marks which blocks are open for writing.

\item \textbf{Crypto.} $\mathtt{key\_table}$ holds device-global encryption-key metadata and lies outside any tenant's view.
\end{itemize}
\end{arevblock}

\subsubsection{Operation Semantics}
\label{sec:model:ops}
\arbiter defines six operations, four \emph{host-visible} requests ($\mathtt{COpRead}$, $\mathtt{COpWrite}$, $\mathtt{COpInvalidate}$, $\mathtt{COpSetTag}$) and two \emph{maintenance} transitions ($\mathtt{COpGC}$, $\mathtt{COpWearLevel}$).
% \footnote{In real SSDs, garbage collection and wear leveling are background tasks. \arbiter treats both as explicit transitions, applied only when a trace requests them, which keeps refinement independent of any maintenance-scheduling policy.} 
\arev{They are exactly \back{the FTL's} responsibilities under \S\ref{sec:security-model:scope}, i.e., the data path (read, write, invalidate), integrity metadata (SetTag), and wear maintenance (GC, wear leveling).} The transition function is \arev{$\mathtt{step} : \mathtt{FTLState} \to \mathtt{COp} \to \mathtt{option}\;\mathtt{FTLState}$, which returns $\mathtt{Some}\;s'$ when the operation applies and $\mathtt{None}$ when its safety checks (Table~\ref{tab:operations}) reject it, leaving the state unchanged (e.g., a write with no free block). Every operation-level theorem (\S\ref{sec:model:oplevel}) assumes the step succeeded, i.e., that $\mathtt{step}$ returned $\mathtt{Some}\;s'$. A rejected operation produces no successor state, so there is nothing to constrain.}

\begin{table}[h]
\centering
\caption{The six FTL operations, their effect on the FTL state, and safety constraint.}
\label{tab:operations}
\resizebox{\columnwidth}{!}{%
\begin{tabular}{@{}lll@{}}
\toprule
Operation & Effect on state & Key constraint \\
\midrule
$\mathtt{COpRead}(a, p)$ & identity & -- (no data if $a$ unmapped or page invalid) \\
$\mathtt{COpWrite}(a, p, d)$ & \drev{\shortstack[l]{stale the old page, program a fresh one\\from the owner's write frontier}} & \drev{\shortstack[l]{$a < \mathtt{addr\_space}$, $p < \mathtt{pages\_per\_block}$;\\$a$ labelled; a page available}} \\
$\mathtt{COpInvalidate}(a, p)$ & mark mapped page $\mathtt{PS\_Invalid}$ & -- (no effect if $a$ unmapped) \\
$\mathtt{COpSetTag}(a, p, \mathit{tag})$ & update the integrity tag of the mapped page & -- (no effect if $a$ unmapped or page invalid) \\
$\mathtt{COpGC}$ & relocate the victim's live pages, erase and free it & \drev{\shortstack[l]{a reclaimable victim; every\\live page relocates successfully}} \\
$\mathtt{COpWearLevel}$ & the same, with a different victim choice & \drev{the same} \\
\bottomrule
\end{tabular}}
\end{table}

\subsubsection{\srev{Command Semantics}}
\label{sec:model:commands}

\srev{Each of \arbiter's six operations expands into a sequence of
\emph{commands}. A faulty FTL can issue these commands directly. \arbiter defines
ten commands, categorized into flash and metadata commands. The flash
commands are $\mathtt{PrimRead}$, $\mathtt{PrimProgram}$, and
$\mathtt{PrimErase}$, which perform NAND read, program, and
erase~\cite{onfi}. $\mathtt{PrimProgram}$ also writes the page's out-of-band
area (\S\ref{sec:prelim:ftl}), including the reverse mapping and integrity tag
supplied by the FTL. The seven metadata commands do not write the flash array.
$\mathtt{PrimMapAddr}$ and $\mathtt{PrimRemap}$ install and redirect mappings,
$\mathtt{PrimInvalidate}$ marks a page stale, removes its forward mapping
\drev{and clears its role, so a stale page counts as neither user data nor metadata},
$\mathtt{PrimFreePush}$ returns a block to the free pool, and
$\mathtt{PrimSetTag}$ writes a page's integrity tag. Two barriers bracket a
compound operation without changing state.}
\srev{For example, a write to a fresh address expands into four commands: an
opening barrier, $\mathtt{PrimMapAddr}$, $\mathtt{PrimProgram}$, and a closing
barrier. Rewriting an address first uses $\mathtt{PrimInvalidate}$ to mark the
old page stale because NAND flash cannot overwrite in place
(\S\ref{sec:prelim:arch}). %
\vrev{An operation expands only when it is valid. A rejected operation, such as a
write outside the address space, issues no commands.}
\drev{$\mathtt{PrimErase}$ erases a block and returns it to the free pool in one
command, so no expansion emits $\mathtt{PrimFreePush}$; that command is there for an
FTL that issues commands directly.} 
% Neither case erases a
% block. Garbage collection and wear leveling use $\mathtt{PrimErase}$ after
% relocating the live pages of a victim block. 
% An invalid operation has no
% command sequence, for example, a write outside the address space or to an
% address the vendor never labelled, or garbage collection when there is no
% victim to reclaim.
}

\subsection{Global Security Contract}
\label{sec:model:invariants}

\begin{table*}[t]
\centering
\caption{\rrev{The 27 clauses of the global FTL invariant, grouped by primary purpose into five categories.} 
% A page is \emph{live} (\texttt{PS\_Valid}), \emph{stale} (\texttt{PS\_Invalid}), or \emph{erased} (\texttt{PS\_Empty}); a page's \emph{back-pointer} is the logical page stamped in its out-of-band area; a block's \emph{write frontier} is its next unwritten page.
% The two well-formedness side conditions on the device geometry, $\mathtt{WF}_0$ and $\mathtt{WF}_1$, are omitted, since both hold of any state.
}
\label{tab:invariants}
\scriptsize
\setlength{\tabcolsep}{3pt}
\renewcommand{\arraystretch}{0.92}
\begin{tabular}{@{}ll p{0.80\textwidth}@{}}
\toprule
\textbf{Category} & \textbf{Inv.} & \textbf{Description} \\
\midrule
\multirow{5}{*}{Mapping}
  & Inv$_0$  & Every live page is the target of some logical page in the L2P map. \\
  & Inv$_1$  & Every L2P entry maps an in-bounds logical page to an in-bounds physical page. \\
  & Inv$_2$  & The L2P map is injective: no two logical pages point to the same physical page. \\
  & Inv$_3$  & If a logical page maps to a live physical page, that page's \fix{reverse mapping} names the same logical page (forward matches reverse). \\
  & Inv$_4$  & If a live page's \fix{reverse mapping} names a logical page, the L2P map sends that logical page back to this page (reverse matches forward). \\
\midrule
\multirow{4}{*}{Isolation}
  & Inv$_5$  & No block that holds a mapped page appears in the free list. \\
  & Inv$_6$  & Every page of a free block is erased and carries no metadata, so a recycled block keeps no owner or tag. \\
  & Inv$_7$  & A live page mapped from address $a$ records $a$'s owning tenant and namespace. \\
  & Inv$_8$ & Every block in the free list is a valid block index ($b < \texttt{total\_blocks}$). \\
\midrule
\multirow{8}{*}{Integrity}
  %
% Inv$_9$ & Every live page carries a cryptographic integrity tag. \\
& Inv$_9$ & \rrev{Every live page carries an integrity tag.} \\
  %
%  Inv$_{10}$ & Every block is free, open for writing, in use (the target of some logical page), or a \emph{garbage block} holding a stale page and awaiting reclamation. \\
%
& Inv$_{10}$ & \rrev{Every in-bounds block} is free, open for writing, in use \vrev{(the target of some logical page)}, or a \emph{garbage block} holding a stale page and awaiting reclamation. \\
  & Inv$_{11}$ & No block appears twice in the free list. \\
& Inv$_{12}$ & \rrev{An in-bounds block} that holds a metadata page is in use or holds a stale page \drev{(derivable, as are Inv$_{15}$ from Inv$_{16}$, and Inv$_2$ from Inv$_3$ and Inv$_{22}$)}. \\
  & Inv$_{13}$ & Every live page is a user-data page. \\
  & Inv$_{14}$ & Every user-data page is live and holds a data value. \\
  & Inv$_{15}$ & Every metadata page is non-erased. \\
  & Inv$_{16}$ & Every erased page has no role assigned. \\
\midrule
\multirow{3}{*}{Ownership}
  & Inv$_{17}$ & Every block that \texttt{free\_block} marks free has no owning tenant and no owning namespace. \\
& Inv$_{18}$ & A block that holds a mapped page carries the tenant and namespace of the address mapped there, \rrev{so block-level and address-level ownership agree}. \\
  & Inv$_{19}$ & Every defined region covers addresses within \texttt{addr\_space}. \drev{No clause ties a region's tenant to the per-address labels; vendor setup installs both.} \\
\midrule
\multirow{7}{*}{\brev{Allocation}}
  & \brev{Inv$_{20}$} & %
\brev{\rrev{An open block is in bounds, not free, marked open, has $\texttt{write\_ptr} \le \texttt{pages\_per\_block}$, and is open for one owner, which owns it or nobody does.}} \\
  & \brev{Inv$_{21}$} & %
\brev{Every page at or beyond \rrev{an open block's} write frontier is still erased and carries no metadata.} \\
  & \brev{Inv$_{22}$} & \brev{The L2P map points only at live pages.} \\
  & \brev{Inv$_{23}$} & \brev{Every block marked open has an owner it is open for.} \\
  & \brev{Inv$_{24}$} & \brev{\texttt{free\_block} and the free list (\texttt{free\_block\_list}) agree on which blocks are free.} \\
  & \brev{Inv$_{25}$} & %
\brev{Every page below \rrev{an open block's} write frontier has been written, i.e., is not erased.} \\
  & \brev{Inv$_{26}$} & %
\brev{\rrev{Every mapped address has an owning tenant and namespace.}} \\
\bottomrule
\end{tabular}
\end{table*}

\rrev{The contract has two parts. The global invariant constrains the FTL's internal state. The refinement relation ties the concrete state to the abstract specification.}

\begin{definition}[FTL Invariant]
\label{def:ftl-invariant}
The \emph{global FTL invariant} $\mathtt{ftl\_invariant}(s)$ is the conjunction of \back{27} clauses, $\mathtt{Inv}_0(s) \wedge \cdots \wedge \back{\mathtt{Inv}_{26}}(s)$, shown in Table~\ref{tab:invariants}, with two well-formedness conjuncts, $\mathtt{WF}_0(s)$ and $\mathtt{WF}_1(s)$.\drev{\footnote{$\mathtt{WF}_0$ restates $\mathtt{pages\_per\_block} > 0$ and ignores its state argument. $\mathtt{WF}_1$ says $\mathtt{page\_state}$ is defined at every in-range block and page, which holds of every state because it is a total function. Neither constrains an FTL, which is why we do not count them among the clauses.}}
% \footnote{\arev{In the artifact the conjunction is $\mathtt{tableII\_invariant}$ (\texttt{Invariants.v}), re-exported to designers as $\mathtt{security\_contract}$ (\texttt{CustomFTLInterface.v}). The latter name appears in the proof scripts of \S\ref{sec:evaluation}.}}

\end{definition}

\arev{\noindent \textbf{Design of the invariant.}
We derived the clauses from prior FTL designs and their known failure modes. \rrev{From these invariants we prove refinement to the idealized block device (Theorem~\ref{thm:refinement}). That specification is written independently of the invariant, so functional correctness is measured against an external standard rather than against the clauses themselves. 
\vrev{From the same set} we also prove the three properties of our threat model, tenant isolation, data integrity, and block ownership. A stronger threat model adds clauses to the set. Each new clause needs a preservation proof under all six operations, and every FTL already verified re-discharges its hypotheses over the extended invariant.}}

\begin{definition}[Refinement Relation]
\label{def:refinement}
The refinement relation $\mathcal{R} \subseteq \mathcal{S} \times \mathtt{FTLState}$ is defined as

{\small
\begin{equation}
\label{eq:refinement}
\begin{aligned}
  \mathcal{R}(h, s) \;\triangleq\; & \forall a, p, d.\; h(a, p) = d \;\Rightarrow\; \rrev{\exists b, q}.\\
  & \mathtt{l2p\_map}(s, a, p) = \mathtt{Some}\;\rrev{(b, q)} \;\wedge\\
& \rrev{\mathtt{page\_state}(s, b, q)} = \mathtt{PS\_Valid}(d).
\end{aligned}
\end{equation}}
\end{definition}

Equation~\eqref{eq:refinement} relates the two devices by what a read returns. Whenever the abstract device holds data $d$ at address $a$ and page offset $p$, the concrete device maps the logical page $(a, p)$ to some physical page that is live and holds $d$, so reading $(a, p)$ returns $d$ on both.

\subsection{\srev{Proving the Global Security Contract}}
\label{sec:model:simulation}
\label{sec:model:oplevel}

\srev{\arbiter proves the invariant at both levels of \S\ref{sec:model:state}, i.e., the \emph{operation level} and the \emph{command level}. %
However, refinement is proved only at the \vrev{operation} level. It carries to the command level by \emph{agreement}, which we explain in \S\ref{sec:model:commandlevel}.}

\subsubsection{\srev{Operation Level}}
\label{sec:model:oplevel:ops} \arev{The verification guarantees are three composable theorems: (i) the invariant is preserved by every operation, %
% (ii) each operation matches the abstract device,
(ii) each operation \rrev{preserves the refinement relation}, and (iii) both properties compose over whole traces (Figure~\ref{fig:simulation}).}

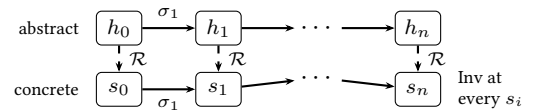
\begin{figure}[h]
\centering
\begin{tikzpicture}[
  node distance=0.35cm and 0.7cm,
  state/.style={draw, rounded corners=2pt, minimum width=0.6cm, minimum height=0.45cm, font=\footnotesize, inner sep=2pt},
  arr/.style={-{Stealth[length=3.5pt]}, thick}
]
\node[state] (h0) {$h_0$};
\node[state, right=of h0] (h1) {$h_1$};
\node[right=of h1, inner sep=2pt] (hdots) {$\cdots$};
\node[state, right=of hdots] (hn) {$h_n$};
\node[state, below=of h0] (s0) {$s_0$};
\node[state, below=of h1] (s1) {$s_1$};
\node[below=of hdots, inner sep=2pt] (sdots) {$\cdots$};
\node[state, below=of hn] (sn) {$s_n$};
\draw[arr] (h0) -- node[above,font=\scriptsize]{$\sigma_1$} (h1);
\draw[arr] (h1) -- (hdots);
\draw[arr] (hdots) -- (hn);
\draw[arr] (s0) -- node[below,font=\scriptsize]{$\sigma_1$} (s1);
\draw[arr] (s1) -- (sdots);
\draw[arr] (sdots) -- (sn);
\draw[arr,dashed] (h0) -- node[right,font=\scriptsize]{$\mathcal{R}$} (s0);
\draw[arr,dashed] (h1) -- node[right,font=\scriptsize]{$\mathcal{R}$} (s1);
\draw[arr,dashed] (hn) -- node[right,font=\scriptsize]{$\mathcal{R}$} (sn);
\node[font=\scriptsize, left=0.1cm of h0] {abstract};
\node[font=\scriptsize, left=0.1cm of s0] {concrete};
\node[font=\scriptsize, right=0.1cm of sn, align=left] {Inv at\\every $s_i$};
\end{tikzpicture}
%
% \caption{Trace-level simulation: each $\sigma_i$ matches between abstract ($h_i$) and concrete ($s_i$), with $\mathcal{R}$ and Inv preserved throughout.}
\caption{\rrev{Trace-level simulation. Each operation ($\sigma_i$) matches between abstract ($h_i$) and concrete ($s_i$), and $\mathcal{R}$ and the global invariant hold throughout.}}
\label{fig:simulation}
\end{figure}

\begin{theorem}[Invariant Preservation]
\label{thm:preservation}
For every state $s$, operation $\sigma$, and successor $s'$, if $\mathtt{ftl\_invariant}(s)$ and $\mathtt{step}(s, \sigma) = \mathtt{Some}\;s'$, then $\mathtt{ftl\_invariant}(s')$.
\end{theorem}

\begin{proof}
\rrev{By case analysis on $\sigma$, one lemma per operation, each re-establishing all 27 clauses. The write case is the largest, because an out-of-place write invalidates the old page, programs a new one, repairs both directions of the mapping, and advances the write frontier. The maintenance cases lean hardest on Inv$_{11}$, Inv$_{15}$, Inv$_{17}$, and Inv$_{24}$.}
\end{proof}

\begin{theorem}[Stepwise Simulation]
\label{thm:simulation}
Let $\mathtt{abs\_step}(h, \sigma)$ apply $\sigma$ to the abstract device, where \rrev{reads, tag writes, and the two maintenance operations leave} $h$ unchanged. For every $h$, $s$, $\sigma$, and $s'$, if $\mathcal{R}(h,s)$, $\mathtt{ftl\_invariant}(s)$, and $\mathtt{step}(s, \sigma) = \mathtt{Some}\;s'$, then $\mathcal{R}(\mathtt{abs\_step}(h, \sigma), s')$.
\end{theorem}

\begin{proof}
By case analysis on $\sigma$. Host-visible cases follow from matching write and invalidate semantics. Maintenance cases follow chiefly from Inv$_2$, Inv$_4$, and Inv$_{22}$.
\end{proof}

\begin{theorem}[\rrev{Refinement and Preservation Under Traces}]
\label{thm:refinement}
Let $\mathtt{abs\_exec}(h, \mathit{ops})$ apply $\mathit{ops}$ to $h$ via the abstract semantics, and let $\mathtt{exec}(s, \mathit{ops})$ fold $\mathtt{step}$ over $\mathit{ops}$, returning $\mathtt{None}$ if any step is rejected. %
For every trace $\mathit{ops}$, \rrev{if $\mathtt{ftl\_invariant}(s)$, $\mathcal{R}(h, s)$, and $\mathtt{exec}(s, \mathit{ops}) = \mathtt{Some}\;s'$, then} $\mathcal{R}(\mathtt{abs\_exec}(h, \mathit{ops}), s')$ and $\mathtt{ftl\_invariant}(s')$.
\end{theorem}

\begin{proof}
By induction on $\mathit{ops}$, applying Theorems~\ref{thm:preservation} and \ref{thm:simulation}.
\end{proof}

\rrev{Together, the three theorems give an end-to-end guarantee. We prove separately that the initial state $q_0$ satisfies every clause. Every state a valid execution reaches from $q_0$ therefore satisfies every clause of the global invariant, and the execution refines the abstract device's run of the same trace.}

\subsubsection{\srev{Command Level}}
\label{sec:model:commandlevel}

\srev{\arbiter fixes the command sequence for each of its six operations with
$\mathtt{op\_primitives} : \mathtt{FTLState} \rightarrow \mathtt{COp}
\rightarrow (\mathtt{list\;FlashPrimitive})?$. We prove two properties of this
expansion. (i) \emph{Agreement} shows that, from a state satisfying the global
invariant, an operation and its command sequence produce the same state,
connecting the operation-level proof in Theorem~\ref{thm:preservation} to the
command level. (ii) \emph{Realizability} shows that every program in the
sequence targets an erased page when it executes, ensuring that the sequence
can execute on NAND flash. Agreement alone does not provide this guarantee.}

\srev{These guarantees apply when the command sequence comes from
$\mathtt{op\_primitives}$. A faulty FTL can instead issue commands directly,
bypassing the operation already proved correct and its expansion. \arbiter therefore gives
each command a \emph{precondition bundle}, which is a set of decidable checks on the
command and the relevant FTL state, evaluated before the command executes
(Table~\ref{tab:bundles}).\drev{\footnote{The bundles apply to commands issued individually, not to intermediate commands in an operation's expansion. The final command restores the invariant, and
\S\ref{sec:security-model:scope} assumes that each operation runs to completion.}}
% Checking the global invariant is not sufficient for
% this purpose. It requires checking all addresses, pages, and blocks, and it
% describes the resulting state rather than whether a command is safe to issue.
%
% Theorem~\ref{thm:preconditions} proves that the precondition bundles are sufficient to preserve the required properties.
Theorem~\ref{thm:preconditions} proves that a command whose bundle holds \vrev{preserves the global invariant}. \S\ref{sec:attacks} shows that
each failure we formalize results from issuing a command without satisfying its
precondition bundle.}

\begin{theorem}[Precondition Soundness]
\label{thm:preconditions}
For each command $\pi$ with precondition bundle $\mathtt{pre}_\pi$ and
application function $\mathtt{apply}_\pi$,
\begin{align*}
  & \mathtt{ftl\_invariant}(s) \;\wedge\; \mathtt{pre}_\pi(s, \mathit{args}) \\
  & \quad \;\Rightarrow\; \mathtt{ftl\_invariant}(\mathtt{apply}_\pi(s, \mathit{args})).
\end{align*}
\end{theorem}

\begin{table}[h]
\centering
%
% \caption{Primary NAND flash instructions that manage the NAND flash device state and micro-operations that manage the metadata.}
%
% \caption{\srev{The ten commands and their precondition bundles. A command whose bundle holds preserves the global invariant (Theorem~\ref{thm:preconditions}).}}
\caption{\srev{The ten commands and their precondition bundles}~\drev{(commands that share a bundle are grouped). A command issued on its own preserves the invariant if its bundle holds (Theorem~\ref{thm:preconditions}).}}
\label{tab:bundles}
\footnotesize
\setlength{\tabcolsep}{4pt}
\renewcommand{\arraystretch}{1.05}
\begin{tabular}{@{}>{\raggedright\arraybackslash}p{0.30\columnwidth}>{\raggedright\arraybackslash}p{0.66\columnwidth}@{}}
\toprule
\textbf{\vrev{Command}} & \textbf{Precondition bundle} \\
\midrule
$\mathtt{PrimRead}$ & None. The \vrev{command} changes no state. %
An FTL that must block raw reads (\rrev{FS\#5}) adds an access-control check here. \\
$\mathtt{PrimProgram}$ & %
\fix{\rrev{Destination block and page in bounds.} The reverse mapping it writes names a logical page that the forward map sends to this page, and it carries an integrity tag. Every page below the destination in the block is written, and every page above it is erased and unstamped. If the program opens a block, the owner's previously open block has its first page written. Issued with no tag, it marks a page live with none.} \\
$\mathtt{PrimErase}$ & Target block in bounds, not in the free pool, not open for writing, and mapped by no live page. \\
\fix{$\mathtt{PrimFreePush}$} & %
\fix{Returns a block to the free pool. Target block in bounds, \rrev{erased and unstamped}, owned by no one, mapped by no live page, not open, and not already free. Issued on a free block, it lists the block twice.} \\
$\mathtt{PrimMapAddr}$, $\mathtt{PrimRemap}$ & Target page live and mapped by no other logical page. \\
$\mathtt{PrimInvalidate}$ & Target page live. \\
$\mathtt{PrimSetTag}$, barriers & None. $\mathtt{PrimSetTag}$ acts only on a live page, and the barriers change no state. \\
\bottomrule
\end{tabular}
\end{table}

\section{FTL Verification Workflow}
\label{sec:verify}

\arbiter verifies an FTL in one of two modes. (i) \emph{Direct verification}
(\S\ref{sec:verify:direct}) proves the FTL's operations correct and secure end-to-end
against the global invariant. %
(ii) \emph{Modular
verification} (\S\ref{sec:verify:modular}) applies to a design that keeps the
invariant. \vrev{The designer describes their FTL through the Rocq module type
$\mathtt{CUSTOM\_FTL}$ and discharges five hypotheses, rather than re-proving
the 27 clauses over their own state.} The designer supplies their FTL's state
and \vrev{its four operations as functions over that state}. The designer supplies a refinement map
that projects that state onto \arbiter's model. The designer
discharges the five hypotheses \vrev{stated} in Figure~\ref{fig:hypotheses} over the projected
state; how much proof this takes depends on what the design changed.
The designer then packages the five proofs into one
certificate theorem. \vrev{Rocq's kernel} and $\mathtt{coqchk}$ re-check
every theorem the certificate depends on.

\subsection{Direct Verification}
\label{sec:verify:direct}

In direct verification, the designer writes their FTL as a Rocq model over
$\mathtt{FTLState}$ and proves, for each of its operations, that the operation
preserves $\mathtt{ftl\_invariant}$ (Definition~\ref{def:ftl-invariant}) %
% and the refinement relation $R$ (Definition~\ref{def:refinement})
and the refinement relation $\srev{\mathcal{R}}$ (Definition~\ref{def:refinement}). %
% This is the proof \arbiter itself discharges once for its six operations (Theorems~\ref{thm:preservation}-\ref{thm:refinement}). A designer reuses these proofs whenever the design \emph{only} changes the use of invariants, e.g., deduplication, which removes injectivity (Inv$_2$) and rewrites the reverse-mapping clause (Inv$_3$).
% This is the proof \arbiter discharges once for its own six operations
CertiFlash performs these proofs once for its six operations
(Theorems~\ref{thm:preservation}--\ref{thm:refinement}). \srev{A design requires direct verification when it changes the invariant, because
the new invariant cannot reuse \arbiter's invariant-preservation proofs. For
example, deduplication maps identical logical pages to a single physical page
to avoid storing duplicate data. This removes injectivity (Inv$_2$) and
changes the reverse-mapping clause (Inv$_3$), so \arbiter's theorems no longer
apply. Deduplication therefore proves its operations against the modified
invariant (case~(d) of \S\ref{sec:evaluation:dftl}).
}

\subsection{Modular Verification}
\label{sec:verify:modular}

\begin{figure}[h]
\centering
{\footnotesize
Let $I(s) \triangleq \mathtt{ftl\_invariant}(\mathtt{user\_to\_model}(s))$, $\mathtt{ready} \triangleq \mathtt{user\_write\_ready}$, and $\mathtt{adm}(s, a, p) \triangleq a < \mathtt{addr\_space} \wedge p < \mathtt{pages\_per\_block} \wedge a$ is labelled with a tenant and a namespace in $s$.
\begin{align*}
\mathbf{Hyp1} \;\triangleq\;& \forall s, a, p, d.\; \mathtt{adm}(s, a, p) \wedge \mathtt{ready}(s, a, p) \wedge I(s) \\[-1pt]
& \quad \Rightarrow I(\mathtt{user\_write}(s, a, p, d)) \\
\mathbf{Hyp2} \;\triangleq\;& \forall s.\; I(s) \;\Rightarrow\; I(\mathtt{user\_gc}(s)) \\
\mathbf{Hyp3} \;\triangleq\;& \forall s.\; I(s) \;\Rightarrow\; I(\mathtt{user\_wear\_level}(s)) \\
\mathbf{Hyp4} \;\triangleq\;& \forall s, a, p, d.\; \mathtt{adm}(s, a, p) \wedge \mathtt{ready}(s, a, p) \wedge I(s) \\[-1pt]
& \quad \Rightarrow \mathtt{user\_read}(\mathtt{user\_write}(s, a, p, d), a, p) = \mathtt{Some}\;d \\
\mathbf{Hyp5} \;\triangleq\;& \forall s, a_1, a_2, p_1, p_2, d_1, d_2.\\[-1pt]
& \quad \arev{\text{let } s' = \mathtt{user\_write}(s, a_1, p_1, d_1) \text{ in}}\\[-1pt]
& \quad \arev{\text{let } s'' = \mathtt{user\_write}(s', a_2, p_2, d_2) \text{ in}}\\[-1pt]
& \quad (a_1 \neq a_2 \vee p_1 \neq p_2) \wedge \mathtt{adm}(s, a_1, p_1) \wedge \mathtt{adm}(s', a_2, p_2)\\[-1pt]
& \quad \wedge\; I(s) \wedge \mathtt{ready}(s, a_1, p_1) \wedge \mathtt{ready}(s', a_2, p_2) \\[-1pt]
& \quad \Rightarrow \mathtt{user\_read}(s'', a_1, p_1) = \mathtt{Some}\;d_1
\end{align*}}
\caption{The five hypotheses an FTL designer's $\mathtt{CUSTOM\_FTL}$ module must discharge. Hyp1--Hyp3: per-operation invariant preservation. Hyp4: read-after-write correctness. Hyp5: cross-address non-interference.}
\label{fig:hypotheses}
\end{figure}

\noindent\textbf{Step~1: Describing the FTL.}
The designer supplies a state type $\mathtt{user\_state}$, which may be any Rocq
type holding the FTL's runtime structures, four operations over it
($\mathtt{user\_read}$, $\mathtt{user\_write}$, $\mathtt{user\_gc}$,
$\mathtt{user\_wear\_level}$), and a readiness predicate
$\mathtt{user\_write\_ready}$ recording when a write is well-defined. The
interface covers the four operations a new design is most likely to
reimplement. Invalidation and tag setting are fixed metadata bookkeeping
rather than per-design policy, so a custom FTL keeps \arbiter's semantics for
them and Theorem~\ref{thm:preservation} covers them unchanged. 
% \drev{The interface
% cannot check this. It supplies $\mathtt{user\_to\_model}$ and no inverse, so a design
% whose state derives from the model, such as DFTL's cached mapping table, must
% reimplement the two operations and show its derived state stays consistent with them.}

%
% Every operation, \arbiter's or the designer's, is a sequence of the commands the device executes. Three of them are the NAND-flash hardware commands~\cite{onfi}: read, program, and erase. Program also writes the page's out-of-band area (\S\ref{sec:prelim:ftl}), which the FTL supplies; a correct FTL keeps that area consistent with the mapping, and a modified one can write anything, e.g., a page with no integrity tag. The remaining commands are \emph{metadata updates} on FTL-resident tables that never change the data a page holds: install or redirect a mapping, mark a page stale, write its integrity tag, return a block to the free pool, and bracket a compound operation. Table~\ref{tab:bundles} lists the details for all of them.

% The function
% $\mathtt{op\_primitives}$ computes this sequence for each of \arbiter's
% operations, and two machine-checked properties tie the two descriptions
% together: on any state satisfying the global invariant, running an operation
% and running its command sequence reach the same state (\emph{agreement}), and
% every program sequence issues targets a page that is erased at that
% moment (\emph{realizability}). A designer's operation is described the same
% way, so the results below apply to it directly.

\noindent\textbf{Step~2: Projecting onto the model.}
The designer supplies $\mathtt{user\_to\_model} : \mathtt{user\_state}
\rightarrow \mathtt{FTLState}$, which projects the FTL's state onto the core
model. The inherited invariant is then
$I(s) \triangleq \mathtt{ftl\_invariant}(\mathtt{user\_to\_model}(s))$, and
every hypothesis in Step~3 is stated over $I$. A design that adds fields
\arbiter does not know about %
(e.g., DFTL's cached mapping table, \srev{case~(b) of \S\ref{sec:evaluation:dftl}})
projects them away; a design that only re-implements how existing fields are
computed projects onto the same fields.

\noindent\textbf{Step~3: Discharging the five hypotheses.}
The five hypotheses of Figure~\ref{fig:hypotheses} are properties \arbiter's own
model satisfies, so an FTL that satisfies them over
its own operations obtains invariant preservation over its own operations. 
% \drev{They do not give refinement: Hyp$_1$--Hyp$_3$ constrain a single state, so a garbage collection that discarded live data would satisfy Hyp$_2$. Hyp$_4$ and Hyp$_5$ constrain writes and reads only, so a design that changes what maintenance does must argue data preservation for itself.}
Hyp$_1$--Hyp$_3$ state that
write, garbage collection, and wear leveling preserve $I$.\footnote{Read
returns the state unchanged, so it has no post-state to constrain, and the
interface accordingly gives $\mathtt{user\_read}$ no result state.}
Hyp$_4$ and Hyp$_5$ are the two data-correctness claims: a read returns the
data last written to its address, and a write to one address does not disturb
another. How the designer discharges them depends on what the design changed.
\srev{A design changes the framework in one of four ways, and
\S\ref{sec:evaluation:dftl} evaluates one of each: (a)~\emph{new state},
(b)~\emph{new implementation}, (c)~\emph{\drev{new local invariant}}, and (d)~\emph{new
invariant}. Cases (a)--(c) keep \arbiter's global invariant and take the
modular route. Case~(d) changes it, so it takes direct verification
(\S\ref{sec:verify:direct}).}

\srev{A design can also add new operations. $\mathtt{CUSTOM\_FTL}$ does not
require a command sequence for them: its four operations are functions over
$\mathtt{user\_state}$, and the five hypotheses establish their correctness.
When a new operation changes only the design's own state, as the \drev{counter-increment}
operation in case~(a), \arbiter's reuse library carries these proofs to the
extended FTL. When it changes \arbiter's state, the designer instead expresses
the operation as a command sequence (Step~1) and proves that each command
satisfies its precondition bundle (\S\ref{sec:model:commandlevel},
Table~\ref{tab:bundles}). Theorem~\ref{thm:preconditions} then preserves
$\mathtt{ftl\_invariant}$ across the sequence, without proving its 27 clauses
again. \arbiter supplies the per-command results needed for this proof, but
does not generate the command sequence.}

\noindent
\srev{The same argument applies when a faulty FTL issues commands directly.
Checking each command's precondition bundle preserves
$\mathtt{ftl\_invariant}$ after every command.}
% Each bundle
% has an executable boolean form in the artifact.

% \emph{(a)~New state, no new behavior.} The design adds fields that no base
% operation reads or writes, e.g., a per-block read-disturb
% counter~\cite{cai2017error}. A lifting functor carries the base FTL's five
% proofs across such an extension. The designer supplies a projection onto the
% base state and one equation per operation stating that the operation commutes
% with it.
% % each hypothesis then closes in four to nine proof lines, and in
% % case~(a) of \S\ref{sec:evaluation} the equations are definitions with no proof
% % text.

% \emph{(b)~Re-implemented operation, same invariant.} The design changes how an
% existing operation computes its result, e.g., DFTL~\cite{gupta2009dftl}
% replaces the in-DRAM L2P table with a demand-paged cache. The designer proves
% Hyp$_1$--Hyp$_5$ directly, typically by stating one local invariant over the
% new fields (DFTL's $\mathtt{dftl\_ok}$, case~(b)) and showing that the
% re-implemented operation preserves both it and $I$. Because the refinement
% map reduces every clause of $I$ to \arbiter's own preservation theorem. A local invariant
% is not free: \arbiter's unmodified maintenance operations need not preserve
% it, which is how case~(c) found a fault in MegIS~\cite{Ghiasi2024MegIS}.

\noindent\textbf{Step~4: Collecting the certificate.}
With the five hypotheses discharged, the designer \srev{instantiates one
module}, $\mathtt{Validator}(F : \mathtt{CUSTOM\_FTL})$, which packages the five
proofs into a single certificate theorem,
$\mathtt{custom\_ftl\_security\_suite}$. \srev{$\mathtt{Validator}$} proves
nothing itself. It only assembles the proofs the designer supplied, so
\srev{instantiating} it cannot fail and adds no proof obligation.

\noindent\textbf{Step~5: Checking everything.}
The certificate is a Rocq theorem, so Rocq's kernel checks it, and
$\mathtt{coqchk}$ re-checks the compiled modules against the kernel and reports
every axiom they depend on. \arbiter's own development reports none: no
admitted lemmas, no classical reasoning, and no axioms. A designer's certificate inherits that status
exactly when their own proofs add no axioms, which the same check confirms.
We evaluate the verification effort of these cases in \S\ref{sec:evaluation} following this workflow.
\section{Evaluation}
\label{sec:evaluation}

\subsection{Evaluation Methodology}

% \arbiter proves its global properties once, and every FTL that implements $\mathtt{CUSTOM\_FTL}$ inherits them. A design that keeps the global invariant discharges the five $\mathtt{CUSTOM\_FTL}$ hypotheses about its own operations and inherits the invariant-preservation and refinement proofs, so it restates no global theorem. A design that changes the global invariant cannot use the 
% interface and verifies its model directly.

% \textbf{System Configuration.} We build and check every proof with \coq~8.20.1 on an Apple M5 Pro laptop. A clean serial build of the development compiles in \rrev{85} seconds, and \texttt{coqchk} re-checks each compiled module against the kernel and reports no axioms, no admitted lemmas, and no classical reasoning.

 We build and check every proof with \coq~8.20.1 on an Apple M5 Pro laptop. \drev{A clean build takes 93 seconds.} The development is 28 \coq files and \numupd{21{,}884} lines, with \rrev{432} \texttt{Definition}s and 778 \texttt{Lemma}/\texttt{Theorem}/\texttt{Corollary} statements, every one closed with \texttt{Qed}. The base framework that every design inherits is \drev{16{,}489} of those lines. Table~\ref{tab:reuse-cost} divides the base framework by component, from the concrete model and operations, through the global invariant and its preservation, to the checker and the failure surfaces. The four case studies (\S\ref{sec:evaluation:dftl}) build on this base: each discharges the five \texttt{CUSTOM\_FTL} hypotheses to inherit these proofs rather than re-proving the invariant over its own state.

\begin{table}[h]
\centering
%
% \caption{Division of proof labor: lines inherited from the framework versus lines written per FTL.}
\caption{Lines the framework proves once, inherited by every FTL, by component.}
\label{tab:reuse-cost}
\footnotesize
\setlength{\tabcolsep}{4pt}
\renewcommand{\arraystretch}{1.0}
\begin{tabular}{@{}llr@{}}
\toprule
 & Component & Lines \\
\midrule
\multirow{10}{*}{\shortstack[l]{Framework\\(proved once,\\inherited by\\every FTL)}}
 & Concrete model, operations, geometry & \numupd{786} \\
 & \srev{Command level}: boundary, agreement, & \\
 & \quad realizability & \rrev{1{,}862} \\
 & Global invariant, preservation, and & \\
 & \quad precondition bundles & \rrev{5{,}834} \\
 & Refinement + abstract specification & 1{,}212 \\
 & \rrev{Crash recovery and crash refinement} & \rrev{1{,}716} \\
& Designer interface + composition library & \drev{817} \\
 & Checker, test harness, tactic pack & 3{,}088 \\
 & Failure surfaces and reachability & \numupd{1{,}174} \\
\cmidrule(l){2-3}
& \textbf{Total inherited} & \textbf{\drev{16{,}489}} \\
\bottomrule
\end{tabular}
\end{table}

% ============================================================================
% 6.  Reachability of Failures  --  tightened
%
% Collapses the old 6 / 6.1 / 6.2 / 6.2.1 / 6.2.2 into: framing paragraph,
% table, one worked failure (FS#2), one closing paragraph. Terminology follows
% the 5.4 pass: "command", not "instruction"; Prim* names appear only where
% they name an artifact identifier.
%
% Revision macros (\nrev, \prev, \rrev) are NOT reassigned here -- the merges
% moved text across their old spans. Re-mark once the wording is settled.
% ============================================================================

\subsection{Reachability of Failures}
\label{sec:attacks}

\drev{\S\ref{sec:model:commandlevel}} proves the precondition bundles \emph{sufficient}: a command
whose bundle holds preserves the global invariant
(Theorem~\ref{thm:preconditions}). This section proves them \emph{necessary}.
For each of the five failure surfaces of \S\ref{sec:motivation:failures}, we
exhibit a single command that an FTL can issue with its bundle unchecked and
prove in \coq that the resulting state falsifies the invariant. We describe this using ten failures for five attack surfaces in Table~\ref{tab:theorems} (one row per failure), each
naming the command that produces it and a clause it falsifies.

A failure is \emph{reachable} when a single FTL action $f$ on
$\mathtt{FTLState}$ carries a well-formed state, i.e., one satisfying
$\mathtt{ftl\_invariant}$, to a state the invariant rejects. We prove two
results per failure. First, for every state meeting the hypothesis stated in
its theorem, e.g., the victim page being mapped or the target page being
erased, the command that performs $f$ is realizable and $f(s)$ falsifies the
invariant. Second, one concrete state satisfying all 27 clauses reaches such an
$f(s)$ at all five surfaces, so no failure is an artifact of a malformed start.

\begin{table}[h]
\centering
\caption{The ten failures at the five failure surfaces, each produced by a
single command issued with its precondition bundle unchecked. The Clause column
gives the clause each proof uses; the state reached can falsify others as well.
Every failure is conditional on a hypothesis stated in its theorem.}
\label{tab:theorems}
\footnotesize
\setlength{\tabcolsep}{4pt}
\renewcommand{\arraystretch}{1.0}
\begin{tabular}{@{}llll@{}}
\toprule
FS & Violation & Command & Clause \\
\midrule
FS\#1 & Erase of a mapped block          & $\mathtt{PrimErase}$    & Inv$_5$  \\
FS\#1 & Erase of an out-of-range block   & $\mathtt{PrimErase}$    & Inv$_8$  \\
FS\#2 & L2P aliasing, cross-address      & $\mathtt{PrimMapAddr}$  & Inv$_2$  \\
FS\#2 & L2P aliasing, intra-address      & $\mathtt{PrimMapAddr}$  & Inv$_2$  \\
FS\#2 & Namespace reassignment           & $\mathtt{PrimMapAddr}$  & Inv$_7$  \\
FS\#3 & Ownership override               & $\mathtt{PrimMapAddr}$  & Inv$_7$  \\
FS\#3 & Namespace-enforcement override   & $\mathtt{PrimMapAddr}$  & Inv$_7$  \\
FS\#4 & Integrity-tag removal            & $\mathtt{PrimProgram}$  & Inv$_9$  \\
FS\#5 & Free-block duplication           & $\mathtt{PrimFreePush}$ & Inv$_{11}$ \\
FS\#5 & Cross-namespace remap            & $\mathtt{PrimMapAddr}$  & Inv$_7$  \\
\bottomrule
\end{tabular}
\end{table}

Together the ten proofs name six distinct clauses, i.e., Inv$_2$, Inv$_5$,
Inv$_7$, Inv$_8$, Inv$_9$, and Inv$_{11}$ (Table~\ref{tab:invariants}). Four of
the ten falsify Inv$_7$ through an unchecked $\mathtt{PrimMapAddr}$, and three
of those four have identical theorem statements. We list them separately
because a surface names the part of the SSD the FTL reaches through
(\S\ref{sec:motivation:failures}), and an FTL can issue this command at three
of the five surfaces.

\subsubsection{An Example: L2P Aliasing at FS\#2}
\label{sec:attacks:proof}

From a state in which page $p$ of a victim address $a_v$ maps to physical page
$\mathit{pa}$, an FTL issues $\mathtt{PrimMapAddr}(a_a, p, \mathit{pa})$, which
maps page $p$ of the attacker address $a_a$ to the same physical page.

\begin{theorem}[FS\#2 L2P aliasing]
\label{thm:as2}
For every state $s$ with $\mathtt{l2p\_map}(a_v, p) = \mathtt{Some}\;\mathit{pa}$
and $a_a \neq a_v$, the command $\mathtt{PrimMapAddr}(a_a, p, \mathit{pa})$ is
realizable in $s$, and the state it reaches falsifies
$\mathtt{ftl\_invariant}$.
\end{theorem}

\begin{proof}
An FTL that skips the bundle of $\mathtt{PrimMapAddr}$ can issue it in any
state, and the command is not a program, so it is realizable in every state
\drev{(\S\ref{sec:model:commandlevel})}. After it, the logical pages $(a_a, p)$ and
$(a_v, p)$ both map to $\mathit{pa}$, since the update at $(a_a, p)$ leaves
every other entry untouched. As $a_a \neq a_v$, injectivity (Inv$_2$) fails at
witness $((a_a, p), (a_v, p), \mathit{pa})$.
\end{proof}

\noindent The aliased mapping is a confidentiality breach: a read of the
attacker's page returns the victim's data, as \S\ref{sec:attacks:board}
demonstrates on real firmware. \arbiter rejects it at the command that causes
it. The bundle of $\mathtt{PrimMapAddr}$ (Table~\ref{tab:bundles}) requires the
target page to be live and mapped by no other logical page, and in the state of
Theorem~\ref{thm:as2} the victim's entry already maps $(a_v, p)$ to
$\mathit{pa}$, so the bundle rejects the command before it executes.

\subsubsection{Coverage}
\label{sec:attacks:prevention}

The remaining nine failures are proved and prevented the same way, and either
route to a verified FTL rules out all ten. An FTL built from \arbiter's
operations starts well-formed and never reaches any of the ten states, because
every operation preserves the invariant
(Theorem~\ref{thm:preservation}). An FTL that issues commands directly, but
gates each on its bundle, reaches none of them either, because a command whose
bundle holds preserves the invariant (Theorem~\ref{thm:preconditions}).

% \vspace{2pt}
% \noindent\textbf{Key Takeaway \mk{[N]}.}
% \emph{Every clause the ten proofs name is necessary: dropping it would admit a
% state an FTL can reach with one unchecked command. The bundles are therefore
% not merely sufficient for the invariant but load-bearing at each of the five
% surfaces, and an FTL that checks them exhibits no failure of
% Definition~\ref{def:failure} at any surface.}
\begin{brevblock}
% \section{Evaluation: \back{Proof Reuse Across Four Designs}}
% \label{sec:evaluation}

% \subsection{Evaluation Methodology}
% \label{sec:evaluation:method}

\subsection{Case Studies}
\label{sec:evaluation:dftl}

% \arbiter proves its global properties once, and every FTL that implements $\mathtt{CUSTOM\_FTL}$ inherits them. A design that keeps the global invariant discharges the five $\mathtt{CUSTOM\_FTL}$ hypotheses about its own operations and inherits the invariant-preservation and refinement proofs, so it restates no global theorem. A design that changes the global invariant cannot use the interface and verifies its model directly.

% The mechanized development comprises 28 \coq files and \numupd{21{,}968} lines, containing \rrev{432} \texttt{Definition}s and 778 \texttt{Lemma}/\texttt{Theorem}/\texttt{Corollary} statements, every one closed with \texttt{Qed}. A clean serial build takes \rrev{85} seconds on an Apple M5 Pro laptop, and \texttt{coqchk} re-checks the compiled modules against the kernel \rrev{and reports no axioms}. \fix{The development has no admitted lemmas, no classical reasoning, and no axioms, which \coq's assumption check confirms for every theorem this section reports.}

\arbiter proves its global invariant once, and every design that keeps it inherits the proof by discharging five hypotheses rather than re-proving the 27 clauses. We evaluate this on four FTL designs, one for each way a design departs from the framework: it adds orthogonal state (a), reimplements an operation (b), adds an operation (c), or changes the global invariant (d). The first three keep the invariant and reuse its proofs without restating a clause; the fourth changes the invariant and verifies its model directly. Across the four, a designer writes \numupd{27} to 3{,}231 lines against the framework's \drev{16{,}489}, and \drev{the effort tracks the local invariant a design adds.} We present each design in the same form. 
% The tracker adds none and costs 27 lines; DFTL adds $\mathtt{dftl\_ok}$ and costs 844; MegIS adds $I_{\text{range}}$ on top of DFTL and costs 1{,}293; deduplication rewrites the global invariant itself and costs 3{,}231.

\noindent \textbf{(a) New state.}
The designer adds a per-block read counter that records a block's accumulated read-disturb exposure. FTLs keep such counters to refresh a block before read disturbance corrupts the pages neighboring its heavily read pages~\cite{cai-dsn-2015, cai-insidessd-2018}. The extension adds three definitions: (i) the counter, (ii) a \drev{counter-increment operation}, which the controller runs to account for each read, and (iii) a refresh check that fires once a block's count crosses a threshold.
\begin{itemize}[nosep,leftmargin=*]
\item \fix{\emph{Verification Type:} Modular. The read-disturb counter does not affect FTL correctness. No base operation
reads it, and reads return the same data regardless of its value.
The global invariant also does not mention the counter. The five
$\mathtt{CUSTOM\_FTL}$ hypotheses therefore carry over unchanged.}
\item \fix{\emph{Effort.} The read-disturb tracker adds 27 lines on top of \arbiter. \drev{Its three definitions and the extension-state declaration take 25 lines.}
The remaining two lines prove that incrementing the counter preserves the
global invariant and leaves reads unchanged. Thus, nearly all of the cost is
in defining the new state and operations; the tracker requires no new global
invariant.}
\end{itemize}

\noindent \textbf{(b) New implementation.}
% \item \emph{Design.} 
Demand-based Flash Translation Layer (DFTL)~\cite{gupta2009dftl} is an example of a new FTL implementation that preserves \arbiter's global invariant.
DFTL keeps only a working set of the L2P mapping in DRAM, in a cached mapping table (CMT), and stores the remaining entries in translation pages on flash. A read whose mapping is not in the CMT fetches the mapping from its translation page. Before evicting a dirty mapping from the CMT, DFTL writes it back to flash. Thus, DFTL changes how the FTL stores and looks up the L2P mapping.

\begin{itemize}[nosep,leftmargin=*]
\item \emph{Verification Type:} Modular. DFTL adds the CMT and an image of the translation pages to \arbiter's flash state. Writes, garbage collection, and wear leveling use \arbiter's verified operations to update the flash state and then update the CMT. \arbiter's global invariant constrains only the flash state, and DFTL's refinement map returns that state directly, so DFTL reuses the invariant unchanged. The designer proves one local invariant that connects the CMT and translation-page image to the flash state. It requires an invariant $\mathtt{dftl\_ok}$ that (i) the CMT is bounded, (ii) every cached entry agrees with $\mathtt{l2p\_map}$, (iii) every clean entry agrees with the translation-page image, and (iv) the translation-page image provides the mapping for every uncached page. The new proof shows that eviction and write-back preserve $\mathtt{dftl\_ok}$. The five $\mathtt{CUSTOM\_FTL}$ hypotheses then follow. \arbiter's verified operations preserve the global invariant. $\mathtt{dftl\_ok}$ keeps the CMT consistent with $\mathtt{l2p\_map}$, so DFTL reads the correct physical page.
\item \emph{Effort.} DFTL adds 844 lines on top of \arbiter. The cache model and DFTL operations take 509 lines, and proving that these operations preserve $\mathtt{dftl\_ok}$ takes 192 lines. The remaining 143 lines discharge the five $\mathtt{CUSTOM\_FTL}$ hypotheses. DFTL does not re-prove \arbiter's 27-clause global invariant. Instead, each operation invokes \arbiter's preservation theorem once to preserve the entire invariant. Thus, most of DFTL's verification effort goes toward its new cache state and $\mathtt{dftl\_ok}$, rather than re-verifying properties already proved by \arbiter.
\end{itemize}

\noindent \textbf{(c) \drev{New local invariant}.}
MegIS~\cite{Ghiasi2024MegIS} performs metagenomic analysis inside the SSD, which reads long contiguous ranges of logical addresses sequentially. MegIS adds coarse-grained \emph{range mappings} to the FTL, one per such range. A range mapping covers the whole range with a single entry: a logical start address $a$, a physical start block $b$, and a length $\ell$, recording that $[a,a+\ell)$ is stored contiguously in the physical blocks $[b,b+\ell)$, so a scan reads the range without a per-page lookup. 
% A range mapping is a cached claim about the L2P map, and it is correct only while the map agrees with it.
\begin{itemize}[nosep,leftmargin=*]
\item \emph{Verification Type:} Modular, as an extension of DFTL. Production FTLs cache their mapping tables, as DFTL does~\cite{gupta2009dftl,
sun2023leaftl}, making DFTL a realistic base for an in-storage system. Building
MegIS on DFTL also shows that verified extensions can stack. MegIS reuses
DFTL's flash state and \arbiter's verified operations, and adds one local
invariant, $I_{\text{range}}$. This invariant requires each range mapping to
agree with the L2P map for every page in the range. Writes, garbage collection,
and wear leveling can relocate a page in a mapped range, making the range
mapping stale. MegIS therefore drops any affected range mapping before these
operations and proves that the resulting operations preserve
$I_{\text{range}}$. Host operations do not use range mappings, so MegIS reuses
DFTL's five $\mathtt{CUSTOM\_FTL}$ hypotheses unchanged. \drev{The scan path that
motivates the design is verified separately: resolving an address through a range
mapping returns the same physical page the L2P map would, whenever
$I_{\text{range}}$ holds.}
\item \emph{Effort.} MegIS adds \numupd{1{,}293} lines on top of DFTL. Most model the range mappings and prove that garbage collection, wear leveling, and write preserve $I_{\text{range}}$. The five hypotheses take only \numupd{13} lines, against DFTL's \numupd{143}, because those lines \drev{re-instantiate} DFTL's hypothesis proofs over MegIS's state \drev{rather than redo them}. MegIS re-proves neither the global invariant nor DFTL's $\mathtt{dftl\_ok}$; its new work is confined to $I_{\text{range}}$ and its preservation. Thus the effort concentrates on the one local invariant the design adds, not on the guarantees it inherits.
\end{itemize}

\noindent \textbf{(d) \drev{Changed global invariant}.}
Deduplication~\cite{chen2011caftl} stores identical data from multiple logical
pages in a single physical page. This sharing conflicts with the global
invariant. Inv$_2$ requires the L2P mapping to be injective, while
deduplication maps multiple logical pages to the same physical page.
\begin{itemize}[nosep,leftmargin=*]
\item \emph{Verification Type:} Direct. The modular interface does not apply to deduplication because it reuses
\arbiter's global invariant, while deduplication changes it. We therefore
model deduplication separately and prove preservation of the modified invariant
end-to-end. Three changes are necessary. First, supporting shared pages requires
removing injectivity (Inv$_2$) and changing the reverse-mapping clause
(Inv$_3$). Removing Inv$_2$ alone is insufficient: Inv$_3$ together with
Inv$_{22}$ implies injectivity, preventing a deduplicating write. We therefore
remove Inv$_2$ and modify Inv$_3$ to allow several logical pages to share one
physical page. Second, sharing changes relocation. \arbiter normally uses a
page's reverse mapping to find the logical page to relocate. A shared page has
multiple logical mappings but only one reverse mapping, so this procedure would
update one mapping and leave the others pointing to the erased page.
Deduplication instead scans the mapping table and updates every logical page
that shares the physical page. It must also update a shared page's reverse
mapping when the logical page recorded there is overwritten, which cannot be
done in place on NAND flash. Third, sharing remains subject to isolation.
Inv$_7$ requires all logical pages sharing a physical page to have the same
tenant and namespace. Deduplication can therefore share pages only within a
tenant's namespace.
\item \emph{Effort.} Deduplication adds 3{,}231 lines, making it the most expensive case.
\rrev{Twenty-five of the 27 invariant clauses and 47 framework lemmas carry
over unchanged}, while re-proving the remaining results takes 1{,}691 lines.
About half of these re-proved lines are identical to the framework proof they
replace. These proofs cannot be reused because each preservation theorem is
stated for a specific invariant, even when its proof remains unchanged.
Abstracting preservation over the invariant would make them reusable.
\end{itemize}

\end{brevblock}

\section{Related Work}
\label{sec:related}
\label{sec:other-related}

\noindent \textbf{FTL and flash file-system verification.}
Qiao et al.~\cite{qiao2019formal} model FTL algorithms in \coq as a state machine over mapping tables, free-block pools, and per-page metadata, define functional correctness as refinement to a disk-like specification, and demonstrate the framework on the hybrid-mapped BAST FTL~\cite{kim2002space}. SCFTL~\cite{chang2020determinizing} is a checkpointing FTL verified for \emph{snapshot consistency}, i.e., recovery to the state at the last flush, using Agda, the Serval symbolic executor, and Z3. Flashix~\cite{schierl2009ubifs, schellhorn2014verified, ernst2015inside, bodenmuller2021flashix} is a concurrent, crash-safe flash file system verified in KIV~\cite{reif1995kiv} as a stack of refinement-connected components, including garbage collection and wear leveling. Cogent~\cite{amani2016cogent} compiles a linearly typed language to C with a machine-checked refinement proof. None of these proofs transfers to a new design. SCFTL, Flashix, and Cogent each verify one fixed implementation, and Flashix's layers are modular within Flashix rather than across designs. 
\rrev{Qiao et al.\ build a framework in which an FTL that discharges five hypotheses over a supplied invariant refines a disk, and demonstrate it on BAST. Two of those five, read-after-write and non-interference between pages, are close to \arbiter's Hyp4 and Hyp5. Their invariant is a parameter, so a second FTL supplies its own and re-proves all five, and their model has three operations, i.e., initialize, read, and write, so garbage collection and wear leveling lie outside it.} \arbiter instead fixes the invariant (Definition~\ref{def:ftl-invariant}), and every FTL that implements $\mathtt{CUSTOM\_FTL}$ inherits it \srev{(\S\ref{sec:verify:modular})} at the cost measured in \S\ref{sec:evaluation}.

\rrev{\noindent \textbf{Verifying flash below the file system.}
A decade of work on the verified-flash-filestore challenge~\cite{joshi2007mini} sits below the file system and above the chip. Butterfield et al.~\cite{butterfield2009mechanising} mechanize a model of the NAND command set, the layer our \srev{commands} abstract. Kang and Jackson~\cite{kang2008alloy} check a flash file system with wear leveling and reclamation against an abstract device in Alloy, Damchoom and Butler~\cite{damchoom2009eventb} develop a flash filestore by Event-B refinement, and Taverne and Pronk~\cite{taverne2009raffs} model check one in SPIN under power loss. Closest to an FTL, Pf\"ahler et al.~\cite{pfahler2013erase} verify Flashix's erase-block-management layer, i.e., block mapping, wear leveling, and power-failure recovery, in KIV. Model checking has also been applied to shipping FTL-layer firmware: Kim et al.~\cite{kim2008flashdriver} check Samsung's OneNAND driver, Tripathy et al.~\cite{tripathy2019nand} model the NAND chip interface, and TxFlash~\cite{prabhakaran2008txflash} model checks the commit protocol of an SSD that exports atomic writes. All of these check one design against a functional specification. None states a property that a second, differently-designed FTL inherits, and none constrains tenant ownership.}

\noindent \textbf{Properties beyond functional correctness.}
The work above verifies functional correctness alone. Tripathy et al.~\cite{tripathy2023formal} model-check security properties of one ransomware-resistant FTL rather than proving a contract that other FTLs can inherit. \arbiter additionally proves that every operation preserves tenant isolation, integrity-tag coverage, and block ownership (Theorem~\ref{thm:preservation}),
and, \rrev{to our knowledge, is the first machine-checked proof of these properties for an FTL, the first such contract that a second FTL design inherits, and the first formalization of where an FTL can violate them}, as five failure surfaces (\S\ref{sec:security-model:surfaces}) with a reachable \prev{failure} on each (\S\ref{sec:attacks}).

\noindent \textbf{Layered, concurrent, and crash-safe storage verification.}
Argosy~\cite{chajed2019argosy} composes layered storage proofs under \emph{recovery refinement}, Perennial~\cite{chajed2019perennial} extends this to concurrent, crash-safe systems in Iris, and GoJournal~\cite{chajed2021gojournal} and the work of Hance et al.~\cite{hance2020storage} verify a concurrent journal and a key-value store. These frameworks compose proofs across the layers of one stack and fix no property that competing implementations of a single layer must satisfy, so they do not transfer a guarantee from one FTL to another. \arbiter models one FTL invocation at a time. Perennial is the most directly applicable infrastructure for adding concurrency.

\noindent \textbf{File systems verified above a block device.}
FSCQ and DFSCQ~\cite{chen2015fscq, chen2017dfscq} verify a UNIX-like file system in \coq with Crash Hoare Logic, Yggdrasil~\cite{sigurbjarnarson2016yggdrasil} verifies the Yxv6 journaling file system by push-button SMT, and DiskSec~\cite{ileri2018disksec} proves data confidentiality for the sealed-block file system SFSCQ. All three treat the device below them as a trusted layer and state no property about it, which is the layer \arbiter's invariant constrains. DiskSec's confidentiality guarantee is stronger than \arbiter's reachability results, but applies to one system rather than to any FTL meeting an interface.

\noindent \textbf{Verified kernels and compilers.}
seL4~\cite{sel4, murray2013sel4} and CertiKOS~\cite{certikos} verify a fixed kernel with the same per-operation preservation pattern as \arbiter, and CompCert~\cite{compcert} proves a C compiler correct via per-pass forward simulations. \arbiter uses a single forward simulation to an idealized block device, with maintenance operations absorbed as stuttering steps, and is parametric across FTLs rather than tied to one system.

\noindent \textbf{Information-flow security.}
Goguen and Meseguer~\cite{goguen1982security} define \emph{noninterference}, usually discharged through a per-call unwinding condition. 
\arbiter's isolation guarantees are state properties (\fix{Inv$_7$, \rrev{Inv$_{17}$}, Inv$_{18}$}), not information-flow properties.

\noindent \textbf{SSD security mechanisms.}
NVMe \emph{namespaces}~\cite{nvmenamespaces}, \emph{NVM sets}~\cite{nvme_base_spec}, and AES-256 self-encryption~\cite{fips197} provide firmware-resident isolation and at-rest confidentiality without formal verification, and commercial self-encrypting drives have shown critical weaknesses~\cite{meijer2019self}. FlashVault~\cite{flashvault} moves encryption into NAND, IceClave~\cite{kang2021iceclave} \rrev{places flash management in a TEE and protects it \emph{from} untrusted offloaded programs, the opposite trust direction to ours}, and FlashGuard~\cite{huang2017flashguard} and SSD-Insider~\cite{baek2018ssdinsider} keep stale pages in the FTL to recover data after a ransomware attack. \rrev{A second line takes the drive itself as the adversary: D-Shield~\cite{chowdhuryy2023dshield} moves encryption and integrity verification to the processor so that SSD firmware sits outside the trusted base, SGX-SSD~\cite{ahn2020sgxssd} anchors in-drive versioning in a host enclave, and eNVMe~\cite{wertenbroek2024pandora} builds a malicious NVMe firmware on real hardware. These works defend a host against a drive it does not trust; \arbiter instead gives the FTL designer a proof that their own firmware cannot reach a violating state. Cross-tenant leakage has also been shown through the device without any FTL modification, by timing SSD contention~\cite{juffinger2025secret} and the host memory buffer that holds the mapping table~\cite{juffinger2025hmb}; those channels are outside our model (\S\ref{sec:security-model:scope}).}
\arbiter is complementary. 
\rrev{Its precondition bundles are decidable conjunctions over FTL state, and each has an executable boolean form in the artifact. An FTL that gates every command on its bundle, extended with, e.g., an NVMe-namespace check, preserves the invariant at each one (Theorem~\ref{thm:preconditions}). We prove the bundles sufficient; we do not build the firmware that evaluates them.}

\section{Conclusion}
\label{sec:conclusion}

\arbiter is a \coq-mechanized verification framework for customizable Flash Translation Layers. It defines a \brev{27-clause global invariant over a 16-field FTL state built on a page-granular L2P mapping table and the per-page reverse mapping that NAND's out-of-band area carries}, proves invariant preservation and refinement to an idealized block device, 
and packages these proofs behind a \coq module type so \vrev{a compliant FTL discharges five hypotheses rather than redoing them}. We additionally formalize five failure surfaces and show that each failure is reachable by a single command issued without its precondition bundle, making any certified FTL provably free of every one. We \rrev{demonstrate} the failures on a DaisyPlus OpenSSD. Across four case studies, a designer adds \numupd{27} to 3{,}231 lines \rrev{against a \drev{16{,}489}-line framework}.
\arev{The proofs rest on the \coq kernel and standard library alone, with no admitted lemmas, no classical reasoning, and no axioms.}

% sections/appendix-bast.tex exists but is deliberately NOT input: it is held out
% of the submission build. Delete it or \input it before the camera-ready.

\bibliographystyle{unsrt}
\bibliography{bib/references}

\begin{thebibliography}{100}

\bibitem{mutlu.imw13}
Onur Mutlu.
\newblock {Memory Scaling: A Systems Architecture Perspective}.
\newblock {\em IMW}, 2013.

\bibitem{kanev.isca15}
Svilen Kanev, Juan~Pablo Darago, Kim Hazelwood, Parthasarathy Ranganathan, Tipp
  Moseley, Gu-Yeon Wei, and David Brooks.
\newblock {Profiling a Warehouse-Scale Computer}.
\newblock In {\em ISCA}, 2015.

\bibitem{wang.micro2016}
Shibo Wang and Engin Ipek.
\newblock {Reducing Data Movement Energy via Online Data Clustering and
  Encoding}.
\newblock In {\em MICRO}, 2016.

\bibitem{mckee2004reflections}
Sally~A McKee.
\newblock {Reflections on the Memory Wall}.
\newblock In {\em {CF}}, 2004.

\bibitem{mutlu.superfri15}
Onur Mutlu and Lavanya Subramanian.
\newblock {Research Problems and Opportunities in Memory Systems}.
\newblock {\em SUPERFRI}, 2014.

\bibitem{Park_MICRO2022}
Jisung Park, Roknoddin Azizi, Geraldo~F. Oliveira, Mohammad Sadrosadati, Rakesh
  Nadig, David Novo, Juan G{\'o}mez-Luna, Myungsuk Kim, and Onur Mutlu.
\newblock {Flash-Cosmos: In-Flash Bulk Bitwise Operations Using Inherent
  Computation Capability of NAND Flash Memory}.
\newblock In {\em MICRO}, 2022.

\bibitem{hajinazar2021simdram}
Nastaran Hajinazar, Geraldo~F Oliveira, Sven Gregorio, Jo{\~a}o~Dinis Ferreira,
  Nika~Mansouri Ghiasi, Minesh Patel, Mohammed Alser, Saugata Ghose, Juan
  G{\'o}mez-Luna, and Onur Mutlu.
\newblock {SIMDRAM: A Framework for Bit-serial SIMD Processing-using-DRAM}.
\newblock In {\em ASPLOS}, 2021.

\bibitem{gu2016biscuit}
Boncheol Gu, Andre~S. Yoon, Duck-Ho Bae, Insoon Jo, Jinyoung Lee, Jonghyun
  Yoon, Jeong-Uk Kang, Moonsang Kwon, Chanho Yoon, Sangyeun Cho, Jaeheon Jeong,
  and Duckhyun Chang.
\newblock {Biscuit: A Framework For Near-Data Processing of Big Data
  Workloads}.
\newblock In {\em ISCA}, 2016.

\bibitem{mutlu2022modern}
Onur Mutlu, Saugata Ghose, Juan G{\'o}mez-Luna, and Rachata Ausavarungnirun.
\newblock {A Modern Primer on Processing in Memory}.
\newblock In {\em Emerging Computing: From Devices to Systems -- Looking Beyond
  Moore and Von Neumann}. Springer, 2022.

\bibitem{boroumand2018google}
Amirali Boroumand, Saugata Ghose, Youngsok Kim, Rachata Ausavarungnirun, Eric
  Shiu, Rahul Thakur, Daehyun Kim, Aki Kuusela, Allan Knies, Parthasarathy
  Ranganathan, and Onur Mutlu.
\newblock {Google Workloads for Consumer Devices: Mitigating Data Movement
  Bottlenecks}.
\newblock In {\em ASPLOS}, 2018.

\bibitem{boroumand2021google}
Amirali Boroumand, Saugata Ghose, Berkin Akin, Ravi Narayanaswami, Geraldo~F.
  Oliveira, Xiaoyu Ma, Eric Shiu, and Onur Mutlu.
\newblock {Google Neural Network Models for Edge Devices: Analyzing and
  Mitigating Machine Learning Inference Bottlenecks}.
\newblock In {\em PACT}, 2021.

\bibitem{ahn2015scalable}
Junwhan Ahn, Sungpack Hong, Sungjoo Yoo, Onur Mutlu, and Kiyoung Choi.
\newblock {A Scalable Processing-in-memory Accelerator for Parallel Graph
  Processing}.
\newblock In {\em ISCA}, 2015.

\bibitem{ahn2015pim}
Junwhan Ahn, Sungjoo Yoo, Onur Mutlu, and Kiyoung Choi.
\newblock {PIM-Enabled Instructions: A Low-Overhead, Locality-Aware
  Processing-in-Memory Architecture}.
\newblock In {\em ISCA}, 2015.

\bibitem{mutlu2019processing}
Onur Mutlu, Saugata Ghose, Juan G{\'o}mez-Luna, and Rachata Ausavarungnirun.
\newblock {Processing Data Where It Makes Sense: Enabling In-Memory
  Computation}.
\newblock {\em MicPro}, 2019.

\bibitem{mutlu2019enabling}
Onur Mutlu, Saugata Ghose, Juan G{\'o}mez-Luna, and Rachata Ausavarungnirun.
\newblock {Enabling Practical Processing in and near Memory for Data-Intensive
  Computing}.
\newblock In {\em DAC}, 2019.

\bibitem{oliveira2021damov}
Geraldo~F. Oliveira, Juan G{\'o}mez-Luna, Lois Orosa, Saugata Ghose, Nandita
  Vijaykumar, Ivan Fernandez, Mohammad Sadrosadati, and Onur Mutlu.
\newblock {DAMOV: A New Methodology And Benchmark Suite For Evaluating Data
  Movement Bottlenecks}.
\newblock {\em IEEE Access}, 2021.

\bibitem{ghiasi2022alp}
Nika~Mansouri Ghiasi, Nandita Vijaykumar, Geraldo~F. Oliveira, Lois Orosa, Ivan
  Fernandez, Mohammad Sadrosadati, Konstantinos Kanellopoulos, Nastaran
  Hajinazar, Juan G{\'o}mez-Luna, and Onur Mutlu.
\newblock {ALP: Alleviating CPU-Memory Data Movement Overheads in
  Memory-Centric Systems}.
\newblock {\em TETC}, 2023.

\bibitem{seshadri-micro-2017}
Vivek Seshadri, Donghyuk Lee, Thomas Mullins, Hasan Hassan, Amirali Boroumand,
  Jeremie Kim, Michael~A Kozuch, Onur Mutlu, Phillip~B Gibbons, and Todd~C
  Mowry.
\newblock {Ambit: In-Memory Accelerator for Bulk Bitwise Operations Using
  Commodity {DRAM} Technology}.
\newblock In {\em MICRO}, 2017.

\bibitem{Gao_MICRO2021}
Congming Gao, Xin Xin, Youyou Lu, Youtao Zhang, Jun Yang, and Jiwu Shu.
\newblock {ParaBit: Processing Parallel Bitwise Operations in NAND Flash
  Memory-based SSDs}.
\newblock In {\em MICRO}, 2021.

\bibitem{barbalace_blockndp_2020}
Antonio Barbalace, Martin Decky, Javier Picorel, and Pramod Bhatotia.
\newblock {blockNDP}: {Block}-storage {Near} {Data} {Processing}.
\newblock In {\em Middleware}, 2020.

\bibitem{augusta2015jafar}
Aurelia Augusta and Stratos Idreos.
\newblock {JAFAR: Near-Data Processing for Databases}.
\newblock In {\em SIGMOD}, 2015.

\bibitem{boroumand2019conda}
Amirali Boroumand, Saugata Ghose, Minesh Patel, Hasan Hassan, Brandon Lucia,
  Rachata Ausavarungnirun, Kevin Hsieh, Nastaran Hajinazar, Krishna~T. Malladi,
  Hongzhong Zheng, and Onur Mutlu.
\newblock {CoNDA: Efficient Cache Coherence Support for Near-Data
  Accelerators}.
\newblock In {\em ISCA}, 2019.

\bibitem{fernandez2020natsa}
Ivan Fernandez, Ricardo Quislant, Christina Giannoula, Mohammed Alser, Juan
  Gomez-Luna, Eladio Gutierrez, Oscar Plata, and Onur Mutlu.
\newblock {NATSA: A Near-Data Processing Accelerator for Time Series Analysis}.
\newblock In {\em ICCD}, 2020.

\bibitem{singh2019napel}
Gagandeep Singh, Juan G{\'o}mez-Luna, Giovanni Mariani, Geraldo~F Oliveira,
  Stefano Corda, Sander Stuijk, Onur Mutlu, and Henk Corporaal.
\newblock {NAPEL: Near-memory Computing Application Performance Prediction via
  Ensemble Learning}.
\newblock In {\em DAC}, 2019.

\bibitem{gao2016hrl}
Mingyu Gao and Christos Kozyrakis.
\newblock {HRL: Efficient and Flexible Reconfigurable Logic for Near-Data
  Processing}.
\newblock In {\em HPCA}, 2016.

\bibitem{lee2020smartssd}
Joo~Hwan Lee, Hui Zhang, Veronica Lagrange, Praveen Krishnamoorthy, Xiaodong
  Zhao, and Yang~Seok Ki.
\newblock {SmartSSD: FPGA-Accelerated Near-Storage Data Analytics on SSD}.
\newblock {\em CAL}, 2020.

\bibitem{singh2021fpga}
Gagandeep Singh, Mohammed Alser, Damla~Senol Cali, Dionysios Diamantopoulos,
  Juan G{\'o}mez-Luna, Henk Corporaal, and Onur Mutlu.
\newblock {FPGA-based Near-Memory Acceleration of Modern Data-Intensive
  Applications}.
\newblock {\em IEEE Micro}, 2021.

\bibitem{medal2019}
Wenqin Huangfu, Xueqi Li, Shuangchen Li, Xing Hu, Peng Gu, and Yuan Xie.
\newblock {MEDAL: Scalable DIMM Based Near Data Processing Accelerator for DNA
  Seeding Algorithm}.
\newblock In {\em MICRO}, 2019.

\bibitem{liang-fpl-2019}
Shengwen Liang, Ying Wang, Cheng Liu, Huawei Li, and Xiaowei Li.
\newblock {InS-DLA: An In-SSD Deep Learning Accelerator for Near-Data
  Processing}.
\newblock In {\em FPL}, 2019.

\bibitem{Mansouri-Ghiasi_ASPLOS2022}
Nika~Mansouri Ghiasi, Jisung Park, Harun Mustafa, Jeremie~S. Kim, Ataberk
  Olgun, Arvid Gollwitzer, Damla~Senol Cali, Can Firtina, Haiyu Mao,
  Nour~Almadhoun Alserr, Rachata Ausavarungnirun, Nandita Vijaykumar, Mohammed
  Alser, and Onur Mutlu.
\newblock {GenStore: A High-Performance In-Storage Processing System for Genome
  Sequence Analysis}.
\newblock In {\em ASPLOS}, 2022.

\bibitem{ghiasi2026sage}
Nika~Mansouri Ghiasi, Talu G{\"u}loglu, Harun Mustafa, Can Firtina, Konstantina
  Koliogeorgi, Konstantinos Kanellopoulos, Haiyu Mao, Rakesh Nadig, Mohammad
  Sadrosadati, Jisung Park, and Onur Mutlu.
\newblock {SAGe: A Lightweight Algorithm-Architecture Co-Design for Mitigating
  the Data Preparation Bottleneck in Large-Scale Genome Sequence Analysis}.
\newblock In {\em HPCA}, 2026.

\bibitem{ghiasi2026grains}
Nika~Mansouri Ghiasi, Harun Mustafa, Talu G{\"u}loglu, Rakesh Nadig,
  Konstantina Koliogeorgi, Susana~Rebolledo Ruiz, Marc Rautmann, Furkan Eris,
  Mohammad Sadrosadati, Jisung Park, and Onur Mutlu.
\newblock {GRAINS: Enabling High-Performance and Low-Cost Graph-Based Genome
  Analysis via Storage-Aware Algorithm-Architecture Co-Design}.
\newblock In {\em ISCA}, 2026.

\bibitem{Ghiasi2024MegIS}
Nika~Mansouri Ghiasi, Mohammad Sadrosadati, Harun Mustafa, Arvid Gollwitzer,
  Can Firtina, Julien Eudine, Haiyu Mao, Jo\"el Lindegger, Meryem~Banu Cavlak,
  Mohammed Alser, Jisung Park, and Onur Mutlu.
\newblock {MegIS: High-Performance, Energy-Efficient, and Low-Cost Metagenomic
  Analysis with In-Storage Processing}.
\newblock In {\em ISCA}, 2024.

\bibitem{oliveira2024mimdram}
Geraldo~F Oliveira, Ataberk Olgun, Abdullah~Giray Ya{\u{g}}l{\i}k{\c{c}}{\i},
  F~Nisa Bostanc{\i}, Juan G{\'o}mez-Luna, Saugata Ghose, and Onur Mutlu.
\newblock {MIMDRAM: An End-to-End Processing-Using-DRAM System for
  High-Throughput, Energy-Efficient and Programmer-Transparent
  Multiple-Instruction Multiple-Data Computing}.
\newblock In {\em HPCA}, 2024.

\bibitem{nider2020processing}
Joel Nider, Craig Mustard, Andrada Zoltan, and Alexandra Fedorova.
\newblock {Processing in Storage Class Memory}.
\newblock In {\em HotStorage}, 2020.

\bibitem{hsieh.isca16}
Kevin Hsieh, Eiman Ebrahimi, Gwangsun Kim, Niladrish Chatterjee, Mike O'Conner,
  Nandita Vijaykumar, Onur Mutlu, and Stephen Keckler.
\newblock {Transparent Offloading and Mapping (TOM): Enabling
  Programmer-Transparent Near-Data Processing in GPU Systems}.
\newblock In {\em ISCA}, 2016.

\bibitem{park2024attacc}
Jaehyun Park, Jaewan Choi, Kwanhee Kyung, Michael~Jaemin Kim, Yongsuk Kwon,
  Nam~Sung Kim, and Jung~Ho Ahn.
\newblock {AttAcc! Unleashing the Power of PIM for Batched Transformer-based
  Generative Model Inference}.
\newblock In {\em ASPLOS}, 2024.

\bibitem{chen2022offload}
Dan Chen, Hai Jin, Long Zheng, Yu~Huang, Pengcheng Yao, Chuangyi Gui, Qinggang
  Wang, Haifeng Liu, Haiheng He, Xiaofei Liao, and Ran Zheng.
\newblock {A General Offloading Approach for Near-DRAM Processing-In-Memory
  Architectures}.
\newblock In {\em IPDPS}, 2022.

\bibitem{li2018cisc}
Dongyang Li, Yafei Yang, Weijun Li, and Qing Yang.
\newblock {CISC: Coordinating Intelligent SSD and CPU to Speedup Graph
  Processing}.
\newblock In {\em ISPDC}, 2018.

\bibitem{li2023optimizing}
Lin Li, Xianzhang Chen, Jiali Li, Jiapin Wang, Duo Liu, Yujuan Tan, and Ao~Ren.
\newblock {Optimizing the Performance of NDP Operations by Retrieving File
  Semantics in Storage}.
\newblock In {\em DAC}, 2023.

\bibitem{maity2025unguided}
Satanu Maity, Manojit Ghose, Avinash Kumar, Anol Chakraborty, and Ankit
  Chakraborty.
\newblock {Unguided Machine Learning-Based Computation Offloading for
  Near-Memory Processing}.
\newblock In {\em VLSID}, 2025.

\bibitem{olivier2019hexo}
Pierre Olivier, AKM~Fazla Mehrab, Stefan Lankes, Mohamed~Lamine Karaoui, Robert
  Lyerly, and Binoy Ravindran.
\newblock {HEXO: Offloading HPC Compute-Intensive Workloads on Low-Cost,
  Low-Power Embedded Systems}.
\newblock In {\em HPDC}, 2019.

\bibitem{wei2022pimprof}
Yizhou Wei, Minxuan Zhou, Sihang Liu, Korakit Seemakhupt, Tajana Rosing, and
  Samira Khan.
\newblock {PIMProf: An Automated Program Profiler for Processing-in-Memory
  Offloading Decisions}.
\newblock In {\em DATE}, 2022.

\bibitem{weiner2022tmo}
Johannes Weiner, Niket Agarwal, Dan Schatzberg, Leon Yang, Hao Wang, Blaise
  Sanouillet, Bikash Sharma, Tejun Heo, Mayank Jain, Chunqiang Tang, and
  Dimitrios Skarlatos.
\newblock {TMO: Transparent Memory Offloading in Datacenters}.
\newblock In {\em ASPLOS}, 2022.

\bibitem{yang2023lambda}
Zhe Yang, Youyou Lu, Xiaojian Liao, Youmin Chen, Junru Li, Siyu He, and Jiwu
  Shu.
\newblock {{$\lambda$-IO}: A Unified {IO} Stack for Computational Storage}.
\newblock In {\em FAST}, 2023.

\bibitem{lincoln-hpca}
Weiyi Sun, Mingyu Gao, Zhaoshi Li, Aoyang Zhang, Iris~Ying Chou, Jianfeng Zhu,
  Shaojun Wei, and Leibo Liu.
\newblock {Lincoln: Real-Time 50-100B LLM Inference on Consumer Devices with
  LPDDR-Interfaced, Compute-Enabled Flash Memory}.
\newblock In {\em HPCA}, 2025.

\bibitem{jang2025inf}
Hongsun Jang, Siung Noh, Changmin Shin, Jaewon Jung, Jaeyong Song, and Jinho
  Lee.
\newblock {INF\^{} 2: High-Throughput Generative Inference of Large Language
  Models using Near-Storage Processing}.
\newblock {\em arXiv preprint arXiv:2502.09921}, 2025.

\bibitem{pan2024instattention}
Xiurui Pan, Endian Li, Qiao Li, Shengwen Liang, Yizhou Shan, Ke~Zhou, Yingwei
  Luo, Xiaolin Wang, and Jie Zhang.
\newblock {InstAttention: In-Storage Attention Offloading for Cost-Effective
  Long-Context LLM Inference}.
\newblock In {\em HPCA}, 2025.

\bibitem{jaliminchecs}
Lokesh Jaliminche, Yangwook Kang, Changho Choi, Pankaj Mehra, and Heiner Litz.
\newblock {CS-Assist: A Tool to Assist Computational Storage Device Offload}.
\newblock {\em NVMW}, 2024.

\bibitem{kang2024isp}
Seokwon Kang, Jongbin Kim, Gyeongyong Lee, Jeongmyung Lee, Jiwon Seo, Hyungsoo
  Jung, Yong~Ho Song, and Yongjun Park.
\newblock {ISP Agent: A Generalized In-storage-processing Workload Offloading
  Framework by Providing Multiple Optimization Opportunities}.
\newblock {\em ACM Transactions on Architecture and Code Optimization}, 2024.

\bibitem{seshadri-osdi-2014}
Sudharsan Seshadri, Mark Gahagan, Sundaram Bhaskaran, Trevor Bunker, Arup De,
  Yanqin Jin, Yang Liu, and Steven Swanson.
\newblock {Willow: A User-Programmable {SSD}}.
\newblock In {\em USENIX OSDI}, 2014.

\bibitem{acharya-asplos-1998}
Anurag Acharya, Mustafa Uysal, and Joel Saltz.
\newblock {Active Disks: Programming Model, Algorithms and Evaluation}.
\newblock In {\em ASPLOS}, 1998.

\bibitem{Kabra2025CipherMatch}
Mayank Kabra, Rakesh Nadig, Harshita Gupta, Rahul Bera, Manos Frouzakis,
  Vamanan Arulchelvan, Yu~Liang, Haiyu Mao, Mohammad Sadrosadati, and Onur
  Mutlu.
\newblock {CIPHERMATCH: Accelerating Homomorphic Encryption-Based String
  Matching via Memory-Efficient Data Packing and In-Flash Processing}.
\newblock In {\em ASPLOS}, 2025.

\bibitem{wang-eurosys-2019}
Xiaohao Wang, Yifan Yuan, You Zhou, Chance~C Coats, and Jian Huang.
\newblock {Project Almanac: A Time-Traveling Solid-State Drive}.
\newblock In {\em EuroSys}, 2019.

\bibitem{kim-fast-2021}
Shine Kim, Yunho Jin, Gina Sohn, Jonghyun Bae, Tae~Jun Ham, and Jae~W. Lee.
\newblock {Behemoth: A Flash-centric Training Accelerator for Extreme-scale
  {DNNs}}.
\newblock In {\em FAST}, 2021.

\bibitem{kang-msst-2013}
Yangwook Kang, {Yang-Suk} Kee, Ethan~L. Miller, and Chanik Park.
\newblock {Enabling Cost-effective Data Processing with Smart {SSD}}.
\newblock In {\em MSST}, 2013.

\bibitem{torabzadehkashi-pdp-2019}
Mahdi Torabzadehkashi, Siavash Rezaei, Ali Heydarigorji, Hosein Bobarshad,
  Vladimir Alves, and Nader Bagherzadeh.
\newblock {Catalina: In-Storage Processing Acceleration for Scalable Big Data
  Analytics}.
\newblock In {\em PDP}, 2019.

\bibitem{keeton-sigmod-1998}
Kimberly Keeton, David~A. Patterson, and Joseph~M. Hellerstein.
\newblock {A Case for Intelligent Disks (IDISKs)}.
\newblock {\em {SIGMOD} Rec.}, 1998.

\bibitem{koo-micro-2017}
Gunjae Koo, Kiran~Kumar Matam, Te~I, H.~V. Krishna~Giri Narra, Jing Li,
  Hung-Wei Tseng, Steven Swanson, and Murali Annavaram.
\newblock {Summarizer: Trading Communication with Computing Near Storage}.
\newblock In {\em MICRO}, 2017.

\bibitem{tiwari-fast-2013}
Devesh Tiwari, Simona Boboila, Sudharshan Vazhkudai, Youngjae Kim, Xiaosong Ma,
  Peter Desnoyers, and Yan Solihin.
\newblock {Active Flash: Towards Energy-Efficient, In-Situ Data Analytics on
  Extreme-Scale Machines}.
\newblock In {\em FAST}, 2013.

\bibitem{tiwari-hotpower-2012}
Devesh Tiwari, Sudharshan~S Vazhkudai, Youngjae Kim, Xiaosong Ma, Simona
  Boboila, and Peter~J Desnoyers.
\newblock {Reducing Data Movement Costs Using Energy-Efficient, Active
  Computation on SSD}.
\newblock In {\em HotPower}, 2012.

\bibitem{boboila-msst-2012}
Simona Boboila, Youngjae Kim, Sudharshan~S. Vazhkudai, Peter Desnoyers, and
  Galen~M. Shipman.
\newblock {Active Flash: Out-of-core Data Analytics on Flash Storage}.
\newblock In {\em MSST}, 2012.

\bibitem{bae-cikm-2013}
Duck-Ho Bae, Jin-Hyung Kim, Sang-Wook Kim, Hyunok Oh, and Chanik Park.
\newblock {Intelligent SSD: A Turbo for Big Data Mining}.
\newblock In {\em CIKM}, 2013.

\bibitem{torabzadehkashi-ipdpsw-2018}
Mahdi Torabzadehkashi, Siavash Rezaei, Vladimir Alves, and Nader Bagherzadeh.
\newblock {CompStor: An In-Storage Computation Platform for Scalable
  Distributed Processing}.
\newblock In {\em IPDPSW}, 2018.

\bibitem{pei-tos-2019}
Shuyi Pei, Jing Yang, and Qing Yang.
\newblock {REGISTOR: A Platform for Unstructured Data Processing inside SSD
  Storage}.
\newblock {\em ACM TOS}, 2019.

\bibitem{do-sigmod-2013}
Jaeyoung Do, Yang-Suk Kee, Jignesh~M. Patel, Chanik Park, Kwanghyun Park, and
  David~J. DeWitt.
\newblock {Query Processing on Smart SSDs: Opportunities and Challenges}.
\newblock In {\em ACM SIGMOD}, 2013.

\bibitem{kim-infosci-2016}
Sungchan Kim, Hyunok Oh, Chanik Park, Sangyeun Cho, Sang-Won Lee, and Bongki
  Moon.
\newblock {In-Storage Processing of Database Scans and Joins}.
\newblock {\em Information Sciences}, 2016.

\bibitem{riedel-computer-2001}
Erik Riedel, Christos Faloutsos, Garth~A. Gibson, and David Nagle.
\newblock {Active Disks for Large-Scale Data Processing}.
\newblock {\em Computer}, 2001.

\bibitem{riedel-vldb-1998}
Erik Riedel, Garth Gibson, and Christos Faloutsos.
\newblock {Active Storage for Large-Scale Data Mining and Multimedia}.
\newblock In {\em VLDB}, 1998.

\bibitem{liang-atc-2019}
Shengwen Liang, Ying Wang, Youyou Lu, Zhe Yang, Huawei Li, and Xiaowei Li.
\newblock {Cognitive SSD: A Deep Learning Engine for In-Storage Data
  Retrieval}.
\newblock In {\em USENIX ATC}, 2019.

\bibitem{cho-wondp-2013}
Benjamin~Y Cho, Won~Seob Jeong, Doohwan Oh, and Won~Woo Ro.
\newblock {XSD: Accelerating MapReduce by Harnessing the GPU inside an SSD }.
\newblock In {\em WoNDP}, 2013.

\bibitem{jun2015bluedbm}
Sang-Woo Jun, Ming Liu, Sungjin Lee, Jamey Hicks, John Ankcorn, Myron King,
  Shuotao Xu, and Arvind.
\newblock {BlueDBM: An Appliance for Big Data Analytics}.
\newblock In {\em ISCA}, 2015.

\bibitem{ajdari-hpca-2019}
Mohammadamin Ajdari, Pyeongsu Park, Joonsung Kim, Dongup Kwon, and Jangwoo Kim.
\newblock {CIDR: A Cost-effective In-line Data Reduction System for
  Terabit-per-second Scale SSD Arrays}.
\newblock In {\em HPCA}, 2019.

\bibitem{jun-hpec-2016}
Sang-Woo Jun, Huy~T Nguyen, Vijay Gadepally, et~al.
\newblock {In-Storage Embedded Accelerator for Sparse Pattern Processing}.
\newblock In {\em HPEC}, 2016.

\bibitem{kang-tc-2021}
Myeonggu Kang, Hyeonuk Kim, Hyein Shin, Jaehyeong Sim, Kyeonghan Kim, and
  Lee-Sup Kim.
\newblock {S-FLASH: A NAND Flash-based Deep Neural Network Accelerator
  Exploiting Bit-Level Sparsity}.
\newblock {\em IEEE TC}, 2022.

\bibitem{kim-sigops-2020}
Minsub Kim and Sungjin Lee.
\newblock {Reducing Tail Latency of DNN-based Recommender Systems using
  In-Storage Processing}.
\newblock In {\em APSys}, 2020.

\bibitem{Lee_ISCA2022}
Yunjae Lee, Jinha Chung, and Minsoo Rhu.
\newblock {SmartSAGE: Training Large-Scale Graph Neural Networks using
  In-Storage Processing Architectures}.
\newblock In {\em ISCA}, 2022.

\bibitem{li2023ecssd}
Siqi Li, Fengbin Tu, Liu Liu, Jilan Lin, Zheng Wang, Yangwook Kang, Yufei Ding,
  and Yuan Xie.
\newblock {ECSSD: Hardware/Data Layout Co-Designed In-Storage-Computing
  Architecture for Extreme Classification}.
\newblock In {\em ISCA}, 2023.

\bibitem{ruan2019insider}
Zhenyuan Ruan, Tong He, and Jason Cong.
\newblock {INSIDER: Designing In-Storage Computing System for Emerging
  High-Performance Drive}.
\newblock In {\em USENIX ATC}, 2019.

\bibitem{wang2016ssd1}
Jianguo Wang, Dongchul Park, Yannis Papakonstantinou, and Steven Swanson.
\newblock {SSD In-storage Computing for Search Engines}.
\newblock {\em ToC}, 2016.

\bibitem{jeong-tpds-2019}
Won~Seob Jeong, Changmin Lee, Keunsoo Kim, Myung~Kuk Yoon, Won Jeon, Myoungsoo
  Jung, and Won~Woo Ro.
\newblock {REACT: Scalable and High-performance Regular Expression Pattern
  Matching Accelerator for In-storage Processing}.
\newblock {\em IEEE TPDS}, 2019.

\bibitem{mao2012cache}
Yandong Mao, Eddie Kohler, and Robert~Tappan Morris.
\newblock {Cache Craftiness for Fast Multicore Key-Value Storage}.
\newblock In {\em {EuroSys}}, 2012.

\bibitem{gouk2024dockerssd}
Donghyun Gouk, Miryeong Kwon, Hanyeoreum Bae, and Myoungsoo Jung.
\newblock {DockerSSD: Containerized In-Storage Processing and Hardware
  Acceleration for Computational SSDs}.
\newblock In {\em HPCA}, 2024.

\bibitem{yavits2021giraf}
Leonid Yavits, Roman Kaplan, and Ran Ginosar.
\newblock {GIRAF: General Purpose In-Storage Resistive Associative Framework}.
\newblock {\em IEEE TPDS}, 2021.

\bibitem{Kim_HPCA2023}
Junkyum Kim, Myeonggu Kang, Yunki Han, {Yang-gon} Kim, and Lee-Sup Kim.
\newblock {OptimStore: In-Storage Optimization of Large Scale DNNs with On-Die
  Processing}.
\newblock In {\em HPCA}, 2023.

\bibitem{lim-icce-2021}
Minje Lim, Jeeyoon Jung, and Dongkun Shin.
\newblock {LSM-Tree Compaction Acceleration Using In-Storage Processing}.
\newblock In {\em ICCE-Asia}, 2021.

\bibitem{narasimhamurthy2019sage}
Sai Narasimhamurthy, Nikita Danilov, Sining Wu, Ganesan Umanesan, Stefano
  Markidis, Sergio Rivas-Gomez, Ivy~Bo Peng, Erwin Laure, Dirk Pleiter, and
  Shaun De~Witt.
\newblock {SAGE: Percipient Storage for Exascale Data Centric Computing}.
\newblock {\em Parallel Computing}, 2019.

\bibitem{jun-isca-2018}
Sang-Woo Jun, Andy Wright, Sizhuo Zhang, Shuotao Xu, and Arvind.
\newblock {{GraFBoost:} Using Accelerated Flash Storage for External Graph
  Analytics}.
\newblock In {\em ISCA}, 2018.

\bibitem{fakhry2023review}
Dina Fakhry, Mohamed Abdelsalam, M~Watheq El-Kharashi, and Mona Safar.
\newblock {A Review on Computational Storage Devices and Near Memory Computing
  for High-Performance Applications}.
\newblock {\em J. Memori}, 2023.

\bibitem{jo2016yoursql}
Insoon Jo, Duck-Ho Bae, Andre~S. Yoon, Jeong-Uk Kang, Sangyeun Cho, Daniel
  D.~G. Lee, and Jaeheon Jeong.
\newblock {YourSQL: A High-Performance Database System Leveraging In-storage
  Computing}.
\newblock {\em PVLDB}, 2016.

\bibitem{Chen2025REIS}
Kangqi Chen, Rakesh Nadig, Manos Frouzakis, Nika~Mansouri Ghiasi, Yu~Liang,
  Haiyu Mao, Jisung Park, Mohammad Sadrosadati, and Onur Mutlu.
\newblock {REIS: A High-Performance and Energy-Efficient Retrieval System with
  In-Storage Processing}.
\newblock In {\em ISCA}, 2025.

\bibitem{Li_ATC2021}
Cangyuan Li, Ying Wang, Cheng Liu, Shengwen Liang, Huawei Li, and Xiaowei Li.
\newblock {GLIST: Towards In-Storage Graph Learning}.
\newblock In {\em USENIX ATC}, 2021.

\bibitem{wang2016ssd}
Jianguo Wang, Dongchul Park, Yang-Suk Kee, Yannis Papakonstantinou, and Steven
  Swanson.
\newblock {SSD In-Storage Computing for List Intersection}.
\newblock In {\em DaMoN}, 2016.

\bibitem{mahapatra2025rag}
Rohan Mahapatra, Harsha Santhanam, Christopher Priebe, Hanyang Xu, and Hadi
  Esmaeilzadeh.
\newblock In-storage acceleration of retrieval augmented generation as a
  service.
\newblock In {\em ISCA}, 2025.

\bibitem{wang2024beacongnn}
Yuyue Wang, Xiurui Pan, Yuda An, Jie Zhang, and Glenn Reinman.
\newblock {BeaconGNN: Large-Scale GNN Acceleration with Out-of-Order Streaming
  In-Storage Computing}.
\newblock In {\em HPCA}, 2024.

\bibitem{Yu2024CambriconLLM}
Zhongkai Yu, Shengwen Liang, Tianyun Ma, Yunke Cai, Ziyuan Nan, Di~Huang,
  Xinkai Song, Yifan Hao, Jie Zhang, Tian Zhi, Yongwei Zhao, Zidong Du, Xing
  Hu, Qi~Guo, and Tianshi Chen.
\newblock {Cambricon-LLM: A Chiplet-Based Hybrid Architecture for On-Device
  Inference of 70B LLM}.
\newblock In {\em MICRO}, 2024.

\bibitem{Chen_MICRO2024}
Jian Chen, Congming Gao, Youyou Lu, Yuhao Zhang, and Jiwu Shu.
\newblock {Ares-Flash: Efficient Parallel Integer Arithmetic Operations Using
  NAND Flash Memory}.
\newblock In {\em MICRO}, 2024.

\bibitem{kim2025crossbit}
Hyunjin Kim, Seunghwan Song, Sukhyun Choi, Jeongin Choe, Sanghyeok Han, Jisung
  Park, Jinho Lee, and Jae-Joon Kim.
\newblock {CrossBit: Bitwise Computing in NAND Flash Memory with Inter-Bitline
  Data Communication}.
\newblock In {\em MICRO}, 2025.

\bibitem{wong2024tcam}
Ryan Wong, Nikita Kim, Kevin Higgs, Sapan Agarwal, Engin Ipek, Saugata Ghose,
  and Ben Feinberg.
\newblock {TCAM-SSD: A Framework for Search-Based Computing in Solid-State
  Drives}.
\newblock {\em arXiv preprint arXiv:2403.06938}, 2024.

\bibitem{wong2025anvil}
Ryan Wong, Nikita Kim, Aniket Das, Kevin Higgs, Engin Ipek, Sapan Agarwal,
  Saugata Ghose, and Ben Feinberg.
\newblock {Anvil: An In-Storage Accelerator for Name-Value Data Stores}.
\newblock In {\em ISCA}, 2025.

\bibitem{chen2024search}
Yun-Chih Chen, Yuan-Hao Chang, and Tei-Wei Kuo.
\newblock {Search-in-Memory (SiM): Conducting Data-Bound Computations on Flash
  Chip for Enhanced Efficiency}.
\newblock In {\em DATE}, 2024.

\bibitem{choi2020flash}
Won~Ho Choi, Pi-Feng Chiu, Wen Ma, Gertjan Hemink, Tung~Thanh Hoang, Martin
  Lueker-Boden, and Zvonimir Bandic.
\newblock {An In-flash Binary Neural Network Accelerator with SLC NAND Flash
  Array}.
\newblock In {\em ISCAS}, 2020.

\bibitem{chun2022pif}
Myoungjun Chun, Jaeyong Lee, Sanggu Lee, Myungsuk Kim, and Jihong Kim.
\newblock {PiF: In-Flash Acceleration for Data-Intensive Applications}.
\newblock In {\em HotStorage}, 2022.

\bibitem{lee2025aif}
Jaeyong Lee, Hyeunjoo Kim, Sanghun Oh, Myoungjun Chun, Myungsuk Kim, and Jihong
  Kim.
\newblock Aif: Accelerating on-device llm inference using in-flash processing.
\newblock In {\em ISCA}, 2025.

\bibitem{wang2018three}
Panni Wang, Feng Xu, Bo~Wang, Bin Gao, Huaqiang Wu, He~Qian, and Shimeng Yu.
\newblock {Three-Dimensional NAND Flash for Vector–Matrix Multiplication}.
\newblock {\em IEEE TVLSI}, 2018.

\bibitem{ghiasi2026enabling}
Nika~Mansouri Ghiasi and Onur Mutlu.
\newblock Enabling fast, efficient, and low-cost genomic and metagenomic
  analyses via storage-centric system designs.
\newblock In {\em ICS Workshops}, 2026.

\bibitem{samsung-980pro}
Samsung.
\newblock {Samsung SSD 980 PRO}.
\newblock
  \url{https://www.samsung.com/semiconductor/minisite/ssd/product/consumer/980pro/}.

\bibitem{adatasu630}
ADATA.
\newblock {ADATA Ultimate Series: SU630}.
\newblock \url{https://www.adata.com/en/consumer/category/ssds/591/}.

\bibitem{intelqlc}
Intel.
\newblock {Intel SSD 660p Series},
  \url{https://www.intel.com/content/www/us/en/products/docs/memory-storage/solid-state-drives/consumer-ssds/660p-series-brief.html}.

\bibitem{intels4510}
{Intel}.
\newblock {Intel SSD D3-S4510 Series},
  \url{https://www.intel.com/content/www/us/en/products/memory-storage/solid-state-drives/data-center-ssds/d3-series/d3-s4510-series/d3-s4510-1-92tb-2-5inch-3d2.html}.

\bibitem{intelp4610}
Intel.
\newblock {Intel SSD DC P4610 Series}.
\newblock
  \url{https://ark.intel.com/content/www/us/en/ark/products/140103/intel-ssd-dc-p4610-series-1-6tb-2-5in-pcie-3-1-x4-3d2-tlc.html}.

\bibitem{nadig2023venice}
Rakesh Nadig, Mohammad Sadrosadati, Haiyu Mao, Nika~Mansouri Ghiasi, Arash
  Tavakkol, Jisung Park, Hamid Sarbazi-Azad, Juan G{\'o}mez-Luna, and Onur
  Mutlu.
\newblock {Venice: Improving Solid-State Drive Parallelism at Low Cost via
  Conflict-Free Accesses}.
\newblock In {\em ISCA}, 2023.

\bibitem{micheloni2010inside}
Rino Micheloni, Luca Crippa, and Alessia Marelli.
\newblock {\em {Inside NAND Flash Memories}}.
\newblock Springer, 2010.

\bibitem{micheloni2013inside}
Rino Micheloni, Alessia Marelli, and Kam Eshghi.
\newblock {\em Inside Solid State Drives (SSDs)}.
\newblock Springer, 2013.

\bibitem{inteloptane}
Intel.
\newblock {Intel Optane SSD DC P4801X Series},
  \url{https://ark.intel.com/content/www/us/en/ark/products/149365/intel-optane-ssd-dc-p4801x-series-100gb-2-5in-pcie-x4-3d-xpoint.html}.

\bibitem{samsungmlc}
Samsung.
\newblock {Samsung SSD 960 PRO NVMe M.2 512GB},
  \url{https://www.samsung.com/us/computing/memory-storage/solid-state-drives/ssd-960-pro-m-2-512gb-mz-v6p512bw}.

\bibitem{samsung2017znand}
Samsung.
\newblock {Ultra-Low Latency with Samsung Z-NAND SSD},
  \url{https://www.samsung.com/semiconductor/global.semi.static/Ultra-Low_Latency_with_Samsung_Z-NAND_SSD-0.pdf}.

\bibitem{Soysal2025MARS}
Melina Soysal, Konstantina Koliogeorgi, Can Firtina, Nika~Mansouri Ghiasi,
  Rakesh Nadig, Haiyu Mao, Geraldo~F. Oliveira, Yu~Liang, Klea Zambaku,
  Mohammad Sadrosadati, and Onur Mutlu.
\newblock {MARS: Processing-In-Memory Acceleration of Raw Signal Genome
  Analysis Inside the Storage Subsystem}.
\newblock In {\em ICS}, 2025.

\bibitem{Matam_ISCA2019}
Kiran~K. Matam, Gunjae Koo, Haipeng Zha, Hung-Wei Tseng, and Murali Annavaram.
\newblock {GraphSSD: Graph Semantics Aware SSD}.
\newblock In {\em ISCA}, 2019.

\bibitem{Liang_DAC2022}
Shengwen Liang, Ying Wang, Ziming Yuan, Cheng Liu, Huawei Li, and Xiaowei Li.
\newblock {VStore: In-Storage Graph-Based Vector Search Accelerator}.
\newblock In {\em DAC}, 2022.

\bibitem{Duffy_PVLDB2023}
Carl Duffy, Jaehoon Shim, Sang-Hoon Kim, and Jin-Soo Kim.
\newblock {Dotori: A Key-Value SSD Based KV Store}.
\newblock {\em Proc. VLDB Endowment}, 2023.

\bibitem{Park_ASPLOS2025}
Chanyoung Park, Jungho Lee, Chun-Yi Liu, Kyungtae Kang, Mahmut~T. Kandemir, and
  Wonil Choi.
\newblock {AnyKey: A Key-Value SSD for All Workload Types}.
\newblock In {\em ASPLOS}, 2025.

\bibitem{Bisson_IPCCC2018}
Tim Bisson, Ke~Chen, Changho Choi, Vijay Balakrishnan, and {Yang-Suk} Kee.
\newblock {Crail-KV: A High-Performance Distributed Key-Value Store Leveraging
  Native KV-SSDs over NVMe-oF}.
\newblock In {\em IPCCC}, 2018.

\bibitem{nadig2026conduit}
Rakesh Nadig, Vamanan Arulchelvan, Mayank Kabra, Harshita Gupta, Rahul Bera,
  Nika~Mansouri Ghiasi, Nanditha Rao, Qingcai Jiang, Andreas~Kosmas Kakolyris,
  Yu~Liang, Mohammad Sadrosadati, and Onur Mutlu.
\newblock {Conduit: Programmer-Transparent Near-Data Processing Using Multiple
  Compute-Capable Resources in Solid State Drives}.
\newblock In {\em HPCA}, 2026.

\bibitem{mutlu2024memory}
Onur Mutlu, Ataberk Olgun, Geraldo~F Oliveira, and Ismail~E Yuksel.
\newblock {Memory-Centric Computing: Recent Advances in Processing-in-DRAM}.
\newblock In {\em IEDM}, 2024.

\bibitem{mutlu2025memory}
Onur Mutlu, Ataberk Olgun, and {\.I}smail~Emir Y{\"u}ksel.
\newblock {Memory-Centric Computing: Solving Computing’s Memory Problem}.
\newblock In {\em IMW}, 2025.

\bibitem{tavakkol2018flin}
Arash Tavakkol, Mohammad Sadrosadati, Saugata Ghose, Jeremie Kim, Yixin Luo,
  Yaohua Wang, Nika~Mansouri Ghiasi, Lois Orosa, Juan G{\'o}mez-Luna, and Onur
  Mutlu.
\newblock {FLIN: Enabling Fairness and Enhancing Performance in Modern NVMe
  Solid State Drives}.
\newblock In {\em ISCA}, 2018.

\bibitem{huang2017flashblox}
Jian Huang, Anirudh Badam, Laura Caulfield, Suman Nath, Sudipta Sengupta,
  Bikash Sharma, and Moinuddin~K. Qureshi.
\newblock {{FlashBlox}: Achieving Both Performance Isolation and Uniform
  Lifetime for Virtualized {SSDs}}.
\newblock In {\em FAST}, 2017.

\bibitem{kwon2020dc}
Miryeong Kwon, Donghyun Gouk, Changrim Lee, Byounggeun Kim, Jooyoung Hwang, and
  Myoungsoo Jung.
\newblock {DC-Store: Eliminating Noisy Neighbor Containers Using Deterministic
  I/O Performance and Resource Isolation}.
\newblock In {\em FAST}, 2020.

\bibitem{min2021isolating}
Donghyun Min and Youngjae Kim.
\newblock {Isolating Namespace and Performance in Key-Value SSDs for
  Multi-Tenant Environments}.
\newblock In {\em HotStorage}, 2021.

\bibitem{min2023multi}
Donghyun Min, Kihyun Kim, Chaewon Moon, Awais Khan, Seungjin Lee, Changhwan
  Yun, Woosuk Chung, and Youngjae Kim.
\newblock {A Multi-Tenant Key-Value SSD with Secondary Index for Search Query
  Processing and Analysis}.
\newblock {\em TECS}, 2023.

\bibitem{kim2015ops}
Jaeho Kim, Donghee Lee, and Sam~H. Noh.
\newblock {Towards SLO Complying SSDs Through OPS Isolation}.
\newblock In {\em FAST}, 2015.

\bibitem{jaliminche2023enabling}
Lokesh~N. Jaliminche, Chandranil~Nil Chakraborttii, Changho Choi, and Heiner
  Litz.
\newblock {Enabling Multi-Tenancy on SSDs} with accurate {IO} interference
  modeling.
\newblock In {\em SoCC}, 2023.

\bibitem{gonzalez2017multi}
Javier Gonz{\'a}lez and Matias Bj{\o}rling.
\newblock {Multi-Tenant I/O Isolation with Open-Channel SSDs}.
\newblock In {\em NVMW}, 2017.

\bibitem{sun2025fleetio}
Jinghan Sun, Benjamin Reidys, Daixuan Li, Jichuan Chang, Marc Snir, and Jian
  Huang.
\newblock {FleetIO: Managing Multi-Tenant Cloud Storage with Multi-Agent
  Reinforcement Learning}.
\newblock In {\em ASPLOS}, 2025.

\bibitem{nvmenamespaces}
{NVM Express}.
\newblock {NVMe Namespaces}, 2022.

\bibitem{nvme_base_spec}
{NVM Express}.
\newblock {NVM Express® Base Specification (Revision 2.3)}.
\newblock NVM Express Specification, 2025.
\newblock Includes updates for Endurance Group Management, Zoned Namespaces,
  and more.

\bibitem{gupta2009dftl}
Aayush Gupta, Youngjae Kim, and Bhuvan Urgaonkar.
\newblock {DFTL}: A flash translation layer employing demand-based selective
  caching of page-level address mappings.
\newblock In {\em ASPLOS}, 2009.

\bibitem{agrawal2008design}
Nitin Agrawal, Vijayan Prabhakaran, Ted Wobber, John~D. Davis, Mark~S. Manasse,
  and Rina Panigrahy.
\newblock {Design Tradeoffs for SSD Performance}.
\newblock In {\em USENIX ATC}, 2008.

\bibitem{xie2014adaptive}
Wei Xie and Yong Chen.
\newblock {An Adaptive Separation-Aware FTL for Improving the Efficiency of
  Garbage Collection in SSDs}.
\newblock In {\em CCGRID}, 2014.

\bibitem{jung2012taking}
Myoungsoo Jung, Ramya Prabhakar, and Mahmut~Taylan Kandemir.
\newblock {Taking Garbage Collection Overheads Off the Critical Path in SSDs}.
\newblock In {\em Middleware}, 2012.

\bibitem{shahidi2016exploring}
Narges Shahidi, Mahmut~T. Kandemir, Mohammad Arjomand, Chita~R. Das, Myoungsoo
  Jung, and Anand Sivasubramaniam.
\newblock {Exploring the Potentials of Parallel Garbage Collection in SSDs for
  Enterprise Storage Systems}.
\newblock In {\em SC}, 2016.

\bibitem{wu2016gcar}
Suzhen Wu, Yanping Lin, Bo~Mao, and Hong Jiang.
\newblock {GCaR: Garbage Collection aware Cache Management with Improved
  Performance for Flash-based SSDs}.
\newblock In {\em ICS}, 2016.

\bibitem{choi2018parallelizing}
Wonil Choi, Myoungsoo Jung, Mahmut Kandemir, and Chita Das.
\newblock {Parallelizing Garbage Collection with I/O to Improve Flash Resource
  Utilization}.
\newblock In {\em HPDC}, 2018.

\bibitem{chang2007wear}
Li-Pin Chang.
\newblock {On Efficient Wear Leveling for Large-Scale Flash-Memory Storage
  Systems}.
\newblock In {\em SAC}, 2007.

\bibitem{yang2014garbage}
Ming-Chang Yang, Yu-Ming Chang, Che-Wei Tsao, Po-Chun Huang, Yuan-Hao Chang,
  and Tei-Wei Kuo.
\newblock {Garbage Collection and Wear Leveling for Flash Memory: Past and
  Future}.
\newblock In {\em SMARTCOMP}, 2014.

\bibitem{daisyplus_openssd}
{CRZ Technology}.
\newblock {DaisyPlus OpenSSD Platform}.
\newblock \url{https://github.com/CRZ-Technology/OpenSSD-OpenChannelSSD}, 2022.
\newblock Accessed 2026.

\bibitem{tavakkol2018mqsim}
Arash Tavakkol, Juan G{\'o}mez-Luna, Mohammad Sadrosadati, Saugata Ghose, and
  Onur Mutlu.
\newblock {MQSim}: A framework for enabling realistic studies of modern
  multi-queue {SSD} devices.
\newblock In {\em FAST}, 2018.

\bibitem{zheng2013power}
Mai Zheng, Joseph Tucek, Feng Qin, and Mark Lillibridge.
\newblock {Understanding the Robustness of SSDs under Power Fault}.
\newblock In {\em FAST}, 2013.

\bibitem{li2018femu}
Huaicheng Li, Mingzhe Hao, Michael~Hao Tong, Swaminathan Sundararaman, Matias
  Bjorling, and Haryadi~S. Gunawi.
\newblock The {CASE} of {FEMU}: Cheap, accurate, scalable and extensible flash
  emulator.
\newblock In {\em FAST}, 2018.

\bibitem{kim2023nvmevirt}
Sang-Hoon Kim, Jaehoon Shim, Euidong Lee, Seongyeop Jeong, Ilkueon Kang, and
  Jin-Soo Kim.
\newblock {NVMeVirt: A Versatile Software-Defined Virtual {NVMe} Device}.
\newblock In {\em FAST}, 2023.

\bibitem{kwak2018cosmos}
Jaewook Kwak, Sangjin Lee, Kibin Park, Jinwoo Jeong, and Yong~Ho Song.
\newblock {Cosmos+ OpenSSD: Rapid Prototype for Flash Storage Systems}.
\newblock {\em TOS}, 2020.

\bibitem{chang2020determinizing}
Yun-Sheng Chang, Yao Hsiao, Tzu-Chi Lin, Che-Wei Tsao, Chun-Feng Wu, Yuan-Hao
  Chang, Hsiang-Shang Ko, and Yu-Fang Chen.
\newblock {Determinizing Crash Behavior with a Verified Snapshot-Consistent
  Flash Translation Layer}.
\newblock In {\em OSDI}, 2020.

\bibitem{bodenmuller2021flashix}
Stefan Bodenm{\"u}ller, Gerhard Schellhorn, Martin Bitterlich, and Wolfgang
  Reif.
\newblock {Flashix: Modular Verification of A Concurrent and Crash-Safe Flash
  File System}.
\newblock In {\em Logic, Computation and Rigorous Methods}, 2021.

\bibitem{qiao2019formal}
Lei Qiao, Shaofeng Li, Hua Yang, and Mengfei Yang.
\newblock {A Formal Modeling and Verification Framework for Flash Translation
  Layer Algorithms}.
\newblock In {\em SETTA}, 2019.

\bibitem{kim2002space}
Jesung Kim, Jong~Min Kim, Sam~H. Noh, Sang~Lyul Min, and Yookun Cho.
\newblock {A Space-Efficient Flash Translation Layer for CompactFlash Systems}.
\newblock {\em IEEE-TCE}, 2002.

\bibitem{coq}
{The Rocq Development Team}.
\newblock {\em The {Coq} Proof Assistant Reference Manual, Version 8.20}.
\newblock Inria, 2024.

\bibitem{chen2011caftl}
Feng Chen, Tian Luo, and Xiaodong Zhang.
\newblock {CAFTL: A Content-Aware Flash Translation Layer Enhancing the
  Lifespan of Flash Memory based Solid State Drives}.
\newblock In {\em FAST}, 2011.

\bibitem{nadig2026harmonia}
Rakesh Nadig, Vamanan Arulchelvan, Rahul Bera, Taha Shahroodi, Gagandeep Singh,
  Andreas~Kosmas Kakolyris, {\.I}smail~Emir Y{\"u}ksel, Mohammad Sadrosadati,
  Jisung Park, and Onur Mutlu.
\newblock {Harmonia: Enhancing Data Placement and Migration in Hybrid Storage
  Systems via Multi-Agent Reinforcement Learning}.
\newblock In {\em ICS}, 2026.

\bibitem{eshghi2012ssd}
Kam Eshghi and Rino Micheloni.
\newblock {SSD Architecture and PCI Express Interface}.
\newblock In {\em {Inside Solid State Drives (SSDs)}}. Springer, 2013.

\bibitem{snai_sdc2021}
Peter Onufryk.
\newblock {NVMe® 2.0 Specifications: The Next Generation of NVMe Technology}.
\newblock In {\em SNIA-SDC\, 2021}, 2021.

\bibitem{cortex_r4}
{ARM Holdings}.
\newblock {Cortex-R4}.
\newblock \url{https://developer.arm.com/Processors/Cortex-R4}, 2011.

\bibitem{park2006high}
Chanik Park, Prakash Talawar, Daeski Won, Myung~Jin Jung, Jung~Been Im, Suksan
  Kim, and Youngjoon Choi.
\newblock {A High Performance Controller for NAND Flash-based Solid State Disk
  (NSSD)}.
\newblock In {\em NVSMW}, 2006.

\bibitem{bang2011memory}
Kwanhu Bang, Sang-Hoon Park, Minje Jun, and Eui-Young Chung.
\newblock {A Memory Hierarchy-Aware Metadata Management Technique for Solid
  State Disks}.
\newblock In {\em MWSCAS}, 2011.

\bibitem{cai-date-2012}
Yu~Cai, Erich~F. Haratsch, Onur Mutlu, and Ken Mai.
\newblock {Error Patterns in MLC NAND Flash Memory: Measurement,
  Characterization, and Analysis}.
\newblock In {\em DATE}, 2012.

\bibitem{cai-iccd-2012}
Yu~Cai, Gulay Yalcin, Onur Mutlu, Erich~F. Haratsch, Adrian Cristal, Osman~S.
  Unsal, and Ken Mai.
\newblock {Flash Correct-and-Refresh: Retention-Aware Error Management for
  Increased Flash Memory Lifetime}.
\newblock In {\em ICCD}, 2012.

\bibitem{cai-inteltechj-2013}
Yu~Cai, Gulay Yalcin, Onur Mutlu, Erich~F. Haratsch, Adrian Cristal, Osman~S.
  Unsal, and Ken Mai.
\newblock {Error Analysis and Retention-Aware Error Management for {NAND} Flash
  Memory}.
\newblock {\em Intel Tech. J.}, 2013.

\bibitem{cai-sigmetrics-2014}
Yu~Cai, Gulay Yalcin, Onur Mutlu, Erich~F. Haratsch, Osman~S. Unsal, Adri\'{a}n
  Cristal, and Ken Mai.
\newblock {Neighbor-cell Assisted Error Correction for MLC NAND Flash
  Memories}.
\newblock In {\em SIGMETRICS}, 2014.

\bibitem{cai-dsn-2015}
Yu~Cai, Yixin Luo, Saugata Ghose, and Onur Mutlu.
\newblock {Read Disturb Errors in {MLC NAND} Flash Memory: Characterization,
  Mitigation, and Recovery}.
\newblock In {\em DSN}, 2015.

\bibitem{cai-hpca-2017}
Yu~Cai, Saugata Ghose, Yixin Luo, Ken Mai, Onur Mutlu, and Erich~F. Haratsch.
\newblock {Vulnerabilities in MLC NAND Flash Memory Programming: Experimental
  Analysis, Exploits, and Mitigation Techniques}.
\newblock In {\em HPCA}, 2017.

\bibitem{luo-sigmetrics-2018}
Yixin Luo, Saugata Ghose, Yu~Cai, Erich~F. Haratsch, and Onur Mutlu.
\newblock {Improving {3D NAND} Flash Memory Lifetime by Tolerating Early
  Retention Loss and Process Variation}.
\newblock In {\em SIGMETRICS}, 2018.

\bibitem{cai-insidessd-2018}
Yu~Cai, Saugata Ghose, Erich~F. Haratsch, Yixin Luo, and Onur Mutlu.
\newblock {Reliability Issues in Flash-memory-based Solid-state Drives:
  Experimental Analysis, Mitigation, Recovery}.
\newblock In {\em Inside Solid State Drives}, 2018.

\bibitem{cai2013program}
Yu~Cai, Onur Mutlu, Erich~F. Haratsch, and Ken Mai.
\newblock {Program Interference in MLC NAND Flash Memory: Characterization,
  Modeling, and Mitigation}.
\newblock In {\em ICCD}, 2013.

\bibitem{cai2015data}
Yu~Cai, Yixin Luo, Erich~F. Haratsch, Ken Mai, and Onur Mutlu.
\newblock {Data Retention in MLC NAND Flash Memory: Characterization,
  Optimization, and Recovery}.
\newblock In {\em HPCA}, 2015.

\bibitem{cai2017error}
Yu~Cai, Saugata Ghose, Erich~F. Haratsch, Yixin Luo, and Onur Mutlu.
\newblock {Error Characterization, Mitigation, and Recovery in
  Flash-Memory-based Solid-State Drives}.
\newblock {\em Proceedings of the IEEE}, 2017.

\bibitem{luo2015warm}
Yixin Luo, Yu~Cai, Saugata Ghose, Jongmoo Choi, and Onur Mutlu.
\newblock {WARM: Improving NAND Flash Memory Lifetime with Write-hotness Aware
  Retention Management}.
\newblock In {\em MSST}, 2015.

\bibitem{luo2016enabling}
Yixin Luo, Saugata Ghose, Yu~Cai, Erich~F. Haratsch, and Onur Mutlu.
\newblock {Enabling Accurate and Practical Online Flash Channel Modeling for
  Modern MLC NAND Flash Memory}.
\newblock {\em JSAC}, 2016.

\bibitem{luo2018heatwatch}
Yixin Luo, Saugata Ghose, Yu~Cai, Erich~F. Haratsch, and Onur Mutlu.
\newblock {HeatWatch: Improving 3D NAND Flash Memory Device Reliability by
  Exploiting Self-Recovery and Temperature Awareness}.
\newblock In {\em HPCA}, 2018.

\bibitem{park-asplos-2021}
Jisung Park, Myungsuk Kim, Myoungjun Chun, Lois Orosa, Jihong Kim, and Onur
  Mutlu.
\newblock {Reducing Solid-State Drive Read Latency by Optimizing Read-Retry}.
\newblock In {\em ASPLOS}, 2021.

\bibitem{onfi}
{ONFI Workgroup}.
\newblock {Open NAND Flash Interface Specification, Revision 5.1}.
\newblock \url{https://onfi.org}, 2022.

\bibitem{grupp2009characterizing}
Laura~M. Grupp, Adrian~M. Caulfield, Joel Coburn, Steven Swanson, Eitan
  Yaakobi, Paul~H. Siegel, and Jack~K. Wolf.
\newblock {Characterizing Flash Memory: Anomalies, Observations, and
  Applications}.
\newblock In {\em MICRO}, 2009.

\bibitem{dirik2009performance}
Cagdas Dirik and Bruce Jacob.
\newblock {The Performance of PC Solid-State Disks (SSDs) as a Function of
  Bandwidth, Concurrency, Device Architecture, and System Organization}.
\newblock In {\em ISCA}, 2009.

\bibitem{pekny2022isscc}
Ted Pekny, Luyen Vu, Jeff Tsai, Dheeraj Srinivasan, Erwin Yu, Jonathan
  Pabustan, Joe Xu, Srinivas Deshmukh, Kim-Fung Chan, Michael Piccardi, et~al.
\newblock {A 1-Tb Density 4b/Cell 3D-NAND Flash on 176-Tier Technology with
  4-Independent Planes for Read using CMOS-Under-the-Array}.
\newblock In {\em ISSCC}, 2022.

\bibitem{cho20221}
Wanik Cho, Jongseok Jung, Jongwoo Kim, Junghoon Ham, Sangkyu Lee, Yujong Noh,
  Dauni Kim, Wanseob Lee, Kayoung Cho, Kwanho Kim, et~al.
\newblock {A 1-Tb, 4b/Cell, 176-Stacked-WL 3D-NAND Flash Memory with Improved
  Read Latency and a 14.8 Gb/mm$^2$ Density}.
\newblock In {\em ISSCC}, 2022.

\bibitem{kang201913}
Dongku Kang, Minsu Kim, Su~Chang Jeon, Wontaeck Jung, Jooyong Park, Gyosoo
  Choo, Dong-kyo Shim, Anil Kavala, Seung-Bum Kim, Kyung-Min Kang, et~al.
\newblock {A 512Gb 3-bit/Cell 3D 6th-Generation V-NAND Flash Memory with 82MB/s
  Write Throughput and 1.2Gb/s Interface}.
\newblock In {\em ISSCC}, 2019.

\bibitem{maejima2018512gb}
Hiroshi Maejima, Kazushige Kanda, Susumu Fujimura, Teruo Takagiwa, Susumu
  Ozawa, Jumpei Sato, Yoshihiko Shindo, Manabu Sato, Naoaki Kanagawa, Junji
  Musha, et~al.
\newblock {A 512Gb 3b/Cell 3D Flash Memory on a 96-Word-Line-Layer Technology}.
\newblock In {\em ISSCC}, 2018.

\bibitem{nvme_nvm_command_set}
{NVM Express}.
\newblock {NVM Express\textregistered{} NVM Command Set Specification (Revision
  1.2)}.
\newblock \url{https://nvmexpress.org/specifications/}, 2025.
\newblock Section 5.3, End-to-end Data Protection.

\bibitem{gal2005algorithms}
Eran Gal and Sivan Toledo.
\newblock {Algorithms and Data Structures for Flash Memories}.
\newblock {\em ACM CSUR}, 2005.

\bibitem{park2009design}
Sang-Hoon Park, Seung-Hwan Ha, Kwanhu Bang, and Eui-Young Chung.
\newblock {Design and Analysis of Flash Translation Layers for Multi-Channel
  {NAND} Flash-Based Storage Devices}.
\newblock {\em IEEE-TCE}, 2009.

\bibitem{chung2009survey}
Tae-Sun Chung, Dong-Joo Park, Sangwon Park, Dong-Ho Lee, Sang-Won Lee, and
  Ha-Joo Song.
\newblock {A Survey of Flash Translation Layer}.
\newblock {\em JSA}, 2009.

\bibitem{goodson2010design}
Garth Goodson and Rahul Iyer.
\newblock {Design Tradeoffs in a Flash Translation Layer}.
\newblock In {\em WEST}, 2010.

\bibitem{gal2005mapping}
Eran Gal and Sivan Toledo.
\newblock {Mapping Structures for Flash Memories: Techniques and Open
  Problems}.
\newblock In {\em SwSTE}, 2005.

\bibitem{ma2014survey}
Dongzhe Ma, Jianhua Feng, and Guoliang Li.
\newblock {A Survey of Address Translation Technologies for Flash Memories}.
\newblock {\em CSUR}, 2014.

\bibitem{sun2023leaftl}
Jinghan Sun, Shaobo Li, Yunxin Sun, Chao Sun, Dejan Vucinic, and Jian Huang.
\newblock {LeaFTL: A Learning-Based Flash Translation Layer For Solid-State
  Drives}.
\newblock In {\em ASPLOS}, 2023.

\bibitem{picoli2019ox}
Ivan~Luiz Picoli.
\newblock {\em {OX}: Deconstructing the {FTL} for Computational Storage}.
\newblock {Ph.D.} thesis, IT University of Copenhagen, 2019.

\bibitem{lee2007log}
Sang-Won Lee, Dong-Joo Park, Tae-Sun Chung, Dong-Ho Lee, Sangwon Park, and
  Ha-Joo Song.
\newblock {A Log Buffer-Based Flash Translation Layer Using Fully-Associative
  Sector Translation}.
\newblock {\em TECS}, 2007.

\bibitem{ma_lazyftl_2011}
Dongzhe Ma, Jianhua Feng, and Guoliang Li.
\newblock {LazyFTL: A Page-level Flash Translation Layer Optimized for NAND
  Flash Memory}.
\newblock In {\em SIGMOD}, 2011.

\bibitem{zhou2015efficient}
You Zhou, Fei Wu, Ping Huang, Xubin He, Changsheng Xie, and Jian Zhou.
\newblock {An Efficient Page-level FTL to Optimize Address Translation in Flash
  Memory}.
\newblock In {\em EuroSys}, 2015.

\bibitem{lim2010faster}
Sang-Phil Lim, Sang-Won Lee, and Bongki Moon.
\newblock {FASTer FTL for Enterprise-Class Flash Memory SSDs}.
\newblock In {\em SNAPI}, 2010.

\bibitem{shin2009ftl}
Ji-Yong Shin, Zeng-Lin Xia, Ning-Yi Xu, Rui Gao, Xiong-Fei Cai, Seungryoul
  Maeng, and Feng-Hsiung Hsu.
\newblock {FTL Design Exploration in Reconfigurable High-Performance SSD for
  Server Applications}.
\newblock In {\em ICS}, 2009.

\bibitem{wang2024learnedftl}
Shengzhe Wang, Zihang Lin, Suzhen Wu, Hong Jiang, Jie Zhang, and Bo~Mao.
\newblock {LearnedFTL: A Learning-Based Page-Level FTL for Reducing Double
  Reads in Flash-Based SSDs}.
\newblock In {\em HPCA}, 2024.

\bibitem{park2016storage}
Dongchul Park, Jianguo Wang, and Yang-Suk Kee.
\newblock {In-storage Computing for Hadoop MapReduce Framework: Challenges and
  Possibilities}.
\newblock {\em ToC}, 2016.

\bibitem{mailthody-micro-2019}
Vikram~Sharma Mailthody, Zaid Qureshi, Weixin Liang, Ziyan Feng, Simon
  Garcia~de Gonzalo, Youjie Li, Hubertus Franke, Jinjun Xiong, Jian Huang, and
  Wen-mei Hwu.
\newblock {DeepStore: In-Storage Acceleration for Intelligent Queries}.
\newblock In {\em MICRO}, 2019.

\bibitem{kang2021iceclave}
Luyi Kang, Yuqi Xue, Weiwei Jia, Xiaohao Wang, Jongryool Kim, Changhwan Youn,
  Myeong~Joon Kang, Hyung~Jin Lim, Bruce Jacob, and Jian Huang.
\newblock {Iceclave: A Trusted Execution Environment For In-Storage Computing}.
\newblock In {\em MICRO}, 2021.

\bibitem{chen2011essential}
Feng Chen, Rubao Lee, and Xiaodong Zhang.
\newblock {Essential Roles of Exploiting Internal Parallelism of Flash
  Memory-Based Solid State Drives in High-Speed Data Processing}.
\newblock In {\em HPCA}, 2011.

\bibitem{hu2011performance}
Yang Hu, Hong Jiang, Dan Feng, Lei Tian, Hao Luo, and Shuping Zhang.
\newblock {Performance Impact and Interplay of SSD Parallelism Through Advanced
  Commands, Allocation Strategy and Data Granularity}.
\newblock In {\em ICS}, 2011.

\bibitem{hu2012exploring}
Yang Hu, Hong Jiang, Dan Feng, Lei Tian, Hao Luo, and Chao Ren.
\newblock {Exploring and Exploiting the Multilevel Parallelism Inside SSDs For
  Improved Performance and Endurance}.
\newblock {\em ToC}, 2013.

\bibitem{kim2022networked}
Jiho Kim, Seokwon Kang, Yongjun Park, and John Kim.
\newblock {Networked SSD: Flash Memory Interconnection Network for
  High-Bandwidth SSD}.
\newblock In {\em MICRO}, 2022.

\bibitem{lightnvm}
Matias Bj{\o}rling, Javier Gonz{\'a}lez, and Philippe Bonnet.
\newblock {LightNVM: The Linux Open-Channel {SSD} Subsystem}.
\newblock In {\em FAST}, 2017.

\bibitem{picoli2020open}
Ivan~Luiz Picoli, Niclas Hedam, Philippe Bonnet, and Pinar T{\"o}z{\"u}n.
\newblock {Open-Channel SSD (What Is It Good For)}.
\newblock In {\em CIDR}, 2020.

\bibitem{znssds}
{NVM Express}.
\newblock {NVMe Zoned Namespace {SSDs} \& the Zoned Storage {Linux} Software
  Ecosystem}, 2020.

\bibitem{jin2017kaml}
Yanqin Jin, Hung-Wei Tseng, Yannis Papakonstantinou, and Steven Swanson.
\newblock {KAML: A Flexible, High-Performance Key-Value SSD}.
\newblock In {\em HPCA}, 2017.

\bibitem{jang2024smart}
Hongsun Jang, Jaeyong Song, Jaewon Jung, Jaeyoung Park, Youngsok Kim, and Jinho
  Lee.
\newblock {Smart-Infinity: Fast Large Language Model Training Using
  Near-Storage Processing on A Real System}.
\newblock In {\em HPCA}, 2024.

\bibitem{das2013cpu}
Sudipto Das, Vivek~R. Narasayya, Feng Li, and Manoj Syamala.
\newblock {CPU Sharing Techniques for Performance Isolation in Multi-Tenant
  Relational Database-as-a-Service}.
\newblock {\em PVLDB}, 2013.

\bibitem{wang2018oc}
Haitao Wang, Zhanhuai Li, Xiao Zhang, Xiaonan Zhao, Xingsheng Zhao, Weijun Li,
  and Song Jiang.
\newblock {OC-Cache: An Open-channel SSD Based Cache for Multi-Tenant Systems}.
\newblock In {\em IPCCC}, 2018.

\bibitem{luo2013s}
Tian Luo, Siyuan Ma, Rubao Lee, Xiaodong Zhang, Deng Liu, and Li~Zhou.
\newblock {S-Cave: Effective SSD Caching to Improve Virtual Machine Storage
  Performance}.
\newblock In {\em PACT}, 2013.

\bibitem{liu2021self}
Renping Liu, Duo Liu, Xianzhang Chen, Yujuan Tan, Runyu Zhang, and Liang Liang.
\newblock {Self-Adapting Channel Allocation for Multiple Tenants Sharing SSD
  Devices}.
\newblock {\em TCAD}, 2022.

\bibitem{tcg_key_per_io}
{Trusted Computing Group}.
\newblock {TCG Storage Security Subsystem Class (SSC): Key Per I/O, Version
  1.00, Revision 1.41}.
\newblock
  \url{https://trustedcomputinggroup.org/resource/tcg-storage-security-subsystem-class-ssc-key-per-i-o/},
  2023.

\bibitem{kim2018utilitarian}
Bryan~S. Kim.
\newblock {Utilitarian Performance Isolation in Shared SSDs}.
\newblock In {\em HotStorage 18}, 2018.

\bibitem{zaddach2013backdoor}
Jonas Zaddach, Anil Kurmus, Davide Balzarotti, Erik-Oliver Blass, Aur{\'e}lien
  Francillon, Travis Goodspeed, Moitrayee Gupta, and Ioannis Koltsidas.
\newblock {Implementation and Implications of a Stealth Hard-Drive Backdoor}.
\newblock In {\em ACSAC}, 2013.

\bibitem{meijer2019self}
Carlo Meijer and Bernard Van~Gastel.
\newblock {Self-Encrypting Deception: Weaknesses in the Encryption of Solid
  State Drives}.
\newblock In {\em S\&P}, 2019.

\bibitem{wertenbroek2024pandora}
Rick Wertenbroek and Alberto Dassatti.
\newblock {Pandora's Box in Your SSD: The Untold Dangers of NVMe}.
\newblock {\em arXiv preprint arXiv:2411.00439}, 2024.

\bibitem{kim2020evanesco}
Myungsuk Kim, Jisung Park, Genhee Cho, Yoona Kim, Lois Orosa, Onur Mutlu, and
  Jihong Kim.
\newblock {Evanesco: Architectural Support for Efficient Data Sanitization in
  Modern Flash-Based Storage Systems}.
\newblock In {\em ASPLOS}, 2020.

\bibitem{juffinger2025secret}
Jonas Juffinger, Fabian Rauscher, Giuseppe~La Manna, and Daniel Gruss.
\newblock {Secret Spilling Drive: Leaking User Behavior through SSD
  Contention}.
\newblock In {\em NDSS}, 2025.

\bibitem{juffinger2025hmb}
Jonas Juffinger, Hannes Weissteiner, Thomas Steinbauer, and Daniel Gruss.
\newblock {The HMB Timing Side Channel: Exploiting the SSD's Host Memory
  Buffer}.
\newblock In {\em DIMVA}, 2025.

\bibitem{fips180-4}
{National Institute of Standards and Technology}.
\newblock {Secure Hash Standard (SHS)}.
\newblock FIPS PUB 180-4, 2015.

\bibitem{amani2016cogent}
Sidney Amani, Alex Hixon, Zilin Chen, Christine Rizkallah, Peter Chubb, Liam
  O'Connor, Joel Beeren, Yutaka Nagashima, Japheth Lim, Thomas Sewell, Joseph
  Tuong, Gabriele Keller, Toby Murray, Gerwin Klein, and Gernot Heiser.
\newblock {Cogent: Verifying High-Assurance File System Implementations}.
\newblock In {\em ASPLOS}, 2016.

\bibitem{tripathy2023formal}
Shivani Tripathy, Debiprasanna Sahoo, Manoranjan Satpathy, and Madhu Mutyam.
\newblock {Formal Modeling and Verification of Security Properties of a
  Ransomware-Resistant SSD}.
\newblock {\em IEEE TCAD}, 2023.

\bibitem{huet1997coq}
G{\'e}rard Huet, Gilles Kahn, and Christine Paulin-Mohring.
\newblock {The Coq Proof Assistant: A Tutorial}.
\newblock Technical report, INRIA, 1997.

\bibitem{chlipala2013certified}
Adam Chlipala.
\newblock {\em {Certified Programming with Dependent Types: A Pragmatic
  Introduction to the Coq Proof Assistant}}.
\newblock MIT Press, 2013.

\bibitem{schierl2009ubifs}
Andreas Schierl, Gerhard Schellhorn, Dominik Haneberg, and Wolfgang Reif.
\newblock {Abstract Specification of the UBIFS File System for Flash Memory}.
\newblock In {\em FM}, 2009.

\bibitem{schellhorn2014verified}
Gerhard Schellhorn, Gidon Ernst, J{\"o}rg Pf{\"a}hler, Dominik Haneberg, and
  Wolfgang Reif.
\newblock {Development of a Verified Flash File System}.
\newblock In {\em ABZ}, 2014.

\bibitem{ernst2015inside}
Gidon Ernst, J{\"o}rg Pf{\"a}hler, Gerhard Schellhorn, and Wolfgang Reif.
\newblock {Inside a Verified Flash File System: Transactions and Garbage
  Collection}.
\newblock In {\em VSTTE}, 2015.

\bibitem{reif1995kiv}
Wolfgang Reif, Gerhard Schellhorn, and Kurt Stenzel.
\newblock {Interactive Correctness Proofs for Software Modules Using KIV}.
\newblock In {\em COMPASS}, 1995.

\bibitem{joshi2007mini}
Rajeev Joshi and Gerard~J. Holzmann.
\newblock {A Mini Challenge: Build a Verifiable Filesystem}.
\newblock {\em FAC}, 19(2), 2007.

\bibitem{butterfield2009mechanising}
Andrew Butterfield, Leo Freitas, and Jim Woodcock.
\newblock {Mechanising a Formal Model of Flash Memory}.
\newblock {\em SCP}, 74(4), 2009.

\bibitem{kang2008alloy}
Eunsuk Kang and Daniel Jackson.
\newblock {Formal Modeling and Analysis of a Flash Filesystem in Alloy}.
\newblock In {\em ABZ}, 2008.

\bibitem{damchoom2009eventb}
Kriangsak Damchoom and Michael~J. Butler.
\newblock {Applying Event and Machine Decomposition to a Flash-Based Filestore
  in Event-B}.
\newblock In {\em SBMF}, 2009.

\bibitem{taverne2009raffs}
Paul Taverne and Cornelis Pronk.
\newblock {RAFFS: Model Checking a Robust Abstract Flash File Store}.
\newblock In {\em ICFEM}, 2009.

\bibitem{pfahler2013erase}
J{\"o}rg Pf{\"a}hler, Gidon Ernst, Gerhard Schellhorn, Dominik Haneberg, and
  Wolfgang Reif.
\newblock {Formal Specification of an Erase Block Management Layer for Flash
  Memory}.
\newblock In {\em HVC}, 2013.

\bibitem{kim2008flashdriver}
Moonzoo Kim, Yunja Choi, Yunho Kim, and Hotae Kim.
\newblock {Formal Verification of a Flash Memory Device Driver: An Experience
  Report}.
\newblock In {\em SPIN}, 2008.

\bibitem{tripathy2019nand}
Shivani Tripathy, Debiprasanna Sahoo, Manoranjan Satpathy, and Srinivas
  Pinisetty.
\newblock {Formal Modeling and Verification of NAND Flash Memory Supporting
  Advanced Operations}.
\newblock In {\em ICCD}, 2019.

\bibitem{prabhakaran2008txflash}
Vijayan Prabhakaran, Thomas~L. Rodeheffer, and Lidong Zhou.
\newblock {Transactional Flash}.
\newblock In {\em OSDI}, 2008.

\bibitem{chajed2019argosy}
Tej Chajed, Joseph Tassarotti, M.~Frans Kaashoek, and Nickolai Zeldovich.
\newblock {Argosy: Verifying Layered Storage Systems with Recovery Refinement}.
\newblock In {\em PLDI}, 2019.

\bibitem{chajed2019perennial}
Tej Chajed, Joseph Tassarotti, M.~Frans Kaashoek, and Nickolai Zeldovich.
\newblock {Verifying Concurrent, Crash-Safe Systems with Perennial}.
\newblock In {\em SOSP}, 2019.

\bibitem{chajed2021gojournal}
Tej Chajed, Joseph Tassarotti, Mark Theng, Ralf Jung, M.~Frans Kaashoek, and
  Nickolai Zeldovich.
\newblock {GoJournal: A Verified, Concurrent, Crash-Safe Journaling System}.
\newblock In {\em OSDI}, 2021.

\bibitem{hance2020storage}
Travis Hance, Andrea Lattuada, Chris Hawblitzel, Jon Howell, Rob Johnson, and
  Bryan Parno.
\newblock {Storage Systems are Distributed Systems (So Verify Them That Way!)}.
\newblock In {\em OSDI}, 2020.

\bibitem{chen2015fscq}
Haogang Chen, Daniel Ziegler, Tej Chajed, Adam Chlipala, M.~Frans Kaashoek, and
  Nickolai Zeldovich.
\newblock Using {Crash} {Hoare} logic for certifying the {FSCQ} file system.
\newblock In {\em SOSP}, 2015.

\bibitem{chen2017dfscq}
Haogang Chen, Tej Chajed, Alex Konradi, Stephanie Wang, At\'{a}lay \'{I}leri,
  Adam Chlipala, M.~Frans Kaashoek, and Nickolai Zeldovich.
\newblock {Verifying a High-Performance Crash-Safe File System Using a Tree
  Specification}.
\newblock In {\em SOSP}, 2017.

\bibitem{sigurbjarnarson2016yggdrasil}
Helgi Sigurbjarnarson, James Bornholt, Emina Torlak, and Xi~Wang.
\newblock {Push-Button Verification of File Systems via Crash Refinement}.
\newblock In {\em OSDI}, 2016.

\bibitem{ileri2018disksec}
Atalay~Mert Ileri, Tej Chajed, Adam Chlipala, M.~Frans Kaashoek, and Nickolai
  Zeldovich.
\newblock {Proving Confidentiality in a File System Using DiskSec}.
\newblock In {\em OSDI}, 2018.

\bibitem{sel4}
Gerwin Klein, Kevin Elphinstone, Gernot Heiser, June Andronick, David Cock,
  Philip Derrin, Dhammika Elkaduwe, Kai Engelhardt, Rafal Kolanski, Michael
  Norrish, Thomas Sewell, Harvey Tuch, and Simon Winwood.
\newblock {seL4}: Formal verification of an {OS} kernel.
\newblock In {\em SOSP}, 2009.

\bibitem{murray2013sel4}
Toby Murray, Daniel Matichuk, Matthew Brassil, Peter Gammie, Timothy Bourke,
  Sean Seefried, Corey Lewis, Xin Gao, and Gerwin Klein.
\newblock {seL4}: From general purpose to a proof of information flow
  enforcement.
\newblock In {\em IEEE S\&P}, 2013.

\bibitem{certikos}
Ronghui Gu, Zhong Shao, Hao Chen, Xiongnan~(Newman) Wu, Jieung Kim, Vilhelm
  Sj{\"o}berg, and David Costanzo.
\newblock {CertiKOS: An Extensible Architecture for Building Certified
  Concurrent OS Kernels}.
\newblock In {\em OSDI}, 2016.

\bibitem{compcert}
Xavier Leroy.
\newblock {Formal Verification of a Realistic Compiler}.
\newblock {\em CACM}, 2009.

\bibitem{goguen1982security}
Joseph~A. Goguen and Jos{\'e} Meseguer.
\newblock {Security Policies and Security Models}.
\newblock In {\em IEEE S\&P}, 1982.

\bibitem{fips197}
{NIST}.
\newblock {Advanced Encryption Standard (AES)}.
\newblock Technical report, National Institute of Standards and Technology,
  2001.

\bibitem{flashvault}
Seock-Hwan Noh, Hoyeon Lee, Junkyum Kim, Junsu Im, Jay~H. Park, Sungjin Lee,
  Sam~H. Noh, Yeseong Kim, and Jaeha Kung.
\newblock {FlashVault: Versatile In-NAND Self-Encryption with Zero Area
  Overhead}.
\newblock {\em arXiv:2508.03866v1}, 2025.

\bibitem{huang2017flashguard}
Jian Huang, Jun Xu, Xinyu Xing, Peng Liu, and Moinuddin~K. Qureshi.
\newblock {FlashGuard: Leveraging Intrinsic Flash Properties to Defend Against
  Encryption Ransomware}.
\newblock In {\em CCS}, 2017.

\bibitem{baek2018ssdinsider}
SungHa Baek, Youngdon Jung, Aziz Mohaisen, Sungjin Lee, and DaeHun Nyang.
\newblock {SSD-Insider: Internal Defense of Solid-State Drive against
  Ransomware with Perfect Data Recovery}.
\newblock In {\em ICDCS}, 2018.

\bibitem{chowdhuryy2023dshield}
Md~Hafizul~Islam Chowdhuryy, Myoungsoo Jung, Fan Yao, and Amro Awad.
\newblock {D-Shield: Enabling Processor-side Encryption and Integrity
  Verification for Secure NVMe Drives}.
\newblock In {\em HPCA}, 2023.

\bibitem{ahn2020sgxssd}
Jinwoo Ahn, Seungjin Lee, Jinhoon Lee, Yungwoo Ko, Donghyun Min, Junghee Lee,
  and Youngjae Kim.
\newblock {SGX-SSD: A Policy-based Versioning SSD with Intel SGX}.
\newblock In {\em HotStorage}, 2020.

\end{thebibliography}

\end{document}